\documentclass[runningheads,11pt]{llncs}

\usepackage{geometry}
\usepackage[activate=true,final,
tracking=true,
kerning=true,
spacing=true,
factor=1100,
stretch=10,
shrink=10]{microtype}
\usepackage{amsmath,mathtools}
\usepackage{amssymb}
\usepackage{enumerate}
\usepackage{esvect}
\usepackage{cite}
\renewcommand{\qed}{\hfill $\blacksquare$}
\usepackage[utf8]{inputenc}
\usepackage[T1]{fontenc}
\usepackage{lmodern}
\usepackage{mathalfa}
\numberwithin{equation}{section}
\usepackage{bm}
\usepackage{spverbatim}
\usepackage{breakcites}
\usepackage{enumitem}
\usepackage[linktocpage=true]{hyperref}
\hypersetup{
	colorlinks,
	citecolor=black,
	filecolor=black,
	linkcolor=black,
	urlcolor=black
}
\usepackage{cleveref}
\usepackage{amsfonts}
\usepackage{tkz-graph}
\usepackage{asymptote}
\usepackage{float}
\usepackage{soul}
\usepackage{xcolor}
\usepackage{stmaryrd}
\usepackage{graphicx}
\usepackage{caption} 
\usepackage{url}
\usepackage{array}
\usepackage{mathrsfs}
\usepackage{multicol,xparse}
\usepackage{changepage}
\usepackage{mdframed} 
\usepackage{nccmath}
\usepackage{multicol,xparse}
\usepackage{upgreek}
\usepackage[sanserif,basic]{complexity}
\newenvironment{myenv}[1][4mm]{\begin{adjustwidth}{#1}{}}{\end{adjustwidth}}
\newcommand\numberthis{\addtocounter{equation}{1}\tag{\theequation}}
\DeclareMathOperator{\var}{var}

\makeatletter
\newcommand{\customlabel}[2]{\protected@write @auxout {}{\string \newlabel {#1}{{#2}{\thepage}{#2}{#1}{}}}\hypertarget{#1}{}}
	
\makeatother

\makeatletter
\def\@fnsymbol#1{\ensuremath{\ifcase#1\or \dagger\or \ddagger\or
		\mathsection\or \mathparagraph\or \|\or **\or \dagger\dagger
		\or \ddagger\ddagger \else\@ctrerr\fi}}
\makeatother

\newcommand{\stargraph}[2]{\begin{tikzpicture}
		\node[circle,fill=black] at (360:0mm) (center) {};
		\foreach \n in {1,...,#1}{
			\node[circle,fill=black] at ({\n*360/#1}:#2cm) (n\n) {};
			\draw (center)--(n\n);
		}
\end{tikzpicture}}

\title{Star-specific Key-homomorphic PRFs\\ from Learning with Linear Regression\thanks{This is the preprint of a paper published in IEEE Access, vol. 11, 2023 \cite{Vipin[23]}.}}
\author{Vipin Singh Sehrawat\inst{1}\thanks{Work done while the author was affiliated to Seagate Technology, Fremont, CA, USA.} \and Foo Yee Yeo\inst{2} \and Dmitriy Vassilyev\inst{3}}
\institute{Circle Internet Financial \\ \email{\{\tt vipin.sehrawat.cs@gmail.com\}} \and 
	Seagate Technology, Singapore \and 
	Seagate Technology, Longmont, CO, USA}

\titlerunning{Star-specific Key-homomorphic PRFs from Learning with Linear Regression}
\authorrunning{V. S. Sehrawat et al.}		

\begin{document}
	\maketitle			
	\begin{abstract}\normalsize 
		We introduce a novel method to derandomize the learning with errors (LWE) problem by generating deterministic yet sufficiently independent LWE instances that are constructed by using linear regression models, which are generated via (wireless) communication errors. We also introduce star-specific key-homomorphic (SSKH) pseudorandom functions (PRFs), which are defined by the respective sets of parties that construct them. We use our derandomized variant of LWE to construct a SSKH PRF family. The sets of parties constructing SSKH PRFs are arranged as star graphs with possibly shared vertices, i.e., the pairs of sets may have non-empty intersections. We reduce the security of our SSKH PRF family to the hardness of LWE. To establish the maximum number of SSKH PRFs that can be constructed --- by a set of parties --- in the presence of passive/active and external/internal adversaries, we prove several bounds on the size of maximally cover-free at most $t$-intersecting $k$-uniform family of sets $\mathcal{H}$, where the three properties are defined as: (i) $k$-uniform: $\forall A \in \mathcal{H}: |A| = k$, (ii) at most $t$-intersecting: $\forall A, B \in \mathcal{H}, B \neq A: |A \cap B| \leq t$, (iii) maximally cover-free: $\forall A \in \mathcal{H}: A \not\subseteq \bigcup\limits_{\substack{B \in \mathcal{H} \\ B \neq A}} B$. For the same purpose, we define and compute the mutual information between different linear regression hypotheses that are generated from overlapping training datasets. 
\end{abstract}

\tableofcontents
\newpage

\section{Introduction}
\subsubsection{Derandomized LWE.} The learning with errors (LWE) problem~\cite{Reg[05]} is at the basis of multiple cryptographic constructions~\cite{Peikert[16],Hamid[19]}. Informally, LWE requires solving a system of `approximate' linear modular equations. Given positive integers $w$ and $q \geq 2$, an LWE sample is defined as: $(\textbf{a}, b = \langle \textbf{a}, \textbf{s} \rangle + e \bmod q)$, where $\textbf{s} \in \mathbb{Z}^w_q$ and $\textbf{a} \xleftarrow{\; \$ \;} \mathbb{Z}^w_q$. The error term $e$ is sampled randomly, typically from a normal distribution with standard deviation $\alpha q$ where $\alpha = 1/\poly(w)$, followed by which it is rounded to the nearest integer and reduced modulo $q$. Banerjee et al.~\cite{Ban[12]} introduced a derandomized variant of LWE, called learning with rounding (LWR), wherein instead of adding a random small error, a \textit{deterministically} rounded version of the sample is announced. Specifically, for some positive integer $p < q$, the elements of $\mathbb{Z}_q$ are divided into $p$ contiguous intervals containing (roughly) $q/p$ elements each. The rounding function, defined as: $\lfloor \cdot \rceil_p: \mathbb{Z}_q \rightarrow \mathbb{Z}_p$, maps the given input $x \in \mathbb{Z}_q$ into the index of the interval that $x$ belongs to. An LWR instance is generated as: $(\textbf{a}, \lfloor \langle \textbf{a}, \textbf{s} \rangle \rceil_p)$ for vectors $\textbf{s} \in \mathbb{Z}_q^w$ and $\textbf{a} \xleftarrow{\; \$ \;} \mathbb{Z}_q^w$. For certain range of parameters, Banerjee et al. proved the hardness of LWR under the LWE assumption. In this work, we propose a new derandomized variant of LWE, called learning with linear regression (LWLR). We reduce the hardness of LWLR to that of LWE for certain choices of parameters.

\subsubsection{Physical Layer Communications and Shared Secret Extraction.}\label{Wave} In the OSI (Open Systems Interconnection) model\footnote{see \cite{DavAnd[10]} Section 1.4.1 for an introduction to the OSI model}, physical layer consists of the fundamental hardware transmission technologies. It provides electrical, mechanical, and procedural interface to the transmission medium for transmitting raw bits over a communication channel. Physical layer communication between parties has certain inherent characteristics that make it an attractive source of renewable, shared secrecy. Multiple methods to extract secret bits from channel measurements have been explored (e.g., \cite{Prem[13],Xiao[08],Zhang[08],Zeng[15],Kepe[15],Jiang[13],Ye[10],Ye[06],Ye[07],JanaSri[09],JieYan[12],PremSne[14],XiQi[16],HaiKui[11],XuKui[12],SunLi[18],DingLi[19],ZhaGre[20],XuHu[19],MarHan[18],JhaWen[18],MarFaf[17]}). See~\cite{Sheh[15],Poor[17]} for an overview of some of the notable results in the area. \textit{Channel reciprocity} simply means that the signal distortion (attenuation, delay, phase shift, and fading) is identical in both directions of a link. Hence, it follows from channel reciprocity that the two receive-nodes of a channel observe identical channel characteristic and state information. Secrecy of this information follows directly from the \textit{spatial decorrelation} property, which states that in rich scattering environments, the receivers located at least half a wavelength away experience uncorrelated channels. Therefore, an eavesdropper separated by at least half a wavelength from the two communicating nodes experiences an entirely different channel, and hence cannot make accurate measurements. In typical cellular or wireles LAN frequencies, this distance --- of half a wavelength --- is less than half a foot, which is an acceptable assumption for separation from an eavesdropping adversary \cite{Ye[06]}. Both channel reciprocity and spatial decorrelation have been examined extensively and demonstrated to hold in practice \cite{MarPao[14],ZenZimm[15],ChristHen[16],MattAgg[16],WoodsTru[16],MadiDong[09],SanCla[10]}. For further details on these two properties of communication channels, we refer the interested reader to \cite{ZanMar[12]}. In this work, we use these two properties to securely generate sufficiently independent yet deterministic errors to derandomize LWE. 

\subsubsection{Cryptography from Physical/Hardware Properties.} Apart from complexity/information-theoretic assumptions, security guarantees of cryptographic protocols can also be based on physical/hardware principles/properties. For instance, the physical principles of non-cloneability of quantum states \cite{Park[70],WootZur[82]} and monogamy of entanglement \cite{TerB[04]} are at the heart of quantum cryptography \cite{Rico[21],GisiRib[02]} --- providing an ensemble of (quantum) cryptographic protocols, including quantum key distribution \cite{BenBra[14],ArEk[91],BesSal[92]}, quantum random number generator \cite{MaYu[16]}, closed group quantum digital signatures \cite{ClaDun[12]}, long-term secure data storage \cite{BucDem[17]} and quantum multiparty computation \cite{NiqRuh[10]}. In classical, i.e., non-quantum, settings, physical/hardware principles/properties have been used to circumvent impossibility results, and efficiency and security bounds (e.g., \cite{NayaLing[17],YuaYin[19],EmuOmo[20],CaiPen[21],GupMoo[16],BrandRese[21],Vaswa[18],MaanQuta[20],RenLiZu[21],HoaYav[20]}). Furthermore, protocols based on physical properties or assumptions may offer qualitatively stronger security guarantees than the ones based on purely complexity-theoretic arguments/assumptions \cite{Koch[19]}. The subclass of such protocols that is related to a portion of our work concerns the so-called physically uncloneable functions (PUFs) which are cryptographic functions, defined over stateless hardware modules that implement/realize a function family with some threshold min-entropy output \cite{PappNei[02]} (see \cite{Doosti[22]} for the quantum analogue, called quantumPUF). Contrary to the standard digital systems, the output of a PUF depends on the unavoidably and sometimes purposefully included nanoscale structural disorders in the hardware which lead to a response behavior that cannot be cloned or reproduced exactly, not even by the hardware manufacturer. To capture their complex and disordered structure, formal definitions for PUFs often include requirements for one-wayness and unforgeability of output (typically against a probabilistic polynomial-time (PPT) adversary) \cite{CanFis[01],EyalYeh[03],PappNei[02],Dai[16],MarHei[11],RafaSca[13],Soled[20],StefVla[12],Schro[22]} --- which, in addition to deterministic behavior, are also the requirements for pseudorandom functions (PRFs). 

In this work, our protocol relies on the inherent (random) channel errors occurring in physical layer communications over Gaussian channels with nonzero standard deviation. Known information theoretic arguments establish that channel communications always have an inherent random error component. Using mathematical proofs/arguments and statistical randomness tests such as the ones provided by the NIST \cite{LawAndNIST[10]} and Dieharder \cite{BrownRob[23]} test suites, channel randomness has been established with repsect to various channel characteristics, including received signal strength information \cite{JanaSri[09],JieYan[12],PremSne[14]}, channel state information \cite{Jiang[13],XiQi[16]} and phase shifts \cite{HaiKui[11],XuKui[12]}.  

\subsubsection{Determinism from Probabilistic Events.} Algorithmic information theory \cite{CalTia[02],Chai[77]} provides a fundamental measure of randomness of (finite) strings and (infinite) sequences in terms of their Kolmogorov complexity \cite{Kolmo[68],LiPa[08]}, leading to the notions of algorithmic \cite{CalTia[02],Chai[77]} and $c$-Kolmogorov randomness \cite{CalTia[02]}. Such formal notions of randomness have been used to establish that some (partially) deterministic procedures and events can lead to (pseudo/perfectly)random outcomes (e.g., see \cite{JakMih[19],Tao[08],GallMas[13],Grang[18],BeckRen[12]}). On the other hand, particle physics establishes that even-even nuclei demonstrate high degree of order in result of completely random interactions \cite{BijFra[00],BijAle[01],BertSon[98]}. Outside of particle physics, approximating integer programs is an example problem for which probabilistic constructions lead to deterministic outcomes \cite{Prabh[88]}. Our goal is similar in this work: we use probabilistic errors occurring in channel communications to generate a static and deterministic model $\mathcal{M}$, which can be used as a black box to generate deterministic errors (from target distributions) that are sufficiently independent --- to a PPT adversary --- due to the probabilistic nature of the channel errors. 

\subsubsection{Rounded Gaussians.} Using discrete Gaussian elements to hide secrets is a common approach in lattice-based cryptography. The majority of digital methods for generating Gaussian random variables are based on transformations of uniform random variables~\cite{Knuth[97]}. Popular methods include Ziggurat~\cite{Zigg[00]}, inversion~\cite{Invert[03]}, Wallace~\cite{Wallace[96]}, and Box-Muller~\cite{Box[58]}. Sampling discrete Gaussians can also be done by sampling from some continuous Gaussian distribution, followed by rounding the coordinates to nearby integers~\cite{Pie[10],Box[58],Hul[17]}. Using such rounded Gaussians can lead to better efficiency and, in some cases, better security guarantees for lattice-based cryptographic protocols~\cite{Hul[17]}. In our work, we use rounded Gaussian errors that are derived from deterministic yet sufficiently independent samples taken from continuous Gaussians --- which are themselves generated via our model $\mathcal{M}$.

\subsubsection{Key-homomorphic PRFs.} In a PRF family~\cite{Gold[86]}, each function is specified by a key such that it can be evaluated deterministically given the key but appears to be a random function without the key. For a PRF $F_k$, the index $k$ is called its key or seed. A PRF family $F$ is called key-homomorphic if the set of keys has a group structure and there is an efficient algorithm that, given $F_{k_1}(x)$ and $F_{k_2}(x)$, outputs $F_{k_1 \oplus k_2}(x)$, where $\oplus$ is the group operation~\cite{Naor[99]}. Multiple key-homomorphic PRF families have been constructed via varying approaches~\cite{Naor[99],Boneh[13],Ban[14],Parra[16],SamK[20],Navid[20]}. In this work, we introduce and construct an extended variant of key-homomorphic PRFs, called star-specific key-homomorphic (SSKH) PRFs, which are defined for settings wherein parties constructing the PRFs are part of an interconnection network that can be (re)arranged as a graph comprised of only (undirected) star graphs with restricted vertex intersections. An undirected star graph $S_k$ can be defined as a tree with one internal node and $k$ leaves. \Cref{FigStar} depicts an example star graph, $S_7$, with seven leaves.

\begin{figure}[h!]
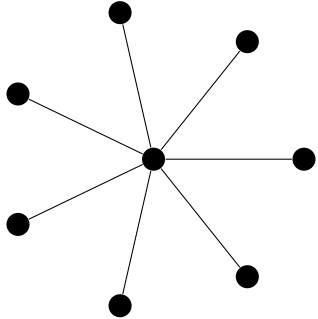

	\centering
	\stargraph{7}{2}
	\caption{An Example Star Graph, $S_7$}\label{FigStar}
\end{figure}

Henceforth, we use the terms star and star graph interchangeably.

\subsubsection{Cover-free Families with Restricted Intersections.} Cover-free families were first defined by Kautz and Singleton in 1964 as superimposed binary codes \cite{Kautz[64]}. They were motivated by investigating binary codes wherein disjunction of at most $r ~(\geq 2)$ codewords is distinct. In early 1980s, cover-free families were studied in the context of group testing \cite{BushFed[84]} and information theory \cite{Ryk[82]}. Erd\"{o}s et al. called the corresponding set systems $r$-cover-free and studied their cardinality for $r=2$~\cite{PaulFrankl[82]} and $r < n$~\cite{PaulFrankl[85]}. 

\begin{definition}[$r$-cover-free Families~\cite{PaulFrankl[82],PaulFrankl[85]}]
	\emph{We say that a family of sets $\mathcal{H} = \{H_i\}_{i=1}^\alpha$ is $r$-cover-free for some integer $r < \alpha$ if there exists no $H_i \in \mathcal{H}$ such that:
		\[H_i \subseteq \bigcup_{H_j \in \mathcal{H}^{(r)}} H_j, \]
		where $\mathcal{H}^{(r)} \subset \mathcal{H}$ is some subset of $\mathcal{H}$ with cardinality $r$. }	
\end{definition}

In addition to earlier applications to group testing \cite{BushFed[84]} and information theory \cite{Ryk[82]}, cover-free families have found many applications in cryptography and communications, including blacklisting~\cite{RaviSri[99]}, broadcast encryption~\cite{CanGara[99],Garay[00],DougR[97],DougRTran[98]}, anti-jamming~\cite{YvoSafavi[99]}, source authentication in networks~\cite{Safavi[99]}, group key predistribution~\cite{ChrisFred[88],Dyer[95],DougRTran[98],DougRTranWei[00]}, compression schemes \cite{Thasis[19]}, fault-tolerant signatures \cite{Gunnar[16],Bardini[21],IdaLuc[18]}, frameproof/traceability codes~\cite{Staddon[01],WeiDoug[98]}, traitor tracing \cite{DonTon[06]}, modification localization on signed documents and redactable signatures \cite{Bardini[15]}, broadcast authentication \cite{Nain[98]}, batch signature verification \cite{Zave[09]}, and one-time and multiple-times digital signature schemes \cite{Josef[03],GM[11]}.

In this work, we initiate the study of new variants of $r$-cover-free families. The motivation behind exploring this direction is to compute the maximum number of SSKH PRFs that can be constructed by overlapping sets of parties. We prove various bounds on the novel variants of $r$-cover-free families and later use them to establish the maximum number of SSKH PRFs that can be constructed by overlapping sets of parties in the presence of active/passive and internal/external adversaries. 

\subsection{Our Contributions}
\subsubsection{Cryptographic Contributions.}
We know that physical layer communications over Gaussian channels introduce (pseudo)random Gaussian errors. Therefore, it is logical to consider whether we can use some processed form of those Gaussian errors to generate deterministic yet sufficiently independent errors to derandomize LWE without weakening its hardness guarantees. In this work, we prove that channel errors errors can be used to derandomize LWE. Our algorithm to realize this uses channel communications over Gaussian channels as the training data for linear regression analysis, whose (optimal) hypothesis is used to generate a static model $\mathcal{M}$ that can be used as a black box to compute deterministic yet sufficiently independent errors belonging to the desired Gaussian distributions. We round the resulting error to the nearest integer, hence moving to a rounded Gaussian distribution, which is reduced modulo the LWE modulus to generate the final error. It is worth mentioning that many hardness proofs for LWE, including Regev's initial proof~\cite{Reg[05]}, used an analogous approach --- but without the linear regression component --- to sample ``LWE errors''~\cite{Reg[05],Gold[10],Duc[15],Hul[17]}. We call our derandomized variant of LWE: learning with linear regression (LWLR). We prove that for certain parameter choices, LWLR is as hard as LWE. 

We introduce a new class of PRFs, called star-specific key-homomorphic (SSKH) PRFs, which are key-homomorphic PRFs that are defined by the respective sets of parties that construct them. In our construction, the sets of parties are arranged as star graphs wherein the leaves represent the parties and edges denote communication channels between them. Each SSKH PRF, $F^{(\partial_i)}_k$, is unique to the set/star of parties, $\partial_i$, that constructs it, i.e., for all inputs $x$, it holds that: $\forall i \neq j: F^{(\partial_i)}_k(x) \neq F^{(\partial_j)}_k(x)$. As an example application of LWLR, we replace LWR with LWLR in the LWR-based key-homomorphic PRF construction from \cite{Ban[14]} to construct the first SSKH PRF family. Due to their conflicting goals, statistical inference and cryptography are almost dual of each other. Given some data, statistical inference aims to identify the distribution that they belong to, whereas in cryptography the central aim is to design a distribution that is hard to predict. Interestingly, our work uses statistical inference to construct novel a cryptographic tool and protocol. In addition to all known applications of key-homomorphic PRFs --- as given in \cite{Boneh[13],Miranda[21]} --- our SSKH PRF family also allows collaborating parties to securely generate pseudorandom nonce/seed without relying on any pre-provisioned secrets; hence supporting applications such as interactive key generation over unauthenticated channels. 

\subsubsection{Mutual Information between Linear Regression Models.}
To quantify the relation between different SSKH PRFs, we examine the mutual information between different linear regression hypotheses that are generated via (training) datasets with overlapping data points. A higher mutual information translates into a stronger relation --- and smaller conditional entropy --- between the corresponding SSKH PRFs, that are generated via those linear regression hypotheses. The following text summarizes the main result that we prove in this context.

Suppose, for $i=1,2,\ldots,\ell$, we have:
$$y_i\sim \mathcal{N}(\alpha+\beta x_i,\sigma^2)\quad\text{and}\quad z_i\sim \mathcal{N}(\alpha+\beta w_i,\sigma^2),$$
with $x_i=w_i$ for $i=1,\ldots,a$. Let $h_1(x)=\hat{\alpha}_1 x+\hat{\beta}_1$ and $h_2(w)=\hat{\alpha}_2 w+\hat{\beta}_2$ be the linear regression hypotheses obtained from the samples $(x_i,y_i)$ and $(w_i,z_i)$, respectively.

\begin{theorem}\label{MutualThm}
	The mutual information between $(\hat{\alpha_1},\hat{\beta_1})$ and $(\hat{\alpha_2},\hat{\beta_2})$ is: 
	\begin{align*}
		&\ -\frac{1}{2}\log\left(1-\frac{\left(\ell C_2-2C_1X_1+aX_2\right)\left(\ell C_2-2C_1W_1+aW_2\right)}{(\ell X_2-X_1^2)(\ell W_2-W_1^2)}\right. \\
		&\qquad\left.+\frac{\left((a-1)C_2-C_3\right)\left((a-1)C_2-C_3+\ell(X_2+W_2)-2X_1W_1\right)}{(\ell X_2-X_1^2)(\ell W_2-W_1^2)}\right),
	\end{align*}
	where $X_1=\sum_{i=1}^{\ell} x_i$, $X_2=\sum_{i=1}^{\ell} x_i^2$, $W_1=\sum_{i=1}^{\ell} w_i$, $W_2=\sum_{i=1}^{\ell} w_i^2$, $C_1=\sum_{i=1}^a x_i=\sum_{i=1}^a w_i$, $C_2=\sum_{i=1}^a x_i^2=\sum_{i=1}^a w_i^2$ and $C_3=\sum_{i=1}^{\ell}\sum_{j=1,j\neq i}^{\ell}x_ix_j$. 
\end{theorem}

\subsubsection{Bounds on $t$-intersection Maximally Cover Free Families.}
We say that a set system $\mathcal{H}$ is (i) $k$-uniform if: $\forall A \in \mathcal{H}: |A| = k$, (ii) at most $t$-intersecting if: $\forall A, B \in \mathcal{H}, B \neq A: |A \cap B| \leq t$. 
\begin{definition}[Maximally Cover-free Families]
	\emph{A family of sets $\mathcal{H}$ is \textit{maximally} cover-free if it holds that:
		\[\forall A \in \mathcal{H}: A \not\subseteq \bigcup\limits_{\substack{B \in \mathcal{H} \\ B \neq A}} B.\]}
\end{definition}

Since we use physical layer communications to generate derandomized LWE instances, a large enough overlap among different sets of parties/devices can lead to reduced collective and conditional entropy for the SSKH PRFs constructed by those sets. It follows trivially that if the sets of parties belong to a maximally cover-free family, then no SSKH PRF can have zero conditional entropy since each set of parties must have at least one member that is exclusive to it. We know from \Cref{MutualThm} that we can compute the mutual information between different linear regression hypotheses from the overlap in their respective training data. Since the training dataset for a set of parties performing linear regression analysis is simply the collection of their mutual communications, it follows that the mutual information between any two SSKH PRFs increases with an increase in the overlap between the sets of parties that construct them. Hence, given a maximum mutual information threshold, \Cref{MutualThm} can be used to compute the maximum acceptable overlap between different sets of parties. To establish the maximum number of SSKH PRFs that can be constructed by such overlapping sets, we derive multiple bounds for the following two types of set-systems (denoted by $\mathcal{H})$: 
\begin{itemize}
	\item $\mathcal{H}$ is at most $t$-intersecting and $k$-uniform,
	\item $\mathcal{H}$ is maximally cover-free, at most $t$-intersecting and $k$-uniform.
\end{itemize}
The following theorem captures our central results on these set-systems:

\begin{theorem}\label{MainThm1}
	Let $k,t\in\mathbb{Z}^+$, and $C<1$ be any positive real number.
	\begin{enumerate}[label = {(\roman*)}]
		\item Suppose $t<k-1$. Then, for all sufficiently large $N$, the maximum size $\nu(N,k,t)$ of a maximally cover-free, at most $t$-intersecting and $k$-uniform family $\mathcal{H}\subseteq 2^{[N]}$ satisfies
		$$CN\leq\nu(N,k,t)<N.$$
		\item Suppose $t<k$. Then, for all sufficiently large $n$, the maximum size $\varpi(n,k,t)$ of an at most $t$-intersecting and $k$-uniform family $\mathcal{H}\subseteq 2^{[n]}$ satisfies
		$$\displayindent0pt
		\frac{Cn^{t+1}}{k(k-1)\cdots(k-t)}\leq\varpi(n,k,t)<\frac{n^{t+1}}{k(k-1)\cdots(k-t)}.$$
		In particular: $$\displayindent0pt
		\nu(N,k,t)\sim N \text{ and } \varpi(n,k,t)\sim\frac{n^{t+1}}{k(k-1)\cdots(k-t)}.$$
	\end{enumerate}
\end{theorem}

We also provide an explicit construction for at most $t$-intersecting and $k$-uniform set systems.

\subsubsection{Maximum Number of SSKH PRFs.}
We use the results from \Cref{MainThm1,MutualThm} to derive the maximum number, $\zeta$, of SSKH PRFs that can be constructed securely against various adversaries (modeled as PPT Turing machines). Specifically, we prove the following:
\begin{itemize}
	\item For an external/eavesdropping adversary with oracle access to the SSKH PRF family, we get:
	\[\zeta \sim \dfrac{n^k}{k!}.\]
	\item For non-colluding semi-honest parties, we get:
	\[\zeta \geq Cn,\]
	where $C<1$ is a positive real number.
\end{itemize}
We also establish the ineffectiveness of the man-in-the-middle attack against our SSKH PRF construction.
 
\subsection{Organization}
The rest of the paper is organized as follows: Section~\ref{Sec2} recalls the concepts and constructs that are relevant to our solutions and constructions. Section~\ref{Sec3} reviews the related work. Section~\ref{Sec4} gives a formal definition of SSKH PRFs. We prove various bounds on (maximally cover-free) at most $t$-intersecting $k$-uniform families in \Cref{Extremal}. In Section~\ref{Sec6}, we present our protocol for generating the desired Gaussian errors from physical layer communications. The section also discusses the implementation, simulation, test results, error analysis, and complexity for our protocol. In \Cref{Mutual}, we analyze the mutual information between different linear regression hypotheses that are generated from overlapping training datasets. In \Cref{LWLRsec}, we define LWLR and generate LWLR instances. In the same section, we reduce the hardness of LWLR to that of LWE. In Section~\ref{Sec7}, we use LWLR to adapt the key-homomorphic PRF construction from~\cite{Ban[14]} to construct the first SSKH PRF family, and prove its security under the hardness of LWLR (and therefore that of LWE). In the same section, we use our results from \Cref{Extremal,Mutual} to establish the maximum number of SSKH PRFs that can be constructed by a given set of parties in the presence of active/passive and external/internal adversaries. Section~\ref{Sec8} gives the conclusion.

\section{Preliminaries}\label{Sec2}
For a positive integer $n$, let: $[n] = \{1, \dots, n\}.$ As mentioned earlier, we use the terms star and star graph interchangeably. For a vector $\mathbf{v}= (v_1, v_2, \ldots, v_w) \in \mathbb{R}^w$, the Euclidean and infinity norms are defined as: $||\mathbf{v}|| = \sqrt{(\sum_{i=1}^w v_i^2)}$ and $||\mathbf{v}||_\infty = \max(|v_1|, |v_2|, \ldots, |v_w|),$ respectively. In this text, vectors and matrices are denoted by bold lower case letters and bold upper case letters, respectively. We say that an algorithm is efficient if its running time is polynomial in its input size.

\begin{definition}
	\emph{The probability density function (p.d.f.) of a continuous random variable $X$ with support $S$ is an integrable function $f_X$ such that following conditions hold:
	\begin{itemize}
		\item $\forall x \in S: f_X(x) > 0$,
		\item $\int_{S} f_X(x) dx = 1$,
		\item $\Pr[X \in \Im] = \int_{\Im} f_X(x) dx.$
	\end{itemize}}	
\end{definition}

\begin{definition}
	\emph{The probability mass function (p.m.f.) $p_X$ of a discrete random variable $X$ with support $S$ is a function for which the following hold:
	\begin{itemize}
		\item $\forall x \in S: \Pr[X = x] = p_X(x) > 0$,
		\item $\sum\limits_{x \in S} p_X(x) = 1$, 
		\item $\Pr[X \in \Im] = \sum\limits_{x \in \Im} p_X(x).$
	\end{itemize}}
\end{definition}

\subsection{Entropy}
The concept of entropy was originally introduced as a thermodynamic construct by Rankine in 1850 \cite{True[80]}. It was later adapted to information theory by Shannon \cite{Shannon[48]}, where it denotes a measure of the uncertainty associated with a random variable, i.e., (information) entropy is defined as a measure of the average information content that is missing when value of a random variable is not known. 

\begin{definition}
	\emph{For a finite set $\mathcal{S} = \{s_1, s_2, \ldots, s_n\}$ with probabilities $p_1, p_2, \ldots, p_n$, the entropy of the probability distribution over $\mathcal{S}$ is defined as:
	\[H(\mathcal{S}) = \sum\limits_{i=1}^n p_i \log \dfrac{1}{p_i}. \]} 
\end{definition}

\subsection{Lattices}
A lattice $\mathrm{\Lambda}$ of $\mathbb{R}^w$ is defined as a discrete subgroup of $\mathbb{R}^w$. In cryptography, we are interested in integer lattices, i.e., $\mathrm{\Lambda} \subseteq \mathbb{Z}^w$. Given $w$-linearly independent vectors $\textbf{b}_1,\dots,\textbf{b}_w \in \mathbb{R}^w$, a basis of the lattice generated by them can be represented as the matrix $\mathbf{B} = (\textbf{b}_1,\dots,\textbf{b}_w) \in \mathbb{R}^{w \times w}$. The lattice generated by $\mathbf{B}$ is the following set of vectors:
\[\mathrm{\Lambda}(\textbf{B}) = \left\{ \sum\limits_{i=1}^w c_i \textbf{b}_i: c_i \in \mathbb{Z} \right\}.\]
Historically, lattices received attention from illustrious mathematicians, including Lagrange, Gauss, Dirichlet, Hermite, Korkine-Zolotareff, and Minkowski (see \cite{Laga[73],Gauss[81],Herm[50],Kork[73],Mink[10],JacSte[98]}). Problems in lattices have been of interest to cryptographers since 1997, when Ajtai and Dwork~\cite{Ajtai[97]} proposed a lattice-based public key cryptosystem following Ajtai's~\cite{Ajtai[96]} seminal worst-case to average-case reductions for lattice problems. In lattice-based cryptography, \textit{q-ary} lattices are of particular interest; they satisfy the following condition: $$q \mathbb{Z}^w \subseteq \mathrm{\Lambda} \subseteq \mathbb{Z}^w,$$ for some (possibly prime) integer $q$. In other words, the membership of a vector $\textbf{x}$ in $\mathrm{\Lambda}$ is determined by $\textbf{x}\bmod q$. Given a matrix $\textbf{A} \in \mathbb{Z}^{w \times n}_q$ for some integers $q, w, n,$ we can define the following two $n$-dimensional \textit{q-ary} lattices,

\[\mathrm{\Lambda}_q(\textbf{A}) = \{\textbf{y} \in \mathbb{Z}^n: \textbf{y} = \textbf{A}^T\textbf{s} \bmod q \text{ for some } \textbf{s} \in \mathbb{Z}^w \}, \]
\[\hspace{-28mm} \mathrm{\Lambda}_q^{\perp}(\textbf{A}) = \{\textbf{y} \in \mathbb{Z}^n: \textbf{Ay} = \textbf{0} \bmod q \}.\]

The first \textit{q-ary} lattice is generated by the rows of $\textbf{A}$; the second contains all vectors that are orthogonal (modulo $q$) to the rows of $\textbf{A}$. Hence, the first \textit{q-ary} lattice, $\mathrm{\Lambda}_q(\textbf{A})$, corresponds to the code generated by the rows of $\textbf{A}$ whereas the second, $\mathrm{\Lambda}_q^{\perp}(\textbf{A})$, corresponds to the code whose parity check matrix is $\textbf{A}$. For a complete introduction to lattices, we refer the interested reader to the monographs by Gr\"{a}tzer~\cite{Gratzer[03],Gratzer[09]}.

\subsection{Gaussian Distributions}
Gaussian sampling is an extremely useful tool in lattice-based cryptography. Introduced by Gentry et al. \cite{Gen[08]}, Gaussian sampling takes a short basis $\textbf{B}$ of a lattice $\mathrm{\Lambda}$ and an arbitrary point $\textbf{v}$ as inputs and outputs a point from a Gaussian distribution discretized on the lattice points and centered at $\textbf{v}$. Gaussian sampling does not leak any information about the lattice $\mathrm{\Lambda}$. It has been used directly to construct multiple cryptographic schemes, including hierarchical identity-based encryption \cite{AgDan[10],CashDenn[10]}, standard model signatures \cite{AgDan[10],XavBoy[10]}, and attribute-based encryption \cite{DanSer[14]}. In addition, Gaussian sampling/distribution also plays an important role in other hard lattice problems, such as, learning single periodic neurons \cite{SongZa[21]}, and has direct connections to standard lattice problems \cite{DivDan[15],NoDan[15],SteDavid[15]}.

\begin{definition}
	\emph{A continuous Gaussian distribution, $\mathcal{N}^w(\textbf{v},\sigma^2)$, over $\mathbb{R}^w$, centered at some $\mathbf{v} \in \mathbb{R}^w$ with standard deviation $\sigma$ is defined for $\textbf{x} \in \mathbb{R}^w$ as the following density function:
	\[\mathcal{N}_\textbf{x}^w(\textbf{v},\sigma^2) = \left( \dfrac{1}{\sqrt{2 \pi \sigma^2}} \right)^w \exp \left(\frac{-||\textbf{x} - \textbf{v}||^2}{2 \sigma^2}\right). \]	}
\end{definition}

A rounded Gaussian distribution can be obtained by simply rounding the samples from a continuous Gaussian distribution to their nearest integers. Rounded Gaussians have been used to establish hardness of LWE \cite{Reg[05],Gold[10],Duc[15]} --- albeit not as frequently as discrete Gaussians. 

\begin{definition}[Adapted from \cite{Hul[17]}]\label{roundGauss}
	\emph{A rounded Gaussian distribution, $\mathrm{\Psi}^w(\textbf{v},\hat{\sigma}^2)$, over $\mathbb{Z}^w$, centered at some $\textbf{v} \in \mathbb{Z}^w$ with parameter $\sigma$ is defined for $\textbf{x} \in \mathbb{Z}^w$ as:
	\[\mathrm{\Psi}^w_\textbf{x}(\textbf{v},\hat{\sigma}^2) = \int_{A_\textbf{x}} \mathcal{N}_\textbf{s}^w(\textbf{v},\sigma^2)\,d\textbf{s} = \int_{A_\textbf{x}} \left( \dfrac{1}{\sqrt{2 \pi \sigma^2}} \right)^w \exp\left( \dfrac{-||\textbf{s} - \textbf{v}||^2}{2 \sigma^2} \right)\,d\textbf{s}, \]
	where $A_\textbf{x}$ denotes the region $\prod_{i=1}^{w} [x_i - \frac{1}{2}, x_i + \frac{1}{2})$; $\hat{\sigma}$ and $\sigma$ are the standard deviations of the rounded Gaussian and its underlying continuous Gaussian, respectively, such that $\hat{\sigma} = \sqrt{\sigma^2 + 1/12}$.}
\end{definition}

\begin{definition}[Gaussian channel]\label{Gauss}
	\emph{A Gaussian channel is a discrete-time channel with input $x_i$ and output $y_i = x_i + \varepsilon_i$, where $\varepsilon_i$ is drawn i.i.d. from a Gaussian distribution $\mathcal{N}(0, \sigma^2)$, with mean 0 and standard deviation $\sigma$, which is assumed to be independent of the signal $x_i$.}
\end{definition}

\begin{definition}[Gram-Schmidt norm]
	\emph{Let $\textbf{B} = (\textbf{b}_i)_{i \in [w]}$ be a finite basis, and $\widetilde{\textbf{B}} = (\tilde{\textbf{b}}_i)_{i \in [w]}$ be its Gram-Schmidt orthogonalization. Then, Gram-Schmidt norm of $\textbf{B}$ is defined as:
	\[||\textbf{B}||_{GS} = \max\limits_{i \in [w]} ||\tilde{\textbf{b}}_i||.\]	}
\end{definition}

For an introduction to Gram-Schmidt orthogonalization, see \cite{Leon[12]}.

\begin{definition}[Discrete Gaussian over Lattices]
	\emph{Given a lattice $\mathrm{\Lambda} \in \mathbb{Z}^w$, the discrete Gaussian distribution over $\mathrm{\Lambda}$ with standard deviation $\sigma \in \mathbb{R}$ and center $\textbf{v} \in \mathbb{R}^w$ is defined as:
	\[D(\mathrm{\Lambda},\textbf{v},\sigma^2)_\textbf{x} = \dfrac{\rho_\textbf{x}(\textbf{v},\sigma^2)}{\rho_\mathrm{\Lambda}(\textbf{v},\sigma^2)};\ \forall \textbf{x} \in \mathrm{\Lambda},\]
	where $\rho_\mathrm{\Lambda}(\textbf{v}, \sigma^2) = \sum\limits_{\textbf{x}_i \in \mathrm{\Lambda}} \rho_{\textbf{x}_i}(\textbf{v},\sigma^2)$
	and $$\rho_\textbf{x}(\textbf{v},\sigma^2) = \exp\left( - \dfrac{||\textbf{x} - \textbf{v}||^2}{2\sigma^2}\right).$$}
\end{definition}

The smoothing parameter is defined as a measure of the ``difference'' between discrete and standard Gaussians, that are defined over identical parameters. Informally, it is the smallest $\sigma$ required by a discrete Gaussian distribution, over a lattice $\mathrm{\Lambda}$, to behave like a continuous Gaussian --- up to some acceptable statistical error. For more details, see \cite{MicReg[04],DO[07],ChungD[13]}. Various methods such as computing the cumulative density function, taking convolutions of smaller deviation discrete Gaussians, and rejection/Bernoulli/Binomial/Ziggurat/Knuth-Yao/CDT (cumulative distribution table) sampling have been employed to efficiently sample from discrete Gaussians for lattice-based cryptography \cite{MarSet[17],BartPie[19],Bra[13],AndPat[13],GalDwa[14],GenDan[18],RicoHow[20],KarnF[16],KarVerVer[18],WaltDan[17],DucMann[14],SakSte[20],ZhaAmi[20],OngDuc[12],DucDur[13],BeeMoo[16],PhiPet[16],AyeRaff[18],XiDu[22]}. 

\begin{theorem}[Drowning/Smudging \cite{GentryThesis[09]}]
	Let $\sigma > 0$ and $y \in \mathbb{Z}$. The statistical distance between $\mathrm{\Psi}(v,\sigma^2)$ and $\mathrm{\Psi}(v,\sigma^2) + y$ is at most $|y|/\sigma$.
\end{theorem}

Typically, in lattice-based cryptography, drowning/smudging is used to hide some information by introducing a sufficiently large random noise --- with large standard deviation --- such that the resulting distribution is, to the desired degree, independent of the information that needs to be hidden \cite{MartDeo[17],Gold[10],Dodis[10],AshaJain[12],AlpErt[12],Ban[12],Grg[13],Gent[09],DamRon[14],GentryThesis[09],MartAlex[21]}. However, in our work, we do not use it for that purpose; instead, we use it argue about the insignificance of a small component of the total error. 

\subsection{Learning with Errors}\label{LWE}
The learning with errors (LWE) problem~\cite{Reg[05]} is at the center of the majority of lattice-based cryptographic constructions~\cite{Peikert[16]}. LWE is known to be hard based on the worst-case hardness of standard lattice problems such as GapSVP (decision version of the Shortest Vector Problem) and SIVP (Shortest Independent Vectors Problem)~\cite{Reg[05],Pei[09]}. Multiple variants of LWE such as ring LWE~\cite{Reg[10]}, module LWE~\cite{Ade[15]}, cyclic LWE~\cite{Charles[20]}, continuous LWE~\cite{Bruna[20]}, \textsf{PRIM LWE}~\cite{SehrawatVipin[21]}, middle-product LWE~\cite{Miruna[17]}, group LWE~\cite{NicMal[16]}, entropic LWE \cite{ZviVin[16]}, universal LWE \cite{YanHua[22]}, and polynomial-ring LWE~\cite{Damien[09]} have been developed since 2010. Many cryptosystems rely on the hardness of LWE, including (identity-based, leakage-resilient, fully homomorphic, functional, public-key/key-encapsulation, updatable, attribute-based, inner product, predicate) encryption~\cite{AnaFan[19],KimSam[19],WangFan[19],Reg[05],Gen[08],Adi[09],Reg[10],Shweta[11],Vinod[11],Gold[13],Jan[18],Round[19],Bos[18],Bos[16],WBos[15],Brak[14],Fan[12],Joppe[13],Adriana[12],Lu[18],AndMig[22],MiaSik[22],Boneh[13],Vipin[19],LiLi[22],RaviHow[22],ShuiTak[20],SerWe[15]}, oblivious transfer~\cite{Pei[08],Dott[18],Quach[20]}, (blind) signatures~\cite{Gen[08],Vad[09],Markus[10],Vad[12],Tesla[20],Dili[17],FALCON[20]}, PRFs with special algebraic properties~\cite{Ban[12],Boneh[13],Ban[14],Ban[15],Bra[15],Vipin[19],KimDan[17],RotBra[17],RanChen[17],KimWu[17],KimWu[19],Qua[18],KevAna[14],SinShi[20],BanDam[18],PeiShi[18]}, verifiable/homomorphic/function secret sharing~\cite{SehrawatVipin[21],GHL[21],Boy[17],DodHal[16],GilLin[17],LisPet[19]}, hash functions~\cite{Katz[09],Pei[06]}, secure matrix multiplication computation~\cite{Dung[16],Wang[17]}, verifiable quantum computations~\cite{Urmila[18],Bra[21],OrNam[21],ZhenAlex[21]}, noninteractive zero-knowledge proof system for (any) NP language~\cite{Sina[19]}, classically verifiable quantum computation~\cite{Urmila[18]}, certifiable randomness generation \cite{Bra[21]}, obfuscation~\cite{Huijia[16],Gentry[15],Hal[17],ZviVin[16],AnanJai[16],CousinDi[18]}, multilinear maps \cite{Grg[13],Gentry[15],Gu[17]}, lossy-trapdoor functions \cite{BellKil[12],PeiW[08],HoWee[12]}, quantum homomorphic encryption \cite{Mahadev[18],OreNic[21]}, key exchange \cite{WBos[15],Alkim[16],StebMos[16]}, zero-knowledge protocols for QMA (Quantum Merlin Arthur: a quantum analogue of NP) \cite{AlaAlex[20],FermiMala[22],MalaJam[21],BitOm[20],HuiChu[20],ColaTho[20],GiuCha[21],Omi[21]}, and many more~\cite{Peikert[16],JiaZhen[20],KatzVadim[21]}.

\begin{definition}[Decision-LWE \cite{Reg[05]}]\label{defLWE}
	\emph{For positive integers $w$ and $q \geq 2$, and an error (probability) distribution $\chi$ over $\mathbb{Z}$, the decision-LWE${}_{w, q, \chi}$ problem is to distinguish between the following pairs of distributions: 
		\[((\textbf{a}_i, \langle \textbf{a}_i, \textbf{s} \rangle + e_i \bmod q))_i \quad \text{and} \quad ((\textbf{a}_i, u_i))_i,\] 
		where $i \in [\poly(w)],\, \textbf{a}_i \xleftarrow{\; \$ \;} \mathbb{Z}^{w}_q,\, \textbf{s} \in \mathbb{Z}^w_q,\, e_i \leftarrow \chi,$ and $u_i \xleftarrow{\; \$ \;} \mathbb{Z}_q$.}
\end{definition} 

Regev~\cite{Reg[05]} showed that for certain noise distributions and a sufficiently large $q$, the LWE problem is as hard as the worst-case SIVP and GapSVP under a quantum reduction (see~\cite{Pei[09],Bra[13],ChrisOded[17]} for other reductions). Standard instantiations of LWE assume $\chi$ to be a rounded or discrete Gaussian distribution. Regev's proof requires $\alpha q \geq 2 \sqrt{w}$ for ``noise rate'' $\alpha \in (0,1)$. These results were extended by Applebaum et al.~\cite{Benny[09]} to show that the fixed secret $\textbf{s}$ can be sampled from a low norm distribution. Specifically, they showed that sampling $\textbf{s}$ from the noise distribution $\chi$ does not weaken the hardness of LWE. Later, Micciancio and Peikert discovered that a simple low-norm distribution also works as $\chi$~\cite{Micci[13]}. 

\subsection{Pseudorandom Functions}
In a pseudorandom function (PRF) family~\cite{Gold[86]}, each function is specified by a key such that it can be evaluated deterministically with the key but behaves like a random function without it. Here, we recall the formal definition of a PRF family. Recall that an ensemble of probability distributions is a sequence $\{X_n\}_{n \in \mathbb{N}}$ of probability distributions.

\begin{definition}[Negligible Function]\label{Neg}
	\emph{For security parameter $\L$, a function $\eta(\L)$ is called \textit{negligible} if for all $c > 0$, there exists a $\L_0$ such that $\eta(\L) < 1/\L^c$ for all $\L > \L_0$.}
\end{definition}

\begin{definition}[Computational Indistinguishability~\cite{Gold[82]}]
	\emph{Let $X = \{X_\lambda\}_{\lambda \in \mathbb{N}}$ and $Y = \{Y_\lambda\}_{\lambda \in \mathbb{N}}$ be ensembles, where $X_\lambda$'s and $Y_\lambda$'s are probability distributions over $\{0,1\}^{\kappa(\lambda)}$ for $\lambda \in \mathbb{N}$ and some polynomial $\kappa(\lambda)$. We say that $\{X_\lambda\}_{\lambda \in \mathbb{N}}$ and $\{Y_\lambda\}_{\lambda \in \mathbb{N}}$ are polynomially/computationally indistinguishable if the following holds for every (probabilistic) polynomial-time algorithm $\mathcal{D}$ and all $\lambda \in \mathbb{N}$:
		\[\Big| \Pr[t \leftarrow X_\lambda: \mathcal{D}(t) = 1] - \Pr[t \leftarrow Y_\lambda: \mathcal{D}(t) = 1] \Big| \leq \eta(\lambda),\]
		where $\eta$ is a negligible function.}
\end{definition}	

\begin{remark}[Perfect Indistinguishability]
	We say that $\{X_\lambda\}_{\lambda \in \mathbb{N}}$ and $\{Y_\lambda\}_{\lambda \in \mathbb{N}}$ are perfectly indistinguishable if the following holds for all $t$:
	\[\Pr[t \leftarrow X_\lambda] = \Pr[t \leftarrow Y_\lambda].\]
\end{remark}

We consider adversaries interacting as part of probabilistic experiments called games. For an adversary $\mathcal{A}$ and two games $\mathfrak G_1, \mathfrak G_2$ with which it can interact, $\mathcal{A}'s$ distinguishing advantage is: 
\[Adv_{\mathcal{A}}(\mathfrak{G}_1, \mathfrak{G}_2) := \Big|\Pr[\mathcal{A} \text{ accepts in } \mathfrak G_1] - \Pr[\mathcal{A} \text{ accepts in } \mathfrak G_2]\Big|.\]
For the security parameter $\L$, the two games are said to be computationally indistinguishable if it holds that: $$Adv_{\mathcal{A}}(\mathfrak{G}_1, \mathfrak{G}_2) \leq \eta(\L),$$ where $\eta$ is a negligible function.

\begin{definition}[PRF] 
	\emph{Let $A$ and $B$ be finite sets, and let $\mathcal{F} = \{ F_k: A \rightarrow B \}$ be a function family, endowed with an efficiently sampleable distribution (more precisely, $\mathcal{F}, A$, and $B$ are all indexed by the security parameter $\L)$. We say that $\mathcal{F}$ is a PRF family if the following two games are computationally indistinguishable:
	\begin{enumerate}[label=(\roman*)]
		\item Choose a function $F_k \in \mathcal{F}$ and give the adversary adaptive oracle access to $F_k$.
		\item Choose a uniformly random function $U: A \rightarrow B$ and give the adversary adaptive oracle access to $U.$
	\end{enumerate}}
\end{definition} 

Hence, PRF families are efficient distributions of functions that cannot be efficiently distinguished from the uniform distribution. For a PRF $F_k \in \mathcal{F}$, the index $k$ is called its key/seed. PRFs have a wide range of applications, most notably in cryptography, but also in computational complexity and computational learning theory. For a detailed introduction to PRFs and review of the noteworthy results, we refer the interested reader to the survey by Bogdanov and Rosen \cite{AndAlo[17]}. 

\subsection{Linear Regression}\label{Sec5}
Linear regression is a linear approach to model relationship between a dependent variable and explanatory/independent variable(s). As is the case with most statistical analysis, the goal of regression is to make sense of the observed data in a useful manner. It analyzes the training data and attempts to model the relationship between the dependent and explanatory/independent variable(s) by fitting a linear equation to the observed data. These predictions (often) have errors, which cannot be predicted accurately~\cite{Trevor[09],Montgo[12]}. For linear regression, the mean and variance functions are defined as:
\[\E(Y |X = x) = \beta_0 + \beta_1 x \quad \text{and} \quad \text{var}(Y |X = x) = \sigma^2,\]
respectively, where $\E(\cdot)$ and $\sigma$ denote the expected value and standard deviation, respectively; $\beta_0$ represents the intercept, which is the value of $\E(Y |X = x)$ when $x$ equals zero; $\beta_1$ denotes the slope, i.e., the rate of change in $\E(Y |X = x)$ for a unit change in $X$. The parameters $\beta_0$ and $\beta_1$ are also known as \textit{regression coefficients}. 

For any regression model, the observed value $y_i$ might not always equal its expected value $\E(Y |X = x_i)$. This difference between the observed data and the expected value is called statistical error, and is defined as: $$\epsilon_i = y_i - \E(Y |X = x_i).$$ For linear regression, errors are random variables that correspond to the vertical distance between the point $y_i$ and the mean function $\E(Y |X = x_i)$. Depending on the type and size of the training data, different algorithms such as gradient descent and least squares may be used to compute the values of $\beta_0$ and $\beta_1$. In this paper, we employ least squares linear regression to estimate the values of $\beta_0$ and $\beta_1$, and generate the optimal hypothesis for the target function. Due to the inherent error in all regression models, it holds that: $$h(x) = f(x) + \varepsilon_x,$$ where $h(x)$ is the (optimal) hypothesis of the linear regression model, $f(x)$ is the target function and $\varepsilon_x$ is the total (reducible + irreducible) error at point $x$.

\subsection{Interconnection Network}
In an interconnection network, each device is independent and connects with other devices via point-to-point links, which are two-way communication lines. Therefore, an interconnection network can be modeled as an undirected graph $G = (V, E)$, where each device is a vertex in $V$ and edges in $E$ represent communication lines/channels between the devices. Next, we recall some basic definitions/notations for undirected graphs. 

\begin{definition}
	\emph{The degree $\deg(v)$ of a vertex $v \in V$ is the number of adjacent vertices it has in a graph $G$. The degree of a graph $G$ is defined as: $\deg(G) = \max\limits_{v \in V}(\deg(v))$.}
\end{definition}

If $\deg(v_i) = \deg(v_j)$ for all $v_i, v_j \in V$, then $G$ is called a regular graph. Since it is easy to construct star graphs that are hierarchical, vertex edge symmetric, maximally fault tolerant, and strongly resilient along with having other desirable properties such as small(er) degree, diameter, genus and fault diameter \cite{Akera[89],Sheldon[94]}, networks of star graphs are well-suited to model interconnection networks. For a detailed introduction to interconnection networks, we refer the interested reader to the comprehensive book by Duato et al. \cite{Sudha[02]}.

\subsection{Some Useful Results}
Here, we recall two useful, elementary results from probability theory. 

\begin{definition}[Chebyshev inequality \cite{Bie[53],Cheby[67]}]
	\emph{Let $X$ be a random variable with mean $\mu$ and variance var$(X) = \sigma^2$. Then, the following holds for all $\varsigma >0$:
	\[\Pr[|X - \mu| \geq \varsigma] \leq \dfrac{\sigma^2}{\varsigma^2}. \]}
\end{definition}

\begin{definition}[The Union Bound \cite{Boole[47]}]
	\emph{For any random events $A_1, A_2, \ldots, A_n$, it holds that:
	\[\Pr\left( \bigcup\limits_{i=1}^n A_i \right) \leq \sum\limits_{i=1}^n \Pr(A_i). \]}
\end{definition}

Let $X_1, X_2, \ldots, X_n$ be i.i.d. random variables from the same distribution, i.e., all $X_i$'s have the same mean $\mu$ and standard deviation $\sigma$. Let random variable $\overline{X}_n$ be the average of $X_1, \ldots, X_n$. Then, $\overline{X}_n$ converges almost surely to $\mu$ as $n\rightarrow\infty$.

\section{Related Work}\label{Sec3}
\subsection{Learning with Rounding}\label{LWR}
Naor and Reingold \cite{Naor[9]} introduced synthesizers to construct PRFs via a hard-to-learn deterministic function. The obstacle in using LWE as the hard learning problem in their synthesizers is that the hardness of LWE relies directly on random errors. In fact, without the error, LWE becomes a trivial problem, that can be solved via Gaussian elimination. Therefore, in order to use these synthesizers for constructing LWE-based PRFs, there was a need to replace the random errors with deterministic yet sufficiently independent errors such that the hardness of LWE is not (significantly) weakened. Banerjee et al.~\cite{Ban[12]} addressed this problem by introducing the learning with rounding (LWR) problem, wherein instead of adding a small random error, as done in LWE, a deterministically rounded version of the sample is generated. For $q \geq p \geq 2$, the rounding function, $\lfloor \cdot \rceil_p: \mathbb{Z}_q \rightarrow \mathbb{Z}_p$, is defined as:
\[\lfloor x \rceil_p = \left\lfloor \dfrac{p}{q} \cdot x \right\rceil,\]
i.e., if $\lfloor x \rceil_p = y$, then $y \cdot \lfloor q/p \rceil$ is the integer multiple of $\lfloor q/p \rceil$ that is nearest to $x$. Hence, the error in LWR originates from deterministically rounding $x$ to a (relatively) nearby value in $\mathbb{Z}_p$.  

\begin{definition}[LWR Distribution~\cite{Ban[12]}]
	\emph{Let $q \geq p \geq 2$ be positive integers, then: for a vector $\textbf{s} \in \mathbb{Z}^w_q$, LWR distribution $L_\textbf{s}$ is defined to be a distribution over $\mathbb{Z}^w_q \times \mathbb{Z}_p$ that is obtained by choosing a vector $\textbf{a} \xleftarrow{\; \$ \;} \mathbb{Z}^w_q$ and outputting $(\textbf{a},b = \lfloor \langle \textbf{a},\textbf{s} \rangle \rceil_p).$}
\end{definition}

For a given distribution over $\textbf{s} \in \mathbb{Z}^w_q$ (e.g., the uniform distribution), the decision-LWR${}_{w,q,p}$ problem is to distinguish (with advantage non-negligible in $w)$ between some fixed number of independent samples $(\textbf{a}_i,b_i) \leftarrow L_\textbf{s}$, and the same number of samples drawn uniformly from $\mathbb{Z}^w_q \times \mathbb{Z}_p$. Banerjee et al. proved decision-LWR to be as hard as decision-LWE for a setting of parameters where the modulus and modulus-to-error ratio are superpolynomial in the security parameter \cite{Ban[12]}. Alwen et al.~\cite{Alwen[13]}, Bogdanov et al.~\cite{Andrej[16]}, and Bai et al.~\cite{ShiBai[18]} made further improvements on the range of parameters and hardness proofs for LWR. LWR has been used to construct pseudorandom generators/functions~\cite{Ban[12],Boneh[13],Ban[14],Vipin[19],VipinThesis[19],BenoSan[17]}, and probabilistic~\cite{Jan[18],Round[19]} and deterministic~\cite{Xie[12]} encryption schemes.

As mentioned earlier, hardness reductions of LWR hold for superpolynomial approximation factors over worst-case lattices. Montgomery~\cite{Hart[18]} partially addressed this issue by introducing a new variant of LWR, called Nearby Learning with Lattice Rounding problem, which supports unbounded number of samples and polynomial (in the security parameter) modulus.  

\subsection{LWR/LWE-based Key-homomorphic PRF\lowercase{s}}\label{foll}
Since LWR allows generating derandomized/deterministic LWE instances, it can be used as the hard-to-learn deterministic function in Naor and Reingold's synthesizers, and therefore, construct LWE-based PRF families for specific parameters. Due to the indispensable small error, LWE-based key-homomorphic PRFs only achieve what is called `almost homomorphism'~\cite{Boneh[13]}. 

\begin{definition}[Key-homomorphic PRF~\cite{Boneh[13]}]
	\emph{Let $F: \mathcal{K} \times \mathcal{X} \rightarrow \mathbb{Z}^w_q$ be an efficiently computable function such that $(\mathcal{K}, \oplus)$ is a group. We say that the tuple $(F, \oplus)$ is a $\gamma$-almost key-homomorphic PRF if the following two properties hold:
	\begin{enumerate}[label=(\roman*)]
		\item $F$ is a secure PRF,
		\item for all $k_1, k_2 \in \mathcal{K}$ and $x \in \mathcal{X}$, there exists $\textbf{e} \in [0, \gamma]^w$ such that: $$F_{k_1}(x) + F_{k_2}(x) = F_{k_1 \oplus k_2}(x) + \textbf{e} \bmod q.$$
	\end{enumerate}}
\end{definition}

Multiple key-homomorphic PRF families have been constructed via varying approaches~\cite{Naor[99],Boneh[13],Ban[14],Parra[16],SamK[20],Navid[20]}. 

\section{SSKH PRF: Definition}\label{Sec4}

For any two sets $X$ and $Y$, let $\texttt{PartFunc}(X,\,Y)$ denote the space of partial functions from $X$ to $Y$.

\begin{definition}\label{P2SI}
	\emph{An efficient randomized algorithm $\mathfrak{A}: \Re \to \texttt{PartFunc}(\mathbb{Z},\,\mathbb{Z})$ is probabilistic to static-independent (P2SI) if it takes some random $r\in\Re$ (chosen according to some fixed probability distribution on $\Re$) as input, and outputs a deterministic function $\mathfrak{M}_r: \mathcal{X}_r \to \mathbb{Z}$, where $\mathcal{X}_r \subseteq \mathbb{Z}$, such that, for all $x_i, x_j\in\mathbb{Z}$ with $x_i \neq x_j$:
	\begin{enumerate}
		\item\label{P21} the probability distributions of $\mathfrak{M}_{r}(x_i)$  (taken with respect to the randomness $r\in\Re$ such that $x_i\in\mathcal{X}_r$) and $\mathfrak{M}_{r}(x_j)$ are both computationally indistinguishable from rounded Gaussians with the same parameters, 
		\item the following quantities:
		$$H[\mathfrak{M}_{r}(x_i) | \mathfrak{M}_{r}(x_j)],\quad H[\mathfrak{M}_{r}(x_j) | \mathfrak{M}_{r}(x_i)],\quad H[\mathfrak{M}_{r}(x_i)],\quad H[\mathfrak{M}_{r}(x_j)]$$
		are equal up to a negligible function.
	\end{enumerate}} 
\end{definition}

\begin{figure}[h!]
	\centering
	\includegraphics[scale=.5]{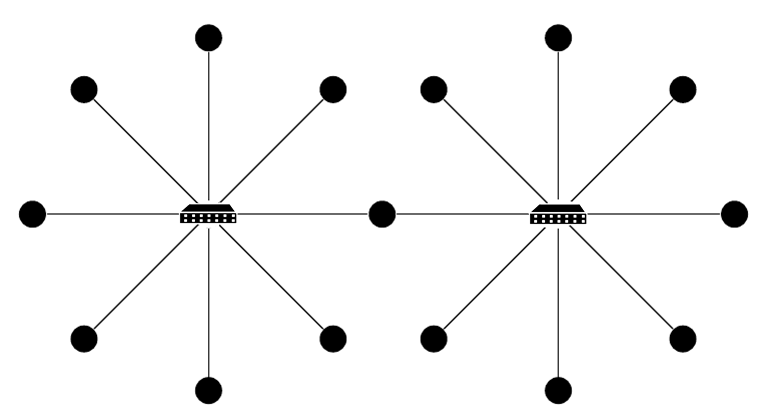}
	\caption{An Example Interconnection Graph}\label{DefFig}
\end{figure}

Next, we define a star-specific key-homomorphic (SSKH) PRF family. Let $G = (V,E)$ be a graph, representing an interconnection network, containing multiple star graphs wherein the leaves of each star graph, $\partial$, represent unique parties and the root represents a central hub that broadcasts messages to all leaves/parties in $\partial$. Different star graphs may have arbitrary numbers of shared leaves. Henceforth, we call such a graph an \textit{interconnection graph}. \Cref{DefFig} depicts a simple interconnection graph with two star graphs, each containing one central hub, respectively, along with eight parties/leaves out of which one leaf is shared by both star graphs. Note that an interconnection graph can also be viewed as a bipartite graph with its vertices partitioned into two disjoint subsets, $V_1, V_2 \subset V$, wherein the vertices in $V_1$ and $V_2$ represent the central hubs and parties, respectively.

\begin{definition}\label{MainDef}
	\emph{Let graph $G = (V, E)$ be an interconnection graph with a set of vertices $V$ and a set of edges $E$. Let there be $\rho$ star graphs $\partial_1, \ldots, \partial_\rho$ in $G$. Let $\mathcal{F}=\left(F^{(\partial_i)}\right)_{i=1,\ldots,\rho}$ be a family of PRFs, where, for each $i$, $F^{(\partial_i)}: \mathcal{K} \times \mathcal{X} \rightarrow \mathbb{Z}^w_q$ with $(\mathcal{K}, \oplus)$ a group. Then, we say that the tuple $(\mathcal{F}, \oplus)$ is a star-specific $(\delta,\gamma,p)$-almost key-homomorphic PRF family if the following two conditions hold:
	\begin{enumerate}[label=(\roman*)]
		\item for all $\partial_i \neq \partial_j~(i,j \in [\rho]), k \in \mathcal{K}$ and $x \in \mathcal{X}$, it holds that:
		\[\Pr[F^{(\partial_i)}_k(x) = F^{(\partial_j)}_k(x)] \leq \delta^w+\eta(\L),\] 
		where $F^{(\partial)}_k(x)$ denotes the PRF computed by parties in star graph $\partial \subseteq V(G)$ on input $x \in \mathcal{X}$ and key $k \in \mathcal{K}$, and $\eta(\L)$ is a negligible function in the security parameter $\L$,
		\item for all $k_1, k_2 \in \mathcal{K}$ and $x \in \mathcal{X}$, there exists a vector $\textbf{e}=(e_1,\ldots,e_w)$ satisfying:
		$$F_{k_1}^{(\partial)}(x) + F^{(\partial)}_{k_2}(x) = F^{(\partial)}_{k_1 \oplus k_2}(x) + \textbf{e} \bmod q,$$
		such that for all $a\in [w]$, it holds that: $\Pr[-\gamma\leq e_a\leq\gamma]\geq p$.
	\end{enumerate}}
\end{definition}

\section{Maximally Cover-free At Most $t$-intersecting $k$-uniform Families}\label{Extremal}
Extremal combinatorics deals with the problem of determining or estimating the maximum or minimum cardinality of a collection of finite objects that satisfies some specific set of requirements. It is also concerned with the investigation of inequalities between combinatorial invariants, and questions dealing with relations among them. For an introduction to the topic, we refer the interested reader to the books by Jukna \cite{Stas[11]} and Bollob\'{a}s \cite{BollB[78]}, and the surveys by Alon \cite{Alon1[03],Alon1[08],Alon1[16],Alon1[20]}. Extremal combinatorics can be further divided into the following distinct fields:
\begin{itemize}
	\item Extremal graph theory, which began with the work of Mantel in 1907 \cite{Aign[95]} and was first investigated in earnest by Tur\'{a}n in 1941 \cite{Tura[41]}. For a survey of the important results in the field, see \cite{Niki[11]}.
	\item Ramsey theory, which was popularised by Erd\H{o}s and Szekeres \cite{Erd[47],ErdGeo[35]} by extending a result of Ramsey from 1929 (published in 1930 \cite{FrankRam[30]}). For a survey of the important results in the field, see \cite{JacFox[15]}.
	\item Extremal problems in arithemetic combinatorics, which grew from the work of van der Waerden in 1927 \cite{BL[27]} and the Erd\H{o}s-Tur\'{a}n conjecture of 1936 \cite{PLPL[36]}. For a survey of the important results in the field, see \cite{Choon[12]}.
	\item Extremal (finite) set theory, which was first investigated by Sperner~\cite{Sperner[28]} in 1928 by establishing the maximum size of an antichain, i.e., a set-system where no member is a superset of another. However, it was Erd\H{o}s et al. \cite{Erdos[61]} who started systematic research in extremal set theory. 
\end{itemize}

Extremal set theory deals with determining the size of set-systems that satisfy certain restrictions. It is one of the most rapidly developing areas in combinatorics, with applications in various other branches of mathematics and theoretical computer science, including functional analysis, probability theory, circuit complexity, cryptography, coding theory, probabilistic methods, discrete geometry, linear algebra, spectral graph theory, ergodic theory, and harmonic analysis \cite{Beimel[15],Beimel[12],Zeev[15],Zeev[11],Klim[09],Sergey[08],Liu[17],SehrawatVipin[21],VipinYvo[20],Polak[13],Blackburn[03],WangThesis[20],Sudak[10],GarVac[94],Gro[00],IWF[20]}. For more details on extremal set theory, we refer the reader to the book by Gerbner and Patkos \cite{GerbBala[18]}; for probabilistic arguments/proofs, see the books by Bollob\'{a}s \cite{Boll[86]} and Spencer \cite{JSpen[87]}. 

Our work in this paper concerns a subfield of extremal set theory, called \textit{intersection theorems}, wherein set-systems under specific intersection restrictions are constructed, and bounds on their sizes are derived. A wide range of methods have been employed to establish a large number of intersection theorems over various mathematical structures, including vector subspaces, graphs, subsets of finite groups with given group actions, and uniform hypergraphs with stronger or weaker intersection conditions. The methods used to derive these theorems have included purely combinatorial methods such as shifting/compressions, algebraic methods (including linear-algebraic, Fourier analytic and representation-theoretic), analytic, probabilistic and regularity-type methods. We shall not give a full account of the known intersection theorems, but only touch upon the results that are particularly relevant to our set-system and its construction. For a broader account, we refer the interested reader to the comprehensive surveys by Ellis \cite{Ell[21]}, and Frankl and Tokushige~\cite{Frankl[16]}. For an introduction to intersecting and cross-intersecting families related to hypergraphs, see~\cite{AMDD[2020],Kleit[79]}. 

\begin{note}
	Set-system and hypergraph are very closely related terms, and commonly used interchangeably. Philosophically, in a hypergraph, the focus is more on vertices, vertex subsets being in ``relation'', and subset(s) of vertices satisfying a specific configuration of relations; whereas in a set-system, the focus is more on set-theoretic properties of the sets.
\end{note}

In this section, we derive multiple intersection theorems for:
\begin{enumerate}
	\item at most $t$-intersecting $k$-uniform families of sets,
	\item maximally cover-free at most $t$-intersecting $k$-uniform families of sets.
\end{enumerate}
We also provide an explicit construction for at most $t$-intersecting $k$-uniform families of sets. Later in the text, we use the results from this section to establish the maximum number of SSKH PRFs that can be  constructed securely by a set of parties against various active/passive and internal/external adversaries.

For $a,b\in\mathbb{Z}$ with $a\leq b$, let $[a,b]:=\{a,a+1,\ldots,b-1,b\}$.  

\begin{definition}
	\emph{$\mathcal{H}\subseteq 2^{[n]}$ is $k$-uniform if $|A|=k$ for all $A\in\mathcal{H}$.}
\end{definition}

\begin{definition}
	\emph{$\mathcal{H}\subseteq 2^{[n]}$ is maximally cover-free if
	$$A\not\subseteq\bigcup_{B\in\mathcal{H}, B\neq A}B$$
	for all $A\in\mathcal{H}$.}
\end{definition}

It is clear that $\mathcal{H}\subseteq 2^{[n]}$ is maximally cover-free if and only if every $A\in\mathcal{H}$ has some element $x_A$ such that $x_A\not\in B$ for all $B\in\mathcal{H}$ where $B\neq A$. Furthermore, the maximum size of a $k$-uniform family $\mathcal{H}\subseteq 2^{[n]}$ that is maximally cover-free is $n-k+1$, and it is realized by the following set system:
$$\mathcal{H}=\left\{[k-1]\cup\{x\}: x\in[k, n]\right\}$$
(and this is unique up to permutations of $[n]$).

\begin{definition}
	\label{intersecting_definitions}
	\emph{Let $t$ be a non-negative integer. We say the set system $\mathcal{H}$ is
	\begin{enumerate}[label = {(\roman*)}]
		\item at most $t$-intersecting if $|A\cap B|\leq t$,
		\item exactly $t$-intersecting if $|A\cap B|=t$,
		\item at least $t$-intersecting if $|A\cap B|\geq t$,
	\end{enumerate}
	for all $A, B\in\mathcal{H}$ with $A\neq B$.}
\end{definition}

Property (iii) in \Cref{intersecting_definitions} is often simply called ``$t$-intersecting'' \cite{Borg[11]}, but we shall use the term ``at least $t$-intersecting'' for clarity.

\begin{definition}
\emph{Let $\mathcal{F},\mathcal{G}\subseteq 2^{[n]}$. We say that $\mathcal{F}$ and $\mathcal{G}$ are equivalent (denoted as $\mathcal{F}\sim\mathcal{G}$) if there exists a permutation $\pi$ of $[n]$ such that $\pi^\ast(\mathcal{F})=\mathcal{G}$, where
$$\pi^\ast(\mathcal{F})=\left\{\{\pi(a) :a\in A\}: A\in\mathcal{F}\right\}.$$}
\end{definition}

For $n,k,t,m\in\mathbb{Z}^+$ with $t\leq k\leq n$, let $N(n,k,t,m)$ denote the collection of all set systems $\mathcal{H}\subseteq 2^{[n]}$ of size $m$ that are at most $t$-intersecting and $k$-uniform, and $M(n,k,t,m)$ denote the collection of set systems $\mathcal{H}\in N(n,k,t,m)$ that are also maximally cover-free.

The following proposition establishes a bijection between equivalence classes of these two collections of set systems (for different parameters):

\begin{proposition}
	\label{maximally_cover_free}
	Suppose $n,k,t,m\in\mathbb{Z}^+$ satisfy $t\leq k\leq n$ and $m<n$. Then there exists a bijection
	$$M(n,k,t,m)\ /\sim\ \leftrightarrow N(n-m,k-1,t,m)\ /\sim.$$
\end{proposition}

\begin{proof}
	We will define functions
	\begin{align*}
		\bar{f}:&& M(n,k,t,m)\ /\sim\ &\to N(n-m,k-1,t,m)\ /\sim \\
		\bar{g}:&& N(n-m,k-1,t,m)\ /\sim\ &\to M(n,k,t,m)\ /\sim
	\end{align*}
	that are inverses of each other.
	
	Let $\mathcal{H}\in M(n,k,t,m)$. Since $\mathcal{H}$ is maximally cover-free, for every $A\in\mathcal{H}$, there exists $x_A\in A$ such that $x_A\not\in B$ for all $B\in\mathcal{H}$ where $B\neq A$. Consider the set system $\{A\setminus\{x_A\}: A\in\mathcal{H}\}$.
	First, note that although this set system depends on the choice of $x_A\in A$ for each $A\in\mathcal{H}$, the equivalence class of $\{A\setminus\{x_A\}: A\in\mathcal{H}\}$ is independent of this choice. Hence, we get the following map:
	\begin{align*}
		f:M(n,k,t,m)&\to N(n-m,k-1,t,m)\ /\sim \\
		\mathcal{H}&\mapsto [\{A\setminus\{x_A\}: A\in\mathcal{H}\}].
	\end{align*}
	Furthermore, it is clear that that if $\mathcal{H}\sim\mathcal{H}'$, then $f(\mathcal{H})\sim f(\mathcal{H}')$. So, $f$ induces a well-defined map as:
	$$\bar{f}: M(n,k,t,m)\ /\sim\ \to N(n-m,k-1,t,m)\ /\sim.$$
	
	Next, for a set system $\mathcal{G}=\{G_1,\ldots, G_m\}\in N(n-m,k-1,t,m)$, define
	\begin{align*}
		g: N(n-m,k-1,t,m) &\to M(n,k,t,m)\ /\sim \\
		\mathcal{G} &\mapsto [\{G_i\cup\{n-m+i\}:i\in[m]\}].
	\end{align*}
	Again, this induces a well-defined map
	$$\bar{g}: N(n-m,k-1,t,m)\ /\sim\ \to M(n,k,t,m)\ /\sim$$
	since $g(\mathcal{G})\sim g(\mathcal{G}')$ for any $\mathcal{G}$, $\mathcal{G}'$ such that $\mathcal{G}\sim\mathcal{G}'$.
	
	We can check that $\bar{f}\circ \bar{g}=id_{N(n-m,k-1,t,m)}$ and that $\bar{g}\circ \bar{f}=id_{M(n,k,t,m)}$. Hence, $\bar{f}$ and $\bar{g}$ are bijections. \qed
\end{proof}

\begin{corollary}
	\label{maximally_cover_free_corollary}
	Let $n,k,t,m\in\mathbb{Z}^+$ be such that $t\leq k\leq n$ and $m<n$. Then there exists a maximally cover-free, at most $t$-intersecting, $k$-uniform set system $\mathcal{H}\subseteq 2^{[n]}$ of size $m$ if and only if there exists an at most $t$-intersecting, $(k-1)$-uniform set system $\mathcal{G}\subseteq 2^{[n-m]}$.
\end{corollary}

\begin{remark}
	Both Proposition \ref{maximally_cover_free} and Corollary \ref{maximally_cover_free_corollary} remain true if, instead of at most $t$-intersecting families, we consider exactly $t$-intersecting or at least $t$-intersecting families.
\end{remark}

At least $t$-intersecting families have been completely characterized by Ahlswede and Khachatrian \cite{Ahlswede[97]}, but the characterization of exactly $t$-intersecting and at most $t$-intersecting families remain open.

Let $\varpi(n,k,t)=\max\left\{|\mathcal{H}|: \mathcal{H}\subseteq 2^{[n]}\text{ is at most }t\text{-intersecting and }k\text{-uniform}\right\}.$

\begin{proposition}
	\label{simple_bound}
	Suppose $n,k,t\in\mathbb{Z}^+$ are such that $t\leq k\leq n$. Then
	$$\varpi(n,k,t)\leq\frac{\binom{n}{t+1}}{\binom{k}{t+1}}.$$
\end{proposition}
\begin{proof}
	Let $\mathcal{H}\subseteq 2^{[n]}$ be an at most $t$-intersecting and $k$-uniform family. The number of pairs $(X, A)$, where $A\in\mathcal{H}$ and $X\subseteq A$ is of size $t+1$, is equal to $|\mathcal{H}|\cdot\binom{k}{t+1}$.
	Since $\mathcal{H}$ is at most $t$-intersecting, any $(t+1)$-element subset of $[n]$ lies in at most one set in $\mathcal{H}$. Thus,
	$$|\mathcal{H}|\cdot\binom{k}{t+1}\leq\binom{n}{t+1}\implies |\mathcal{H}|\leq\frac{\binom{n}{t+1}}{\binom{k}{t+1}}.$$ \qed
\end{proof}

Using Proposition \ref{maximally_cover_free}, we immediately obtain the following as a corollary:

\begin{corollary}
	\label{simple_bound_corollary}
	Suppose $\mathcal{H}\subseteq 2^{[n]}$ is maximally cover-free, at most $t$-intersecting and $k$-uniform. Then
	$$|\mathcal{H}|\leq\frac{\binom{n-|\mathcal{H}|}{t+1}}{\binom{k-1}{t+1}}.$$
\end{corollary}

Similarly, by applying Proposition \ref{maximally_cover_free}, other results on at most $t$-intersecting and $k$-uniform set systems can also be translated into results on set systems that, in addition to having these two properties, are maximally cover-free. Thus, henceforth, we do not explicitly state such results when their derivation is trivial.

\subsection{Bounds for Small $n$}
In this section, we give several bounds on $\varpi(n,k,t)$ for small values of $n$.

\begin{lemma}
	\label{bound_for_small_n_lemma}
	Let $n,k,t\in\mathbb{Z}^+$ be such that $t\leq k\leq n<\frac{1}{2}k\left(\frac{k}{t}+1\right)$. Let $m'$ be the least positive integer such that $n<m'k-\frac{1}{2}m'(m'-1)t$. Then
	$$\varpi(n,k,t)=m'-1.$$
\end{lemma}
\begin{proof}
	First, we show that there exists $m^\star\in\mathbb{Z}^+$ such that $n<m^\star k-\frac{1}{2}m^\star(m^\star-1)t$. Consider the quadratic polynomial $p(x)=xk-\frac{1}{2}x(x-1)t$. Note that $p(x)$ achieves its maximum value at $x=\frac{k}{t}+\frac{1}{2}$. If we let $m^\star$ be the unique positive integer such that $\frac{k}{t}\leq m^\star<\frac{k}{t}+1$, then
	$$p(m^\star)\geq p\left(\frac{k}{t}\right)=\frac{1}{2}k\left(\frac{k}{t}+1\right)>n,$$
	as required.
	
	Next, suppose $\mathcal{H}$ is an at most $t$-intersecting, $k$-uniform set family with $|\mathcal{H}|\geq m'$. Let $A_1,\ldots, A_{m'}\in\mathcal{H}$ be distinct. Then
	$$n\geq\left|\bigcup_{i=1}^{m'} A_i\right|=\sum_{i=1}^{m'} \left|A_i\setminus\bigcup_{j=1}^{i-1}A_j\right|\geq\sum_{i=0}^{m'-1} (k-it)=m'k-\frac{1}{2}m'(m'-1)t,$$
	which is a contradiction. This proves that $|\mathcal{H}|\leq m'-1$.
	
	It remains to construct an at most $t$-intersecting, $k$-uniform set family $\mathcal{H}\subseteq 2^{[n]}$ with $|\mathcal{H}|=m'-1$. Let $m=m'-1$. The statement is trivial if $m=0$, so we may assume that $m\in\mathbb{Z}^+$. By the minimality of $m'$, we must have $n\geq mk-\frac{1}{2}m(m-1)t$. Let $k=\alpha t+\beta$ with $\alpha,\beta\in\mathbb{Z}$ and $0\leq\beta\leq t-1$. Define a set system $\mathcal{H}=\{A_1,\ldots,A_m\}$ as follows:
	$$A_i=\left\{(l,\{i,j\}):l\in[t], j\in[\alpha+1]\setminus\{i\}\right\}\cup\left\{(i,j):j\in[\beta]\right\}.$$
	It is clear, by construction, that $\mathcal{H}$ is at most $t$-intersecting and $k$-uniform. Furthermore, since $\alpha=\lfloor k/t\rfloor\geq m$, the number of elements in the universe of $\mathcal{H}$ is
	\begin{align*}
		&t\cdot\left|\{i,j\}:1\leq i<j\leq\alpha+1,\,i\leq m\right|+m\beta \\
		=\ &t\left(\binom{\alpha+1}{2}-\binom{\alpha+1-m}{2}\right)+m\beta \\
		=\ &t\left(m\alpha-\frac{1}{2}m(m-1)\right)+m\beta \\
		=\ &mk-\frac{1}{2}m(m-1)t.
	\end{align*} \qed
\end{proof}

\begin{proposition}
	\label{bound_for_small_n}
	Let $n,k,t\in\mathbb{Z}^+$ be such that $t\leq k\leq n$. 
	\begin{enumerate}[label = {(\alph*)}]
		\item If $n<\frac{1}{2}k\left(\frac{k}{t}+1\right)$, then
		$$\varpi(n,k,t)=\left\lfloor\frac{1}{2}+\frac{k}{t}-\sqrt{\left(\frac{1}{2}+\frac{k}{t}\right)^2-\frac{2n}{t}}\right\rfloor.$$
		\item If $t\mid k$ and $n=\frac{1}{2}k\left(\frac{k}{t}+1\right)$, then
		$$\varpi(n,k,t)=\frac{k}{t}+1.$$
	\end{enumerate}
\end{proposition}
\begin{proof}
	\begin{enumerate}[label = {(\alph*)}]
		\item Note that $m=\left\lfloor\frac{1}{2}+\frac{k}{t}-\sqrt{\left(\frac{1}{2}+\frac{k}{t}\right)^2-\frac{2n}{t}}\right\rfloor$ satisfies $n\geq mk-\frac{1}{2}m(m-1)t$ and $m'=m+1$ satisfies $n< m'k-\frac{1}{2}m'(m'-1)t$; hence, the result follows immediately from Lemma \ref{bound_for_small_n_lemma}.
		\item Let $\mathcal{H}\subseteq 2^{[n]}$ be an at most $t$-intersecting, $k$-uniform set family. We may assume that $|\mathcal{H}|\geq\frac{k}{t}$. We will first show that any three distinct sets in $\mathcal{H}$ have empty intersection. Let $A_1$, $A_2$ and $A_3$ be any three distinct sets in $\mathcal{H}$, and let $A_4,\,\ldots,\,A_{\frac{k}{t}}\in\mathcal{H}$ be such that the $A_i$'s are all distinct. Then
		$$\left|\bigcup_{i=1}^{\frac{k}{t}} A_i\right|=\sum_{i=1}^{\frac{k}{t}} \left|A_i\setminus\bigcup_{j=1}^{i-1}A_j\right|\geq\sum_{i=0}^{\frac{k}{t}-1} (k-it)=\frac{1}{2}k\left(\frac{k}{t}+1\right)=n,$$
		and thus we have must equality everywhere. In particular, we obtain $|A_3\setminus (A_1\cup A_2)|=k-2t$, which together which the fact that $\mathcal{H}$ is at most $t$-intersecting, implies that $A_1\cap A_2\cap A_3=\varnothing$, as claimed. Therefore, every $x\in[n]$ lies in at most $2$ sets in $\mathcal{H}$. Now,
		$$|\mathcal{H}|\cdot k=|(A,x):A\in\mathcal{H},\ x\in A|\leq 2n\implies |\mathcal{H}|\leq\frac{2n}{k}=\frac{k}{t}+1,$$
		proving the first statement.
		
		Next, we shall exhibit an at most $t$-intersecting, $k$-uniform set family $\mathcal{H}\subseteq 2^{[n]}$, where $n=\frac{1}{2}k\left(\frac{k}{t}+1\right)$, with $|\mathcal{H}|=\frac{k}{t}+1$. Let $\mathcal{H}=\{A_1,\ldots,A_{\frac{k}{t}+1}\}$ with
		$$A_i=\left\{(l,\{i,j\}):l\in[t], j\in\left[\frac{k}{t}+1\right]\setminus\{i\}\right\}.$$
		It is clear that $\mathcal{H}$ is exactly $t$-intersecting and $k$-uniform, and that it is defined over a universe of $t\cdot\dbinom{k/t+1}{2}=n$ elements.  \qed
	\end{enumerate}
\end{proof}

\begin{remark}
	The condition $n<\frac{1}{2}k\left(\frac{k}{t}+1\right)$ in Proposition \ref{bound_for_small_n}(a) is necessary. Indeed, if $n=\dfrac{1}{2}k\left(\dfrac{k}{t}+1\right)$, then $$\left\lfloor\frac{1}{2}+\frac{k}{t}-\sqrt{\left(\frac{1}{2}+\frac{k}{t}\right)^2-\frac{2n}{t}}\right\rfloor=\frac{k}{t}<\frac{k}{t}+1.$$
\end{remark}

Next, we examine the case where $n=\frac{1}{2}k\left(\frac{k}{t}+1\right)+1$. Unlike earlier cases, we do not have exact bounds for this case. But what is perhaps surprising is that, for certain $k$ and $t$, the addition of a single element to the universe set can increase the maximum size of the set family by $3$ or more.

\begin{proposition}
	\label{bound_beyond_small_n}
	Let $n,k,t\in\mathbb{Z}^+$ be such that $t\leq k\leq n$ and $t\mid k$. If $n=\frac{1}{2}k\left(\frac{k}{t}+1\right)+1$, then
	$$\varpi(n,k,t)\leq\frac{\frac{k}{t}+1}{1-\frac{k}{n}}=\left(\frac{k^2+kt+2t}{k^2-kt+2t}\right)\left(\frac{k}{t}+1\right).$$
\end{proposition}
\begin{proof}
	Let $\mathcal{H}\subseteq 2^{[n]}$ be an at most $t$-intersecting and $k$-uniform family. There exists some element $x\in[n]$ such that $x$ is contained in at most $\lfloor\frac{k|\mathcal{H}|}{n}\rfloor$ sets in $\mathcal{H}$. We construct a set family $\mathcal{H}'\subseteq 2^{[n]\setminus\{x\}}$ by taking those sets in $\mathcal{H}$ that do not contain $x$. Since $\mathcal{H}'$ is defined over a universe of $\frac{1}{2}k\left(\frac{k}{t}+1\right)$ elements, we obtain the following by applying Proposition \ref{bound_for_small_n}:
	\begin{align*}
		|\mathcal{H}|-\left\lfloor\frac{k|\mathcal{H}|}{n}\right\rfloor\leq|\mathcal{H}'|\leq\frac{k}{t}+1&\implies\left\lceil|\mathcal{H}|-\frac{k|\mathcal{H}|}{n}\right\rceil\leq\frac{k}{t}+1 \\
		&\implies|\mathcal{H}|-\frac{k|\mathcal{H}|}{n}\leq\frac{k}{t}+1 \\
		&\implies|\mathcal{H}|\leq\frac{\frac{k}{t}+1}{1-\frac{k}{n}}.
	\end{align*}\qed
\end{proof}

\begin{remark}\label{FanoRemark}
	\begin{enumerate}[label = {(\alph*)}]
		\item If $k=3$, $t=1$, and $n=\frac{1}{2}k\left(\frac{k}{t}+1\right)+1=7$, then the bound in the Proposition \ref{bound_beyond_small_n} states that $\varpi(n,k,t)\leq\left(\frac{k^2+kt+2t}{k^2-kt+2t}\right)\left(\frac{k}{t}+1\right)=7$. The Fano plane, depicted in \Cref{fano_plane}, is an example of a $3$-uniform family of size $7$, defined over a universe of $7$ elements, that is exactly $1$-intersecting. Thus, the bound in \Cref{bound_beyond_small_n} can be achieved, at least for certain choices of $k$ and $t$. An interesting side note: Fano plane has applications/relations to integer factorization \cite{GuyR[75],Leh2[74],Shanks[85]} and octonians \cite{HansFre[85],MicMJ[22],MicMar[22],Baez[02]}, both of which have direct applications to cryptography \cite{Wmull[16],Zake[10],Lipi[21],Neal[00],SongYan[08],RSA[78]}. 
		\begin{figure}[H]
			\centering
			\includegraphics[scale=.5]{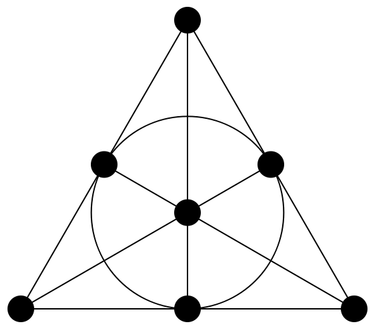}
			\caption{The Fano plane}
			\label{fano_plane}
		\end{figure}
		\item Note that
		$$\left(\frac{k^2+kt+2t}{k^2-kt+2t}\right)\left(\frac{k}{t}+1\right)-\left(\frac{k}{t}+1\right)=\frac{2k^2+2kt}{k^2-kt+2t}.$$
		We can show that the above expression is (strictly) bounded above by $6$ (for $k\neq t$), with slightly better bounds for $t=1,\,2,\,3,\,4$. It follows that
		$$\varpi(n,k,t)\leq\begin{cases}
			\frac{k}{t}+4 & \text{if }t=1, \\
			\frac{k}{t}+5 & \text{if }t=2,\,3,\,4, \\ 
			\frac{k}{t}+6 & \text{if }t\geq 5.
		\end{cases}
		$$
		Furthermore, $\lim_{k\rightarrow\infty}\frac{2k^2+2kt}{k^2-kt+2t}=2$; thus, for fixed $t$, we have $\varpi(n,k,t)\leq\frac{k}{t}+3$ for large enough $k$.
	\end{enumerate}
\end{remark}

Next, we give a necessary condition for the existence of at most $t$-intersecting and $k$-uniform families $\mathcal{H}\subseteq 2^{[n]}$, which implicitly gives a bound on $\varpi(n,k,t)$.

\begin{proposition}
	\label{larger_n}
	Let $n,k,t\in\mathbb{Z}^+$ satisfy $t\leq k\leq n$, and $\mathcal{H}\subseteq 2^{[n]}$ be an at most $t$-intersecting and $k$-uniform family with $|\mathcal{H}|=m$. Then
	$$(n-r)\left\lfloor\frac{km}{n}\right\rfloor^2+r\left\lceil\frac{km}{n}\right\rceil^2\leq (k-t)m+tm^2$$
	where $r=km-n\left\lfloor\frac{km}{n}\right\rfloor$.
\end{proposition}
\begin{proof}
	Let $\alpha_j$ be the number of elements that is contained in exactly $j$ sets in $\mathcal{H}$. We claim that the following holds: 
	\begin{align}
		\sum_{j=0}^{m}\alpha_j&=n, \label{eq_1} \\
		\sum_{j=0}^{m}j\alpha_j&=km, \label{eq_2} \\
		\sum_{j=0}^{m}j(j-1)\alpha_j&\leq tm(m-1) \label{eq_3}.
	\end{align}
	(\ref{eq_1}) is immediate, (\ref{eq_2}) follows from double counting the set $\{(A,x):A\in\mathcal{H},\ x\in A\}$, and (\ref{eq_3}) follows from considering $\{(A,B,x):A,B\in\mathcal{H},\ A\neq B,\ x\in A\cap B\}$ and using the fact that $\mathcal{H}$ is at most $t$-intersecting. This proves the claim.
	
	Next, let us find non-negative integer values of $\alpha_0,\ldots,\alpha_m$ satisfying both (\ref{eq_1}) and (\ref{eq_2}) that minimize the expression $\sum_{j=0}^{m}j(j-1)\alpha_j$. Note that
	$$\sum_{j=0}^{m}j(j-1)\alpha_j=\sum_{j=0}^{m}(j^2\alpha_j-j\alpha_j)=\sum_{j=0}^{m}j^2\alpha_j-km.$$
	So, we want to minimize $\sum_{j=0}^{m}j^2\alpha_j$, subject to the restrictions (\ref{eq_1}) and (\ref{eq_2}). If $n\nmid km$, this is achieved by letting $\alpha_{\lfloor\frac{km}{n}\rfloor}=n-r$ and $\alpha_{\lceil\frac{km}{n}\rceil}=r$, with all other $\alpha_j$'s equal to $0$. If $n\mid km$, we let $\alpha_{\frac{km}{n}}=n$ with all other $\alpha_j$'s equal to $0$.
	
	Indeed, it is easy to see that the above choice of $\alpha_0,\ldots,\alpha_m$ satisfy both (\ref{eq_1}) and (\ref{eq_2}). Now, let $\alpha_0,\ldots,\alpha_m$ be some other choice of the $\alpha_j$'s that also satisfy both (\ref{eq_1}) and (\ref{eq_2}). We will show that the function $f(\alpha_0,\ldots,\alpha_m)=\sum_{j=0}^{m}j^2\alpha_j$ can be decreased with a different choice of $\alpha_0,\ldots,\alpha_m$.
	
	Suppose $\alpha_i\neq 0$ for some $i\neq\lfloor\frac{km}{n}\rfloor,\lceil\frac{km}{n}\rceil$, and assume that $i<\lfloor\frac{km}{n}\rfloor$ (the other case where $i>\lceil\frac{km}{n}\rceil$ is similar). Since the $\alpha_j$'s satisfy both (\ref{eq_1}) and (\ref{eq_2}), there must be some $i_1$ with $i_1\geq\lceil\frac{km}{n}\rceil$ (the inequality is strict if $n\mid km$) such that $\alpha_{i_1}\neq 0$. 	It follows that if we decrease $\alpha_i$ and $\alpha_{i_1}$ each by one, and increase $\alpha_{i+1}$ and $\alpha_{i_1-1}$ each by one, constraints (\ref{eq_1}) and (\ref{eq_2}) continue to be satisfied. Furthermore, considering that $i_1\geq\lfloor\frac{km}{n}\rfloor+1>i+1,$ we get:
	\begin{align*}
		f(&\alpha_1,\ldots,\alpha_i,\alpha_{i+1},\ldots,\alpha_{i_1-1},\alpha_{i_1},\ldots,\alpha_m)-f(\alpha_1,\ldots,\alpha_i-1,\alpha_{i+1}+1,\ldots,\alpha_{i_1-1}+1,\alpha_{i_1}-1,\ldots,\alpha_m) \\
		&=\ i^2-(i+1)^2-(i_1-1)^2+i_1^2 = 2i_1-2i-2>0.
	\end{align*}
	This proves the claim that the choice of $\alpha_{\lfloor\frac{km}{n}\rfloor}=n-r$ and $\alpha_{\lceil\frac{km}{n}\rceil}=r$ minimizes $f$.
	
	Therefore, we can only find non-negative integers $\alpha_0,\ldots,\alpha_m$ satisfying all three conditions above if and only if
	$$(n-r)\left\lfloor\frac{km}{n}\right\rfloor^2+r\left\lceil\frac{km}{n}\right\rceil^2-km\leq tm(m-1),$$
	as desired.  \qed
\end{proof}

\begin{remark}
	For fixed $k$ and $t$, if $n$ is sufficiently large, then the inequality in Proposition \ref{larger_n} will be true for all $m$. Thus, the above proposition is only interesting for values of $n$ that are not too large.
\end{remark}

\subsection{Asymptotic Bounds}
The study of Steiner systems has a long history, dating back to the 19th century work on triple block designs by Pl\"{u}cker \cite{Plucker[35]}, Kirkman \cite{Kirkman[47]}, Steiner \cite{Steiner[53]}, Reiss \cite{Reiss[59]}, Noether \cite{Noet[79]}, Netto \cite{Netto[93]}, Moore \cite{Moore[93]}, Sylvester \cite{Syl[93]}, Power \cite{Power[67]}, and Clebsch and Lindemann \cite{Ferdi[76]}. The term Steiner (triple) systems was coined in 1938 by Witt \cite{Witt[38]}. A Steiner system is defined as an arrangement of a set of elements in triples such that each pair of elements is contained in exactly one triple. The Fano plane --- discussed in \Cref{FanoRemark} and depicted in \Cref{fano_plane} --- represents a unique Steiner system of order 7; it can be seen as having 7 elements in 7 blocks of size 3 such that each pair of elements is contained in exactly 1 block. Steiner systems have strong connections to a wide range of topics, including statistics, linear coding, finite group theory, finite geometry, combinatorial design, experimental design, storage systems design, wireless communication, low-density parity-check code design, distributed storage, and batch codes. In cryptography, Steiner systems primarily have applications to (anonymous) secret sharing \cite{Stinson[87],Blundo[96]} and low-redundancy private information retrieval \cite{Fazeli[15]}. For a broader introduction to the topic, we refer the interested reader to~\cite{Wilson[03]} (also see \cite{Tri[99],Charles[06]}). 

In this section, we will see how at most $t$-intersecting families are related to Steiner systems. Using a result from Keevash \cite{Keevash[19]} about the existence of Steiner systems with certain parameters, we obtain an asymptotic bound on the maximum size of at most $t$-intersecting families.

\begin{definition}
	\emph{A Steiner system $S(t,k,n)$, where $t\leq k\leq n$, is a family $\mathcal{S}$ of subsets of $[n]$ such that
	\begin{enumerate}
		\item $|A|=k$ for all $A\in\mathcal{S}$,
		\item any $t$-element subset of $[n]$ is contained in exactly one set in $\mathcal{S}$.
	\end{enumerate}
	The elements of $\mathcal{S}$ are known as blocks.}
\end{definition}

From the above definition, it is clear that there exists a family $\mathcal{H}$ that achieves equality in Proposition \ref{simple_bound} if and only if $S(t+1,k,n)$ exists. It is easy to derive the following well known necessary condition for the existence of a Steiner system with given parameters:

\begin{proposition}
	If $S(t,k,n)$ exists, then $\binom{n-i}{t-i}$ is divisible by $\binom{k-i}{t-i}$ for all $0\leq i\leq t$, and the number of blocks in $S(t,k,n)$ is equal to $\binom{n}{t}/\binom{k}{t}$.
\end{proposition}

In 2019, Keevash \cite{Keevash[19]} proved the following result, providing a partial converse to the above, and answering in the affirmative a longstanding open problem in the theory of designs.

\begin{theorem}[\cite{Keevash[19]}]
	\label{keevash}
	For any $k,t\in\mathbb{Z}^+$ with $t\leq k$, there exists $n_0(k,t)$ such that for all $n\geq n_0(k,t)$, a Steiner system $S(t,k,n)$ exists if and only if
	$$\binom{k-i}{t-i}\text{ divides }\binom{n-i}{t-i}\quad\text{for all }i=0,1\ldots,t-1.$$
\end{theorem}

Using this result, we will derive asymptotic bounds for the maximum size of an at most $t$-intersecting and $k$-uniform family.

\begin{proposition}
	\label{asymptotic_bound}
	Let $k,t\in\mathbb{Z}^+$ with $t<k$, and $C<1$ be any positive real number.
	\begin{enumerate}[label = {(\roman*)}]
	\item \label{i} There exists $n_1(k,t,C)$ such that for all integers $n\geq n_1(k,t,C)$, there is an at most $t$-intersecting and $k$-uniform family $\mathcal{H}\subseteq 2^{[n]}$ with
	$$|\mathcal{H}|\geq\frac{Cn^{t+1}}{k(k-1)\cdots(k-t)}.$$
	\item For all sufficiently large $n$,
	$$\frac{Cn^{t+1}}{k(k-1)\cdots(k-t)}\leq\varpi(n,k,t)<\frac{n^{t+1}}{k(k-1)\cdots(k-t)}.$$
	In particular,
	$$\varpi(n,k,t)\sim\frac{n^{t+1}}{k(k-1)\cdots(k-t)}.$$
	\end{enumerate}
\end{proposition}

\begin{proof}
	\begin{enumerate}[label = {(\roman*)}]
	\item Let $t'=t+1$. By Theorem \ref{keevash}, there exists $n_0(k,t')$ such that for all $N\geq n_0(k,t')$, a Steiner system $S(t',k,N)$ exists if
	\begin{equation}
		\tag{$\ast$}
		\label{keevash_condition}
		\binom{k-i}{t'-i}\text{ divides }\binom{N-i}{t'-i}\quad\text{for all }i=0,1\ldots,t'-1.
	\end{equation}
	
	Suppose $n$ is sufficiently large. Let $n'\leq n$ be the largest integer such that (\ref{keevash_condition}) is satisfied with $N=n'$. Since
	\begin{align*}
	&\binom{k-i}{t'-i}\text{ divides }\binom{N-i}{t'-i} \\
	\iff\ &(k-i)\cdots(k-t'+1)\mid(N-i)\cdots(N-t'+1),
	\end{align*}
	all $N$ of the form $\lambda k(k-1)\cdots(k-t'+1)+t'-1$ with $\lambda\in\mathbb{Z}$ will satisfy (\ref{keevash_condition}). Hence,
	$n-n'\leq k(k-1)\cdots(k-t'+1)$.
	
	By our choice of $n'$, there exists a Steiner system $S(t',k,n')$, which is an at most $t$-intersecting and $k$-uniform set family, defined over the universe $[n']\subseteq [n]$, such that
	\begin{align*}
		|S(t',k,n')|&=\frac{\binom{n'}{t'}}{\binom{k}{t'}}=\frac{n'(n'-1)\cdots(n'-t'+1)}{k(k-1)\cdots(k-t'+1)} \\
		&\geq\frac{(n-\alpha)(n-\alpha-1)\cdots(n-\alpha-t'+1)}{k(k-1)\cdots(k-t'+1)},
	\end{align*}
	where $\alpha=\alpha(k,t')=k(k-1)\cdots(k-t'+1)$ is independent of $n$. Since $C<1$, there exists $n_2(k,t',C)$ such that for all $n\geq n_2(k,t',C)$,
	$$\frac{(n-\alpha)(n-\alpha-1)\cdots(n-\alpha-t'+1)}{n^{t'}}\geq C,$$
	from which it follows that
	$$|S(t',k,n')|\geq\frac{Cn^{t'}}{k(k-1)\cdots(k-t'+1)}=\frac{Cn^{t+1}}{k(k-1)\cdots(k-t)}$$
	for all sufficiently large $n$. From the above argument, we see that we can pick $$n_1(k,t,C)=\max\left(n_0(k,t')+\alpha(k,t'), n_2(k,t',C)\right). $$
	
	\item 	By Proposition \ref{simple_bound},
	$$\varpi(n,k,t)\leq\frac{\binom{n}{t+1}}{\binom{k}{t+1}}=\frac{n(n-1)\cdots(n-t)}{k(k-1)\cdots(k-t)}<\frac{n^{t+1}}{k(k-1)\cdots(k-t)}.$$
	The other half of the inequality follows immediately from \ref{i}.  \qed
	\end{enumerate}
\end{proof}

\begin{proposition}
	\label{asymptotic_bound_for_maximally_cover_free}
	Let $k,t\in\mathbb{Z}^+$ with $t<k-1$, and $C<1$ be any positive real number. Then for all sufficiently large $N$,
	\begin{enumerate}[label = {(\roman*)}]
		\item there exists a maximally cover-free, at most $t$-intersecting and $k$-uniform family $\mathcal{H}\subseteq 2^{[N]}$ with $|\mathcal{H}|\geq CN$,
		\item the maximum size $\nu(N,k,t)$ of a maximally cover-free, at most $t$-intersecting and $k$-uniform family $\mathcal{H}\subseteq 2^{[N]}$ satisfies
		$$CN\leq\nu(N,k,t)<N.$$
	\end{enumerate}
\end{proposition}

\begin{proof}
	We note that (ii) follows almost immediately from (i). So, we prove (i).
	
	Fix $C_0$ such that $C<C_0<1$. It follows from Propositions \ref{maximally_cover_free} and \ref{asymptotic_bound} that for all integers $n\geq n_1(k-1,t,C_0)$, there exists a maximally cover-free, at most $t$-intersecting and $k$-uniform family $\mathcal{H}\subseteq 2^{\left[n+\frac{n^{t+1}}{(k-1)(k-2)\cdots(k-t-1)}\right]}$ with
	$$|\mathcal{H}|\geq\frac{C_0n^{t+1}}{(k-1)(k-2)\cdots(k-t-1)}.$$
	
	Since $C<C_0$, there exist $\delta>1$ and $\varepsilon>0$ such that $C_0>\delta(1+\varepsilon)C$. Given $N$, let $n\in\mathbb{Z}^+$ be maximum such that
	$$n+\frac{n^{t+1}}{(k-1)(k-2)\cdots(k-t-1)}\leq N.$$
	Assume that $N$ is sufficiently large so that $n\geq n_1(k-1,t,C_0)$. Then, by the above, there is a maximally cover-free, at most $t$-intersecting and $k$-uniform family $\mathcal{H}\subseteq 2^{[N]}$ so that
	$$|\mathcal{H}|\geq\frac{C_0n^{t+1}}{(k-1)(k-2)\cdots(k-t-1)}.$$
	Since $n$ is maximal, we have
	$$N<(n+1)+\frac{(n+1)^{t+1}}{(k-1)(k-2)\cdots(k-t-1)}.$$
	If $N$ (and thus $n$) is sufficiently large such that
	$$(n+1)<\frac{\varepsilon(n+1)^{t+1}}{(k-1)(k-2)\cdots(k-t-1)}\quad\text{and}\quad\left(1+\frac{1}{n}\right)^{t+1}<\delta,$$
	then
	$$N<\frac{(1+\varepsilon)(n+1)^{t+1}}{(k-1)(k-2)\cdots(k-t-1)}<\frac{\delta(1+\varepsilon)n^{t+1}}{(k-1)(k-2)\cdots(k-t-1)}$$
	and it follows that
	$$|\mathcal{H}|\geq\frac{C_0n^{t+1}}{(k-1)(k-2)\cdots(k-t-1)}>\frac{C_0N}{\delta(1+\varepsilon)}>CN.$$ \qed
\end{proof}

\subsection{An Explicit Construction}
While Proposition \ref{asymptotic_bound} provides an answer for the maximum size of an at most $t$-intersecting and $k$-uniform family for large enough $n$, we cannot explicitly construct such set families since Theorem \ref{keevash} (and hence Proposition \ref{asymptotic_bound}) is nonconstructive. In this section, we establish a method that explicitly constructs set families with larger parameters from set families with smaller parameters.

Fix a positive integer $t$. For an at most $t$-intersecting and $k$-uniform family $\mathcal{H}\subseteq 2^{[n]}$, define $$s(\mathcal{H})=\frac{k|\mathcal{H}|}{n}$$ as the ``relative size'' of $\mathcal{H}$ with respect to the parameters $k$ and $n$. Note that the maximum possible value of $|\mathcal{H}|$ should increase with larger $n$ and decrease with larger $k$, hence $s(\mathcal{H})$ is a reasonable measure of the ``relative size'' of $\mathcal{H}$.

The following result shows that it is possible to construct a sequence of at most $t$-intersecting and $k_j$-uniform families $\mathcal{H}\subseteq 2^{[n_j]}$, where $k_j\rightarrow\infty$, such that all set families in the sequence have the same relative size.

\begin{proposition}
	Let $\mathcal{H}\subseteq 2^{[n]}$ be an at most $t$-intersecting and $k$-uniform family. Then there exists a sequence of set families $\mathcal{H}_j$ such that
	\begin{enumerate}[label = {(\alph*)}]
		\item $\mathcal{H}_j$ is an at most $t$-intersecting and $k_j$-uniform set family,
		\item $s(\mathcal{H}_j)=s(\mathcal{H})$ for all $j$,
		\item $\lim_{j\rightarrow\infty}k_j=\infty$.
	\end{enumerate}
\end{proposition}
\begin{proof}
	We will define the set families $\mathcal{H}_j$ inductively. Let $\mathcal{H}_1=\mathcal{H}$, and $\mathcal{H}_j\subseteq 2^{[n_j]}$ be an at most $t$-intersecting $k_j$-uniform family for some $j\in\mathbb{Z}^+$ such that $m=|\mathcal{H}_j|$. Consider set families $\mathcal{G}^{(1)},\ldots,\mathcal{G}^{(m)},\mathcal{H}^{(1)},\ldots,\mathcal{H}^{(m)}$, defined over disjoint universes such that each $\mathcal{G}^{(\ell)}$ (and similarly, each $\mathcal{H}^{(\ell)}$) is isomorphic to $\mathcal{H}_j$. Let
	$$\mathcal{G}^{(\ell)}=\{B_1^{(\ell)},\ldots,B_m^{(\ell)}\},\quad \mathcal{H}^{(\ell)}=\{C_1^{(\ell)},\ldots,C_m^{(\ell)}\}.$$
	For $1\leq h,i\leq m$, define the sets $A_{h,i}=B_h^{(i)}\sqcup C_i^{(h)}$, and let
	$$\mathcal{H}_{j+1}=\{A_{h,i}: 1\leq h,i\leq m\}.$$
	It is clear that $\mathcal{H}_{j+1}$ is a $2k_j$-uniform family defined over a universe of $2mn_j$ elements, and that $|\mathcal{H}_{j+1}|=m^2$. We claim that $\mathcal{H}_{j+1}$ is at most $t$-intersecting. Indeed, if $(h_1,i_1)\neq (h_2,i_2)$, then
	\begin{align*}
		|A_{h_1,i_1}\cap A_{h_2,i_2}|&=|(B_{h_1}^{(i_1)}\sqcup C_{i_1}^{(h_1)})\cap (B_{h_2}^{(i_2)}\sqcup C_{i_2}^{(h_2)})| \\
		&=|B_{h_1}^{(i_1)}\cap B_{h_2}^{(i_2)}|+|C_{i_1}^{(h_1)}\cap C_{i_2}^{(h_2)}| \\
		&=
		\begin{cases}
			|C_{i_1}^{(h_1)}\cap C_{i_2}^{(h_2)}|\leq t & \text{if }h_1=h_2\text{ and }i_1\neq i_2, \\
			|B_{h_1}^{(i_1)}\cap B_{h_2}^{(i_2)}|\leq t & \text{if }h_1\neq h_2\text{ and }i_1=i_2, \\
			0 & \text{if }h_1\neq h_2\text{ and }i_1\neq i_2.
		\end{cases}
	\end{align*}
	Finally,
	$$s(\mathcal{H}_{j+1})=\frac{k_{j+1}|\mathcal{H}_{j+1}|}{n_{j+1}}=\frac{2k_jm^2}{2mn_j}=\frac{k_j|\mathcal{H}_j|}{n_j}=s(\mathcal{H}_j).$$ \qed
\end{proof}

\begin{remark}
	In the above proposition, $n_j$, $k_j$, and $|\mathcal{H}_j|$ grow with $j$. Clearly, given a family $\mathcal{H}$, it is also possible to construct a sequence of set families $\mathcal{H}_j$ such that $s(\mathcal{H}_j)=s(\mathcal{H})$ for all $j$, where $n_j$ and $|\mathcal{H}_j|$ grow with $j$, while $k_j$ stays constant.
	
	It is natural to ask, therefore, if it is possible to construct a sequence of set families satisfying $s(\mathcal{H}_j)=s(\mathcal{H})$, where $n_j$ and $k_j$ grow with $j$, but $|\mathcal{H}_j|$ stays constant. In fact, this is not always possible. Indeed, let $\mathcal{H}$ be the Fano plane, then $t=1$, $n=7$, $k=3$, and $|\mathcal{H}|=7$. Note that $\mathcal{H}$ satisfies Proposition \ref{larger_n} with equality, i.e.,
	$$\frac{(k|\mathcal{H}|)^2}{n}=k|\mathcal{H}|+t(|\mathcal{H}|^2-|\mathcal{H}|).$$
	If we let $n'=\lambda n$ and $k'=\lambda k$ for some $\lambda>1$, then
	$$\frac{(k'|\mathcal{H}|)^2}{n'}=\lambda\frac{(k|\mathcal{H}|)^2}{n}=\lambda\left(k|\mathcal{H}|+t(|\mathcal{H}|^2-|\mathcal{H}|)\right)>k'|\mathcal{H}|+t(|\mathcal{H}|^2-|\mathcal{H}|).$$
	Hence, by Proposition \ref{larger_n}, there is no $k'$-uniform and at most $t$-intersecting family $\mathcal{H}'\subseteq 2^{[n']}$ such that $|\mathcal{H}'|=|\mathcal{H}|=7$.
\end{remark}

\section{Generating Rounded Gaussians from Physical Communications}\label{Sec6}
In this section, we describe our procedure, called \textit{Rounded Gaussians from Physical Communications (RGPC)}, that generates deterministic errors from a rounded Gaussian distribution --- which we later prove to be sufficiently independent in specific settings. RGPC is comprised of the following two subprocedures: 
\begin{itemize}
	\item Hypothesis generation: a protocol to generate a linear regression hypothesis from the training data, which, in our case, is comprised of the physical layer communications between participating parties.
	\item Rounded Gaussian error generation: this procedure allows us to use the linear regression hypothesis --- generated by using physical layer communications as training data --- to derive deterministic rounded Gaussian errors. The outcome of this procedure is that it samples from a rounded Gaussian distribution in a manner that is (pseudo)random to a PPT external/internal adversary but is deterministic to the authorized parties. 
\end{itemize}

\subsection{Setting and Central Idea}\label{broad}
For the sake of intelligibility, we begin by giving a brief overview of our central idea. Let there be a set of $n \geq 2$ parties, $\mathcal{P} = \{P_i\}_{i=1}^n$. All parties agree upon a function $f(x) = \beta_0 + \beta_1 x,$ with intercept $\beta_0 \leftarrow \mathbb{Z}$ and slope $\beta_1 \leftarrow \mathbb{Z}$. Let $\mathcal{H} \subseteq 2^{\mathcal{P}}$ be a family of sets such that each set $H_i \in \mathcal{H}$ forms a star graph $\partial_i$ wherein each party is connected to a central hub $C_i \notin H_i$ (for all $i \in [|\mathcal{H}|])$ via two channels: one Gaussian and another error corrected. If $\mathcal{H}$ is $k$-uniform and at most $t$-intersecting, then each star in the interconnection graph formed by the sets $H_i \in \mathcal{H}$ contains exactly $k$ members and $2k$ channels such that $|\partial_i \cap \partial_j| \leq t$. During the protocol, each party $P_j$ sends out message pairs of the form $x_j, f(x_j)$, where $x_j \leftarrow \mathbb{Z}$ and $f$ is a randomly selected function of specific type (more on this later), to the central hubs of all stars that it is a member of, such that:
\begin{itemize}
	\item $f(x)$ is sent over the Gaussian channel, 
	\item $x$ is sent over the error corrected channel.
\end{itemize}

For the sake of simplicity, we only consider a single star for analyses in this section. All arguments and analyses from this section naturally extend to the setting with multiple stars. Due to the guaranteed errors occurring in the Gaussian channel, the messages recorded at each central hub $C_i$ are of the form: $y = f(x) + \varepsilon_x,$ where $\varepsilon_x$ belongs to some Gaussian distribution $\mathcal{N}(0, \sigma^2)$ with mean zero and standard deviation $\sigma$. Recall that for most schemes, the value of $\sigma$ primarily relies on the smoothing parameter. In our experiments, which are discussed in \Cref{Sim}, we choose the range of $\sigma$ as $\sigma \in [10,300]$ --- which lies in the range specified by New Hope \cite{Alkim[16]}, BCNS \cite{WBos[15]}, and BLISS \cite{DucDur[13]}. Note that even though techniques such as Gaussian `convolutions' \cite{Pie[10]} can be used to shrink the range of $\sigma$ that is required for the security guarantees of the given cryptosystem \cite{Regazz[18]}, such techniques are not necessary for our solution as any range for $\sigma$ can be trivially realized by generating training data from a distribution with the suitable $\sigma$ value/range. The central hub, $C_i$, forwards $\{x, y\}$ to all parties over the respective error corrected channels in $\partial_i$. 

In our algorithm, we use least squares linear regression which aims to minimize the sum of the squared residuals. We know that the hypothesis generated by linear regression is of the form: $h(x) = \hat{\beta}_0 + \hat{\beta}_1 x$. Thus, the statistical error, with respect to our target function, comes out as: 
\begin{equation}\label{hypoEq}
	\bar{e}_x = |y - h(x)|.
\end{equation}
Due to the nature of the physical layer errors and independent channels, we know that the errors $\varepsilon_x$ are random and independent. Thus, it follows that for restricted settings, the error terms $\bar{e}_{x_i}$ and $\bar{e}_{x_j}$ are independent (with respect to a PPT adversary) for all $x_i \neq x_j$, and --- are expected to --- belong to a Gaussian distribution. Next, we round $\bar{e}_x$ to the nearest integer as: $e_x =  \lfloor \bar{e}_x \rceil$ to get the final error, $e_x$, which: 
\begin{itemize}
	\item is determined by $x$, 
	\item belongs to a rounded Gaussian distribution.
\end{itemize}

We know from~\cite{Reg[05],Gold[10],Duc[15],Hul[17]} that --- with appropriate parameters --- rounded Gaussians satisfy the hardness requirements for multiple LWE-based constructions. We are now ready to discuss RGPC protocol in detail.

\begin{note}\label{ImpNote}
With a sufficiently large number of messages, $f(x)$ can be very closely approximated by the linear regression hypothesis $h(x)$. Therefore, with a suitable choice of parameters, it is quite reasonable to expect that the error distribution is Gaussian (which is indeed the case --- see \Cref{Lemma1}, where we use drowning/smudging to argue about the insignificance of negligible uniform error introduced by linear regression analysis). Considering this, we also examine the more interesting case wherein the computations are performed in $\mathbb{Z}_m$ (for some $m \in \mathbb{Z}^+ \setminus \{1\}$) instead of over $\mathbb{Z}$. However, our proofs and arguments are presented according to the former case, i.e., where the computations are performed over $\mathbb{Z}$. We leave adapting the proofs and arguments for the latter case, i.e., computations over $\mathbb{Z}_m$, as an open problem.
\end{note}
 
\subsection{Hypothesis Generation from Physical Layer Communications}\label{Gen}
In this section, we describe hypothesis generation from physical layer communications which allows us to generate an optimal linear regression hypothesis, $h(x)$, for the target function $f(x)$. As mentioned in \Cref{ImpNote}, we consider the case wherein the error computations are performed in $\mathbb{Z}_m$. As described in \Cref{broad}, the linear regression data for each subset of parties $H_i \in \mathcal{H}$ is comprised of the messages exchanged within star graph $\partial_i$ --- that is formed by the parties in $H_i \cup C_i$.
 
\subsubsection{Assumptions.}\label{assumption}
We assume that the following conditions hold:
\begin{enumerate}
	\item Value of the integer modulus $m$: 
	\begin{itemize}
		\item is either known beforehand, or 
		\item can be derived from the target function.
	\end{itemize}
	\item \label{assume2} Size of the dataset, i.e., the total number of recorded physical layer messages, is reasonably large such that there are enough data points to accurately fit linear regression on any function period. In our experiments, we set it to $2^{16}$ messages.
	\item \label{assume} For a dataset $\mathcal{D} = \{(x_i,y_i)\}~(i \in [\ell])$, it holds for the slope, $\beta_1$, of $f(x)$, that $\ell/\beta_1$ is superpolynomial. For our experiments, we set $\beta_1$ such that $\ell/\beta_1\geq 100$.
\end{enumerate}

\subsubsection{Setup.}
Recall that the goal of linear regression is to find subset(s) of data points that can be used to generate a hypothesis $h(x)$ to approximate the target function, which in our case is $f(x) + \varepsilon_x$. Then, we extrapolate it to the entire dataset. However, since modulo is a periodic function, there is no explicit linear relationship between $x \leftarrow \mathbb{Z}$ and $y = f(x) + \varepsilon_x \bmod m$, even without the error term $\varepsilon_x$. Thus, we cannot directly apply linear regression to the entire dataset $\mathcal{D} = \{(x_i,y_i)\}~(i \in [\ell])$ and expect meaningful results unless $\beta_0=0$ and $\beta_1 = 1$.

\begin{figure}[h!]
	\centering
	\includegraphics[scale=.084]{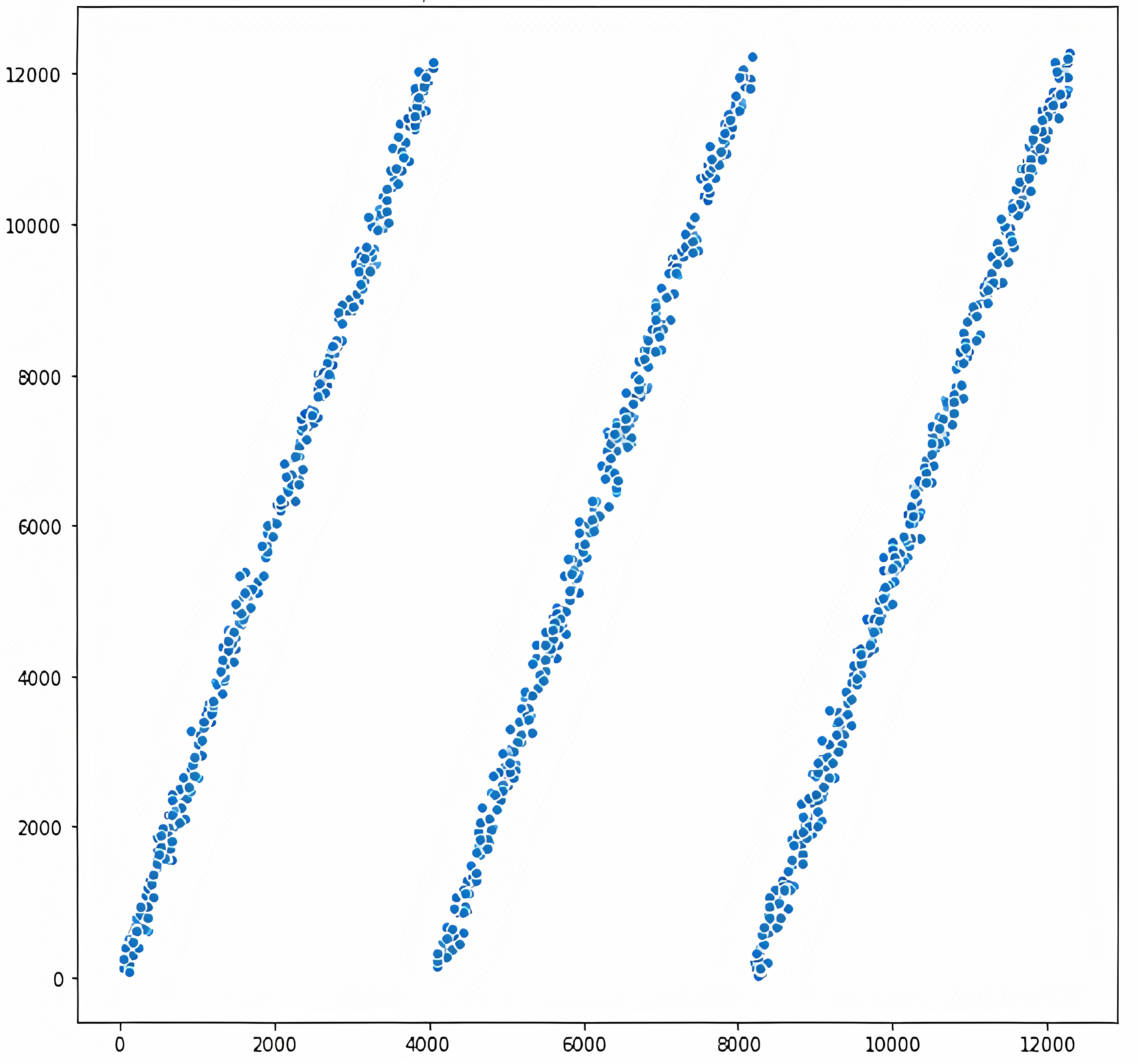}
	\caption{Scatter plot for $y = 3 x + \varepsilon_x \bmod 12288$ (three periods)}\label{Fig1}
\end{figure}

We arrange the dataset, $\mathcal{D}$, in ascending order with respect to the $x_i$ values, i.e., for $1 \leq i < j \leq m$ and all $x_i, x_j \in \mathcal{D}$, it holds that: $x_i < x_j$. Let $\mathcal{S} = \{x_i\}_{i=1}^\ell$ denote the ordered set with $x_i$ $(\forall i \in [\ell])$ arranged in ascending order. Observe that the slope of $y = f(x) + \varepsilon_x \bmod m$ is directly proportional to the number of periods on any given range, $[x_i, x_j]$. For example, observe the slope in \Cref{Fig1}, which depicts the scatter plot for $y = 3x + \varepsilon_x \bmod 12288$ with three periods. Therefore, in order to find a good linear fit for our target function on a subset of dataset that lies inside the given range, $[x_i, x_j]$, we want to correctly estimate the length of a single period. Consequently, our aim is to find a range $[x_i, x_j]$ for which the following is minimized:

\begin{equation}\label{eq}
	\Big| \hat{\beta}_1 - \dfrac{m}{x_j - x_i} \Big|,
\end{equation}
where $\hat{\beta}_1$ denotes the slope for our linear regression hypothesis $h(x) = \hat{\beta}_0 + \hat{\beta}_1 x$ computed on the subset with $x$ values in $[x_i,x_j]$.

\subsubsection{Generating Optimal Hypothesis.}
The following procedure describes our algorithm for finding the optimal hypothesis $h(x)$ and the target range $[x_i, x_j]$ that satisfies \Cref{eq} for $\beta_0=0$. When $\beta_0$ is not necessarily $0$, a small modification to the procedure (namely, searching over all intervals $[x_i, x_j]$, instead of searching over only certain intervals as described below) is needed.

Let $\kappa$ denote the total number of periods, then it follows from Assumption~\ref{assume} (from Section~\ref{assumption}) that $\kappa \leq \lceil \ell/100 \rceil$. Let $\delta_{\kappa,i} = |\hat{\beta}_1(\kappa, i) - \kappa|$, where $\hat{\beta}_1(\kappa, i)$ denotes that $\hat{\beta}_1$ is computed over the range $\left[x_{\left\lfloor(i-1)\ell/\kappa\right\rfloor+1},x_{\left\lfloor i\ell/\kappa\right\rfloor}\right]$.

\begin{enumerate}
	\item Initialize the number of periods with $\kappa = 1$ and calculate $\delta_{1,1} = |\hat{\beta}_1(1,1) - 1|$. 
	\item Compute the $\delta_{\kappa,i}$ values for all $1 < \kappa \leq \lceil \ell/100 \rceil$ and $i \in [\kappa]$. For instance, $\kappa = 2$ denotes that we consider two ranges: $\hat{\beta}_1 (2,1)$ is calculated on $\left[ x_1, x_{\lfloor \ell/\kappa \rfloor}\right]$ and $\hat{\beta}_1 (2,2)$ on $\left[x_{\lfloor \ell/\kappa \rfloor +1}, x_\ell\right]$. Hence, we compute $\delta_{2,i}$ for these two ranges. Similarly, $\kappa =3$ denotes that we consider three ranges $\left[ x_1, x_{\lfloor \ell/\kappa \rfloor}\right]$, $\left[x_{\lfloor \ell/\kappa \rfloor +1}, x_{\lfloor 2\ell/\kappa \rfloor}\right]$ and $\left[x_{\lfloor 2\ell/\kappa \rfloor +1}, x_\ell\right]$, and we compute $\hat{\beta}(3,i)$ and $\delta_{3,i}$ over these three ranges. Hence, $\delta_{\kappa,i}$ values are computed for all $(\kappa,i)$ that satisfy $1 \leq i \leq \kappa \leq \lceil \ell/100 \rceil$.
	\item Identify the optimal value $\delta = \min_{\kappa,i}(\delta_{\kappa,i})$, which is the minimum over all $\kappa \in [\lceil \ell/100 \rceil]$ and $i \in [\kappa]$.
	\item After finding the minimal $\delta$, output the corresponding (optimal) hypothesis $h(x)$.
\end{enumerate}

What the above algorithm does is basically a grid search over $\kappa$ and $i$ with the performance metric being minimizing the $\delta_{\kappa,i}$ value.

	\begin{center}
		\textbf{Grid search: more details}
	\end{center}
	\begin{myenv}
		Grid search is an approach used for hyperparameter tuning. It methodically builds and evaluates a model for each combination of parameters. Due to its ease of implementation and parallelization, grid search has prevailed as the de-facto standard for hyperparameter optimization in machine learning, especially in lower dimensions. For our purpose, we tune two parameters $\kappa$ and $i$. Specifically, we perform grid search to find hypotheses $h(x)$ for all $\kappa$ and $i$ such that $\kappa \in [\lceil \ell/100 \rceil]$ and $i \in [\kappa]$. Optimal hypothesis is the one with the smallest value of the performance metric $\delta_{\kappa,i}$. 
	\end{myenv}

\subsection{Simulation and Testing}\label{Sim}
We tested our RGPC algorithm with varying values of $m$ and $\beta_1$ for the following functions: 
\begin{itemize}
	\item $f(x) = \beta_0 + \beta_1 x,$ 
	\item $f(x) = \beta_0 + \beta_1 \sqrt{x},$ 
	\item $f(x) = \beta_0 + \beta_1 x^2,$ 
	\item $f(x) = \beta_0 + \beta_1 \sqrt[3]{x},$ 
	\item $f(x) = \beta_0 + \beta_1 \ln (x+1).$ 
\end{itemize}

To generate the training data, we simulated the channel noise, $\varepsilon_x$, as a random Gaussian noise (introduced by the Gaussian channel), which we sampled from various distributions with zero mean, standard deviation $\sigma \in [10, 300]$, and $m \in [20000]$. Final channel noise was computed by rounding $\varepsilon_x$ to the nearest integer and reducing the result modulo $m$.

For each function, we generated $2^{16}$ unique input-output pairs, exchanged over Gaussian channels, i.e., the dataset for each function is of the form $\mathcal{D} = \{(x_i,y_i)\}$, where $i \in [2^{16}]$. As expected, given the dataset $\mathcal{D}$ with data points $x_i, y_i = f(x_i) + \varepsilon_i \bmod m$, our algorithm always converged to the optimal range, yielding close approximations for the target function with deterministic errors, $\bar{e}_{x} = |y - h(x)| \bmod m$. \Cref{Fig2} shows a histogram of the errors $\bar{e}_x$ generated by our RGPC protocol --- with our training data --- for the target (linear) function $y = 546 x + \varepsilon_x \bmod 12288$. The errors indeed belong to a truncated Gaussian distribution, bounded by the modulus $12288$ from both tails.

\begin{figure}
	\centering
	\includegraphics[scale=.080]{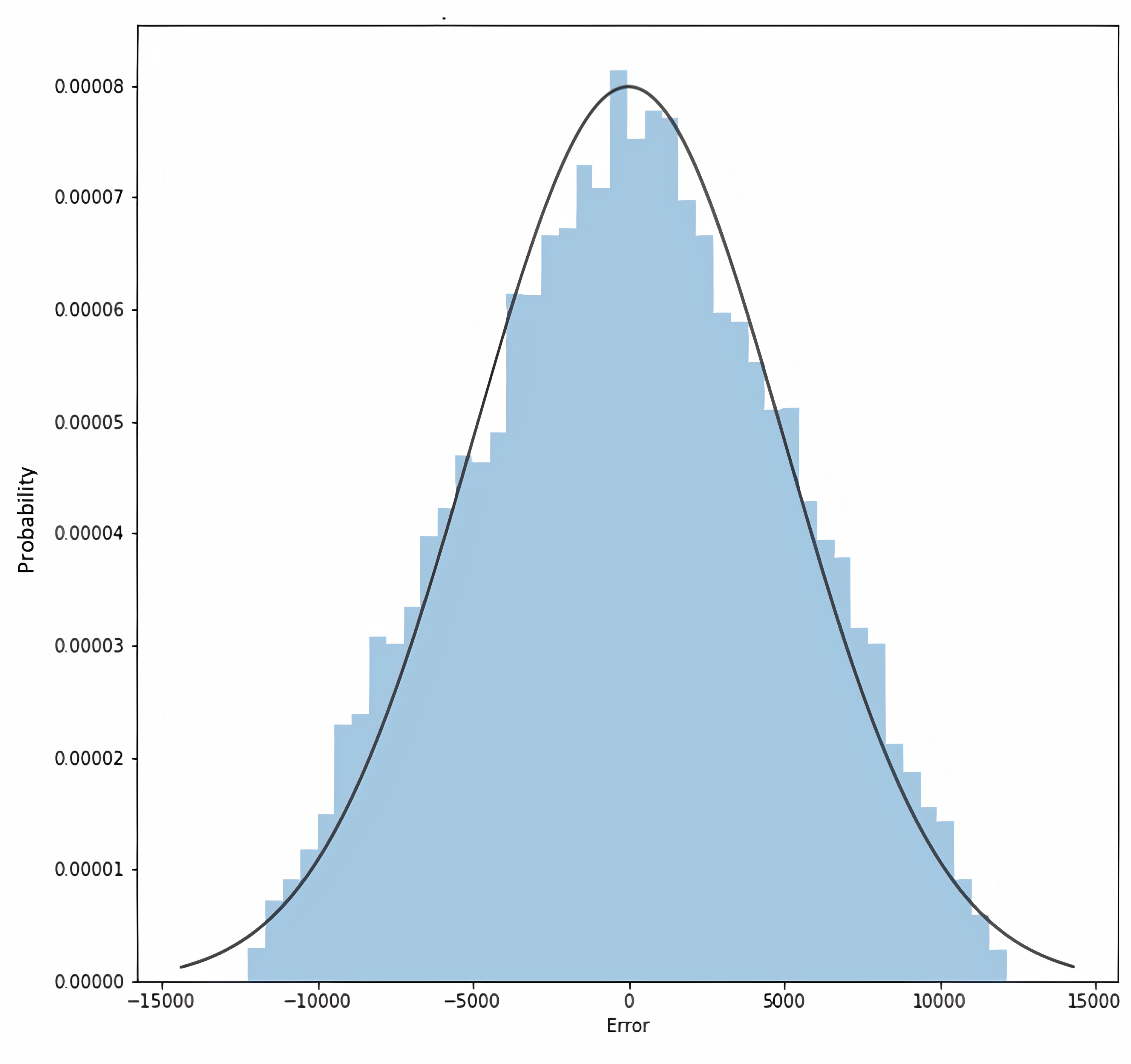}
	\caption{Distribution plot of $\bar{e}_x$ for $y = 546 x + \varepsilon_x \bmod 12288$. Slope estimate: $\beta_1 = 551.7782$.}\label{Fig2}
\end{figure}

\begin{figure}
	\centering
	\includegraphics[scale=.115]{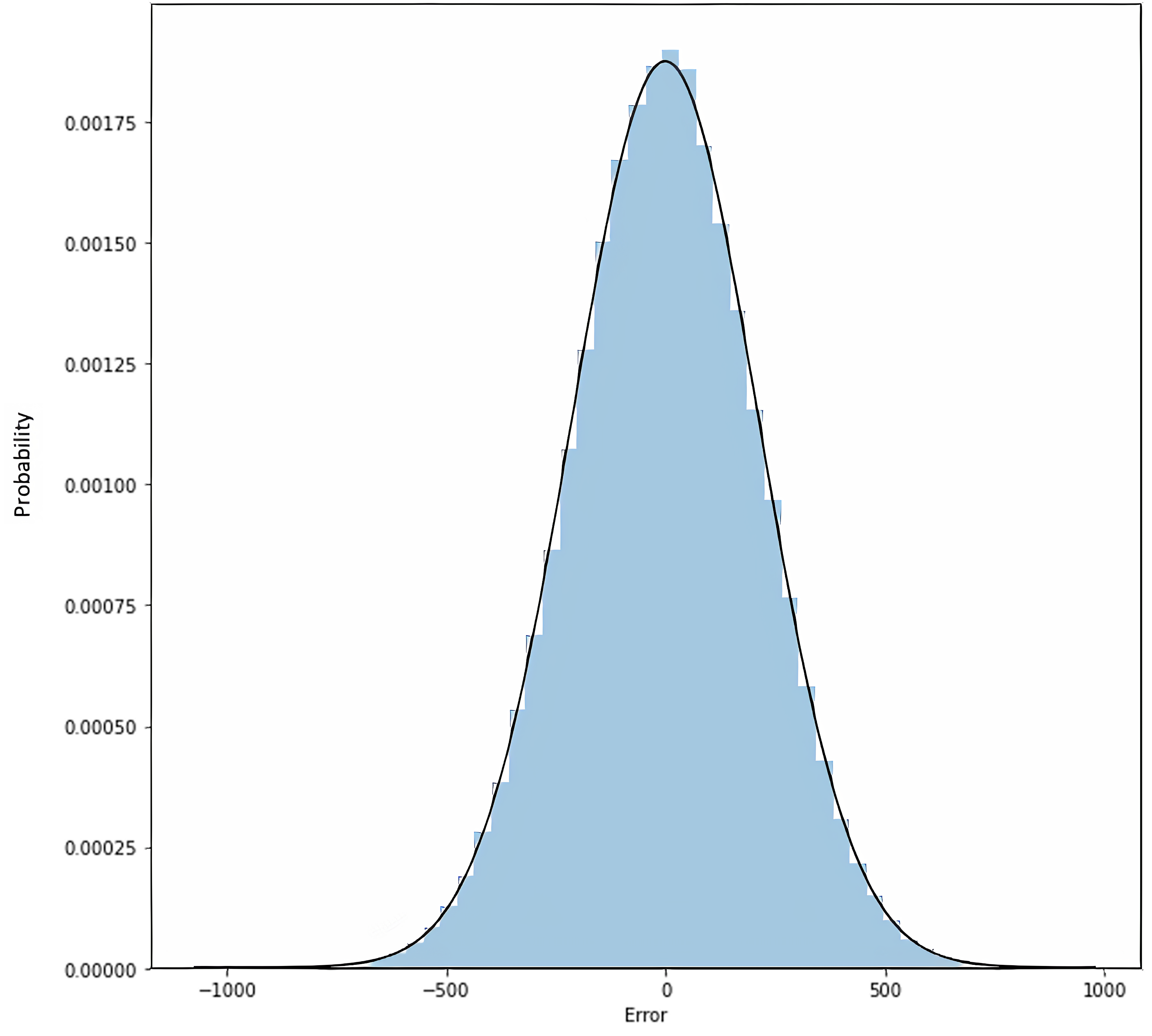}
	\caption{Distribution plot of $\bar{e}_x$ for $y = 240 \sqrt{x} + \varepsilon_x \bmod 12288$. Slope estimate: $\beta_1 = 239.84$.}\label{Fig3}
\end{figure}

Moving on to the cases wherein the independent variable $x$ and the dependent variable $y$ have a nonlinear relation: the most representative example of such a relation is the power function $f(x)=\beta_1 x^\vartheta$, where $\vartheta \in \mathbb{R}$. We know that nonlinearities between variables can sometimes be linearized by transforming the independent variable. Hence, we applied the following transformation: if we let $x_{\upsilon}=x^\vartheta$, then $f_\upsilon(x_\upsilon) = \beta_1 x_\upsilon = f(x)$ is linear in $x_\upsilon$. This can now be solved by applying our hypothesis generation algorithm for linear functions. \Cref{Fig3,Fig4,Fig6,Fig7} show the histograms of the errors $\bar{e}_x$ generated by our training datasets for the various nonlinear functions from the list given at the beginning of \Cref{Sim}. It is clear that the errors for these nonlinear functions also belong to truncated Gaussian distributions, bounded by their respective moduli from both tails.

\begin{figure}
	\centering
	\includegraphics[scale=.115]{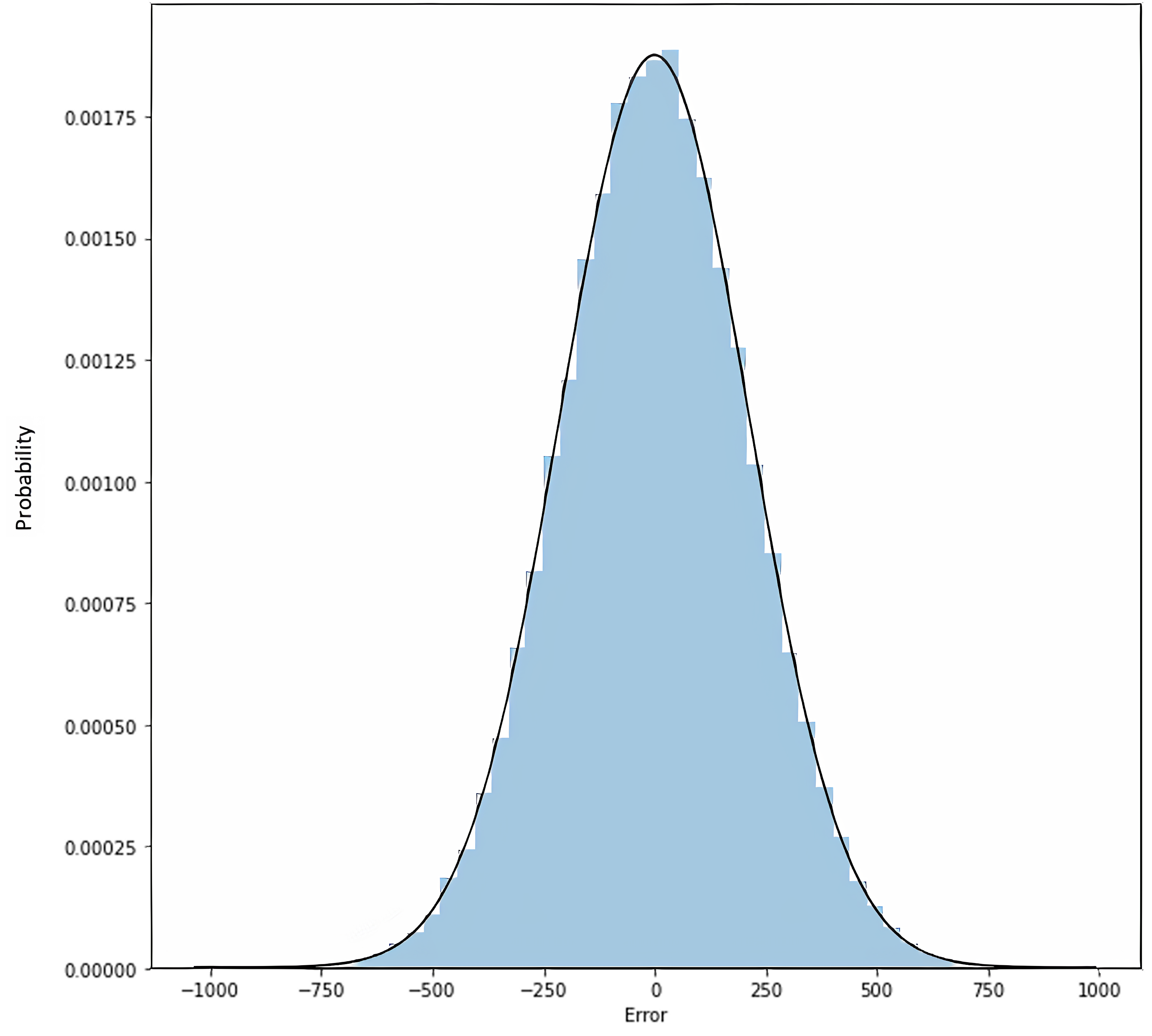}
	\caption{Distribution plot of $\bar{e}_x$ for $y = 125 x^2 + \varepsilon_x \bmod 10218$. Slope estimate: $\beta_1 = 124.51$.}\label{Fig4}
\end{figure} 

\begin{figure}
	\centering
	\includegraphics[scale=.115]{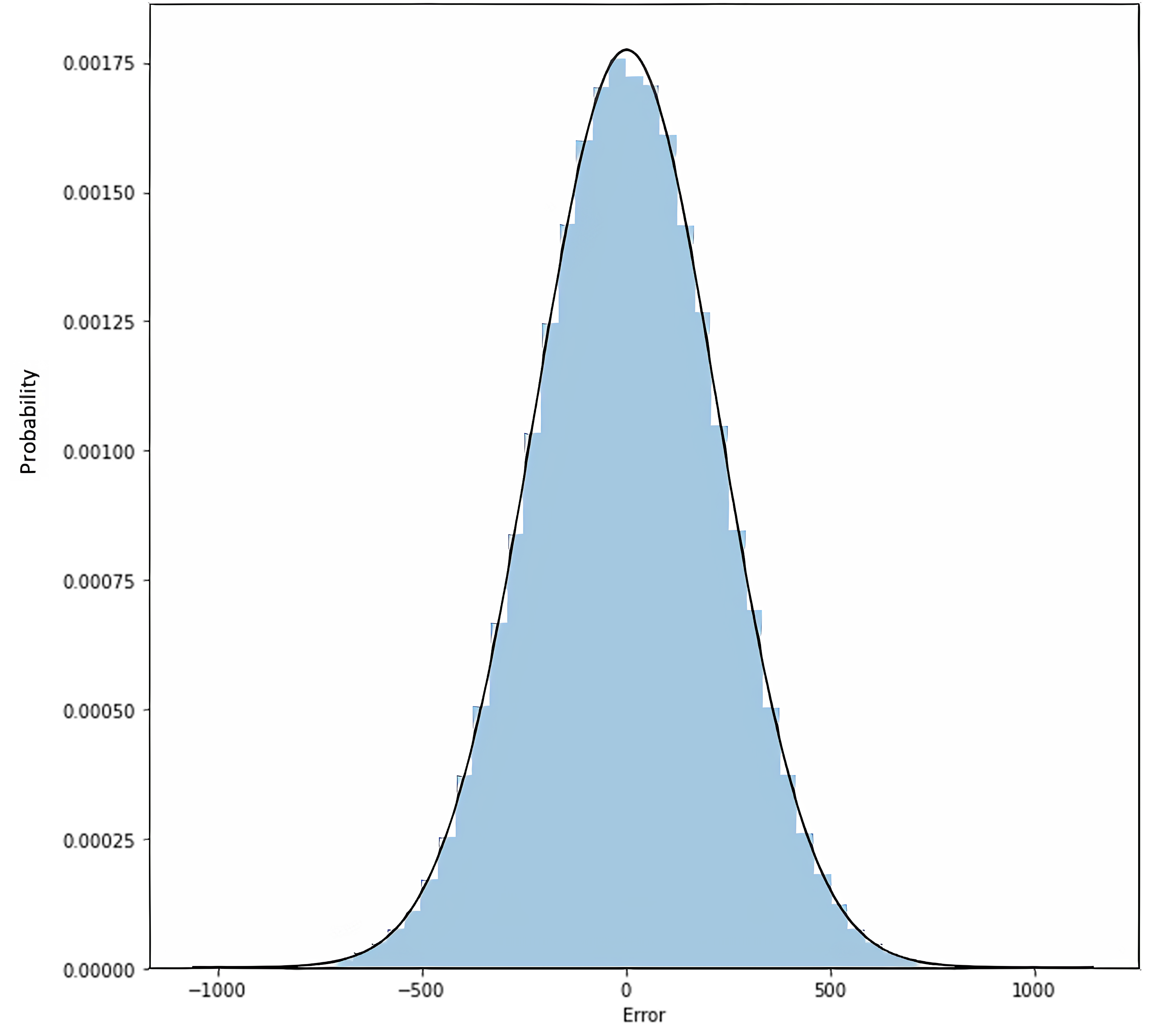}
	\caption{Distribution plot of $\bar{e}_x$ for $y = 221 \sqrt[3]{x} + \varepsilon_x \bmod 11278$. Slope estimate: $\beta_1 = 221.01$.}\label{Fig6}
\end{figure} 

\begin{figure}
	\centering
	\includegraphics[scale=.115]{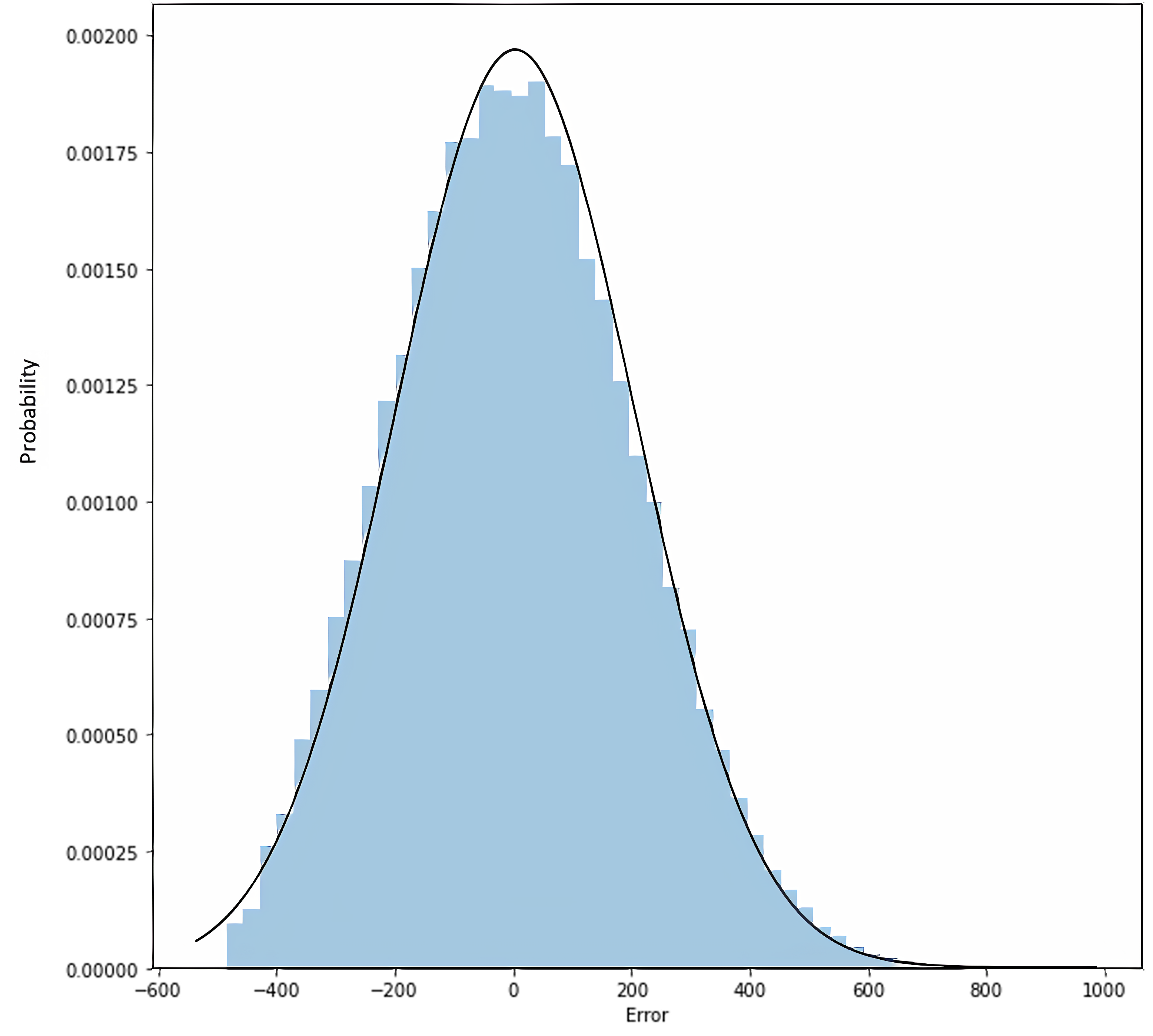}
	\caption{Distribution plot of $\bar{e}_x$ for $y = 53 \ln (x+1) + \varepsilon_x \bmod 8857$. Slope estimate: $\beta_1 = 54.48$.}\label{Fig7}
\end{figure} 

\subsection{Complexity}\label{Time}
Let the size of the dataset collected by recording the physical layer communications be $\ell$. Then, the complexity for least squares linear regression is $\mathrm{\Theta}(\ell)$ additions and multiplications. It follows from Assumption~\ref{assume} (from \Cref{assumption}) that $\ell^2$ is an upper bound on the maximum number of evaluations required for grid search. Therefore, the overall asymptotic complexity of our algorithm to find optimal hypothesis, and thereafter generate deterministic rounded Gaussian errors is $O(\ell^3)$. In comparison, the complexity of performing least squares linear regression on a dataset $\{(x_i, y_i)\}$ that has not been reduced modulo $m$ is simply $\mathrm{\Theta}(\ell)$ additions and multiplications.

\subsection{Error Analysis}
Before moving ahead, we recommend the reader revisit \Cref{ImpNote}.

Our RGPC protocol generates errors, $e_{x} = \lceil |y - h(x)| \rfloor$, in a deterministic manner by generating a mapping defined by the hypothesis $h$. Recall that $h$ depends on two factors, namely the randomly generated function $f(x)$ and the $\ell$ random channel errors $\varepsilon_x$ --- both of which are derived via $\ell$ random inputs $x \leftarrow \mathbb{Z}$. Revisiting \Cref{P2SI}, we observe that: 
\begin{itemize}
	\item RGPC takes $\ell$ random elements $x \leftarrow \mathbb{Z}$ and $\ell$ evaluations of $f(x)$, each with an added random channel error $\varepsilon_x$. Collectively, these form the random input, $r$, described in \Cref{P2SI},
	\item RGPC returns a deterministic mapping which depends on the above random input, $r$. Following the notation from \Cref{P2SI}, we can write it as $\mathfrak{M}_r$. However, we choose the notation $\mathfrak{M}_h$ instead, but note that both are equivalent henceforth,
	\item Proving that $\mathfrak{M}_h$ satisfies the requirements outlined for the deterministic function/mapping in \Cref{P2SI} would directly establish RGPC as a P2SI algorithm --- as described in \Cref{P2SI}. 
\end{itemize}

Let difference of two values modulo $m$ is always represented by an element in $(-m/2,(m+1)/2]$. We make the further assumption that there exists a constant $b\geq 1$ such that the $x_i$ values satisfy:
\begin{equation}
	\label{assumption_on_x}
	\tag{$\dagger$}
	\frac{(x_i-\bar{x})^2}{\sum_{j=1}^{\ell} (x_j-\bar{x})^2}\leq\frac{b}{\ell}
\end{equation}
for all $i=1,\ldots,\ell$, where $\bar{x}=\sum_{j=1}^{\ell}\frac{x_j}{\ell}$.

Observe that if $\ell-1$ divides $m-1$ and the $x_i$ values are $0,\frac{m-1}{\ell-1},\ldots,\frac{(\ell-1)(m-1)}{\ell-1}$, then $$\sum_{j=1}^{\ell} (x_j-\bar{x})^2=\frac{\ell(\ell^2-1)(m-1)^2}{12(\ell-1)^2},$$ and the numerator is bounded above by $\frac{(m-1)^2}{4}$. Thus, \Cref{assumption_on_x} is satisfied for $b=3$. In general, by the strong law of large numbers, choosing a large enough number of $x_i$'s uniformly at random from $[0, m-1]$ will, with very high probability, yield $\bar{x}$ close to $\frac{m-1}{2}$ and $\frac{1}{\ell}{\sum_{j=1}^{\ell} (x_j-\bar{x})^2}$ close to $\frac{(m^2-1)}{12}$ (since $X\sim U(0, m-1)\implies \E(X)=\frac{m-1}{2}$ and $\var(X)=\frac{m^2-1}{12}$). Hence, \Cref{assumption_on_x} is satisfied with a small constant $b$, e.g., $b=4$.

Recall that the dataset is $\mathcal{D} = \{(x_i, y_i)\}_{i=1}^\ell$, where $y_i=f(x_i)+\varepsilon_i=\beta_0+\beta_1 x_i+\varepsilon_i$, with $\varepsilon_i\sim \mathcal{N}(0,\sigma^2)$. Further recall that the regression line is given by $h(x)=\hat{\beta_0}+\hat{\beta}_1x$. Then the error $\bar{e}_i$ is given by
\begin{align*}
\bar{e}_i&=(\hat{\beta_0}+\hat{\beta}_1x_i)-y_i=(\hat{\beta_0}+\hat{\beta}_1x_i)-(\beta_0+\beta_1 x_i+\varepsilon_i) \\
&=(\hat{\beta_0}-\beta_0)+(\hat{\beta_1}-\beta_1)x_i-\varepsilon_i.
\end{align*}

The joint distribution of the regression coefficients $\hat{\beta_0}$ and $\hat{\beta_1}$ is given by the following well known result:

\begin{proposition}
	\label{regression_hypothesis_distribution}
	Let $y_1,y_2,\ldots,y_\ell$ be independently distributed random variables such that $y_i\sim \mathcal{N}(\alpha+\beta x_i,\sigma^2)$ for all $i=1,\ldots ,\ell$. If $\hat{\alpha}$ and $\hat{\beta}$ are the least square estimates of $\alpha$ and $\beta$ respectively, then:
	$$\begin{pmatrix}
		\hat{\alpha} \\
		\hat{\beta}
	\end{pmatrix}
	\sim \mathcal{N}\left(
	\begin{pmatrix}
		\alpha \\
		\beta
	\end{pmatrix}
	,\,
	\sigma^2
	\begin{pmatrix}
		\ell & \sum_{i=1}^{\ell} x_i \\
		\sum_{i=1}^{\ell} x_i & \sum_{i=1}^{\ell} x_i^2
	\end{pmatrix}^{-1}
	\right).$$
\end{proposition}

Applying Proposition \ref{regression_hypothesis_distribution}, and using the fact that $\mathbf{X}\sim\mathcal{N}(\boldsymbol{\mu},\mathbf{\Sigma})\implies\mathbf{A}\mathbf{X}\sim\mathcal{N}(\mathbf{A}\boldsymbol{\mu},\mathbf{A}\mathbf{\Sigma}\mathbf{A}^T)$, we get:
\begin{align*}
(\hat{\beta_0}-\beta_0)+(\hat{\beta_1}-\beta_1)x_i&\sim
	\mathcal{N}\left(0,\,
		\sigma^2\frac{\sum_{j=1}^\ell x_j^2-2x_i\sum_{j=1}^\ell x_j+\ell x_i^2}{\ell\sum_{j=1}^\ell x_j^2-(\sum_{j=1}^\ell x_j)^2}
	\right) \\
	&=
	\mathcal{N}\left(0,\,
		\frac{\sigma^2}{\ell}\left(1+\frac{\ell(x_i-\bar{x})^2}{\sum_{j=1}^\ell (x_j-\bar{x})^2}\right)
	\right).
\end{align*}
Thus, by \Cref{assumption_on_x}, the variance of $(\hat{\beta_0}-\beta_0)+(\hat{\beta_1}-\beta_1)x_i$ is bounded above by $(1+b)\sigma^2/\ell$.

Since $Z\sim \mathcal{N}(0,1)$ satisfies $|Z|\leq 2.807034$ with probability $0.995$, by the union bound, $\bar{e}_i$ is bounded by
$$|\bar{e}_i|\leq 2.807034\left(1+\sqrt{\frac{1+b}{\ell}}\right)\sigma$$
with probability at least $0.99$.

\begin{note}\label{note1}
	Our protocol allows fine-tuning the Gaussian errors by tweaking the standard deviation $\sigma$ and mean $\mu$ for the Gaussian channel. Hence, in our proofs and arguments, we only use the term ``target rounded Gaussian''. 
\end{note}

\begin{lemma}\label{Lemma1}
	Suppose that the number of samples $\ell$ is superpolynomial, and that \Cref{assumption_on_x} is satisfied with some constant $b$. Then, the errors $e_i$ belong to the target rounded Gaussian distribution.
\end{lemma}
\begin{proof}
	Recall that the error $\bar{e}_i$ has two components: one is the noise introduced by the Gaussian channel and second is the error due to regression fitting. The first component is naturally Gaussian. The standard deviation for the second component is of order $\sigma/\sqrt{\ell}$. Hence, it follows from drowning/smudging that for a superpolynomial $\ell$, the error distribution for $\bar{e}_i$ is statistically indistinguishable from the Gaussian distribution to which the first component belongs. Therefore, it follows that the final output, $e_i = \lceil \bar{e}_i \rfloor$, of the RGPC protocol belongs to the target rounded Gaussian distribution (see \Cref{note1}).  \qed
\end{proof}

To prove that all conditions in \Cref{P2SI} are satisfied for $\mathfrak{M}_h$, it remains to show that no external PPT adversary $\mathcal{A}$ has non-negligible advantage in guessing $e_x$ beforehand. We assume that for an attack query $x$, $\mathcal{A}$ makes $\poly(\L)$ queries to $\mathfrak{M}_h$ for $x' \neq x$. 

\begin{lemma}\label{Obs1}
	For an external PPT adversary $\mathcal{A}$, it holds that $\mathfrak{M}_h: x \mapsto e_x$ maps a random element $x \leftarrow \mathbb{Z}$ to a random element $e_x$ in the target rounded Gaussian distribution $\mathrm{\Psi}(0, \hat{\sigma}^2)$.
\end{lemma}
\begin{proof}
	It follows from \Cref{Lemma1} that $\mathfrak{M}_h$ outputs from the target rounded Gaussian distribution. Note the following straightforward observations:
	\begin{itemize}
		\item Since each coefficient of $f(x)$ is randomly sampled from $\mathbb{Z}_m$, $f(x)$ is a random function.
		\item The inputs to $f(x)$, $x \leftarrow \mathbb{Z}$, are sampled randomly.
		\item The Gaussian channel introduces a noise $\varepsilon_x$ to $f(x)$, that is drawn i.i.d. from a Gaussian distribution $\mathcal{N}(0,\sigma^2)$ \cite{JanaSri[09],JieYan[12],PremSne[14],XiQi[16],HaiKui[11],XuKui[12],Ye[10],Jiang[13],Zeng[15]}. Hence, the receiving parties get a random element $f(x) + \varepsilon_x$.
		\item It follows from the observations stated above that $\lceil |f(x) - h(x)| \rfloor$ outputs a random element from the target rounded Gaussian $\mathrm{\Psi}(0, \hat{\sigma}^2)$.
	\end{itemize}
	Hence, $\mathfrak{M}_h: x \mapsto e_x$ is a deterministic mapping that maps each of the $\ell$ inputs $\{x_i\}_{i=1}^\ell$ to an element $e_x$ in the desired rounded Gaussian $\mathrm{\Psi}(0, \hat{\sigma}^2)$ such that all conditions from \Cref{P2SI} are satisfied. It follows that given an input $x$ not seen before, no external PPT adversary $\mathcal{A}$ has non-negligible advantage in guessing $e_x$. \qed
\end{proof}

\section{Mutual Information Analysis}\label{Mutual}
\begin{definition}
	\emph{Let $f_X$ be the p.d.f. of a continuous random variable $X$. Then, the \emph{differential entropy} of $X$ is
	$$H(X)=-\int f_X(x)\log f_X(x)\,dx.$$}
\end{definition}

\begin{definition}
	\emph{The \emph{mutual information}, $I(X;Y)$, of two continuous random variables $X$ and $Y$ with joint p.d.f.\ $f_{X,Y}$ and marginal p.d.f.'s $f_X$ and $f_Y$ respectively, is
	$$I(X;Y)=\int\int f_{X,Y}(x,y)\log\left(\frac{f_{X,Y}(x,y)}{f_X(x)f_Y(y)}\right)\,dy\,dx.$$}
\end{definition}

From the above definitions, it is easy to prove that $I(X;Y)$ satisfies the equality:
$$I(X;Y)=H(X)+H(Y)-H(X,Y).$$

Let us now describe our aim. Suppose, for $i=1,2,\ldots,\ell$, we have
$$y_i\sim \mathcal{N}(\alpha+\beta x_i,\sigma^2)\quad\text{and}\quad z_i\sim \mathcal{N}(\alpha+\beta w_i,\sigma^2),$$
with $x_i=w_i$ for $i\in[a]$. Let $h_1(x)=\hat{\alpha}_1 x+\hat{\beta}_1$ and $h_2(w)=\hat{\alpha}_2 w+\hat{\beta}_2$ be the linear regression hypotheses obtained from the samples $(x_i,y_i)$ and $(w_i,z_i)$, respectively. Our goal is to compute an expression for the mutual information
$$I((\hat{\alpha_1},\hat{\beta_1});(\hat{\alpha_2},\hat{\beta_2})).$$

We begin by recalling the following standard fact:

\begin{proposition}
	\label{multivariate_normal_differential_entropy}
	Let $\mathbf{X}\sim \mathcal{N}(\mathbf{v},\mathbf{\Sigma})$, where $\mathbf{v}\in\mathbb{R}^d$ and $\mathbf{\Sigma}\in\mathbb{R}^{d\times d}$. Then:
	$$H(\mathbf{X})=\frac{1}{2}\log(\det\mathbf{\Sigma})+\frac{d}{2}(1+\log(2\pi)).$$
\end{proposition}

Our main result is the following:

\begin{proposition}\label{mainProp}
	Let $\hat{\alpha_1}$, $\hat{\beta_1}$, $\hat{\alpha_2}$, $\hat{\beta_2}$ be as above. Then
	\begin{gather*}
		H(\hat{\alpha_1},\hat{\beta_1})=2\log\sigma-\frac{1}{2}\log(\ell X_2-X_1^2)+(1+\log(2\pi)), \\
		H(\hat{\alpha_2},\hat{\beta_2})=2\log\sigma-\frac{1}{2}\log(\ell W_2-W_1^2)+(1+\log(2\pi)),
	\end{gather*}
	and
	\begin{align*}
		&\ H(\hat{\alpha_1},\hat{\beta_1},\hat{\alpha_2},\hat{\beta_2}) \\
		=&\ 4\log\sigma-\frac{1}{2}\log\left((\ell X_2-X_1^2)(\ell W_2-W_1^2)\right)+2(1+\log(2\pi)) \\
		&\qquad+\frac{1}{2}\log\left(1-\frac{\left(\ell C_2-2C_1X_1+a X_2\right)\left(\ell C_2-2C_1W_1+a W_2\right)}{(\ell X_2-X_1^2)(\ell W_2-W_1^2)}\right. \\
		&\qquad\left.+\frac{\left((a -1)C_2-C_3\right)\left((a -1)C_2-C_3+\ell(X_2+W_2)-2X_1W_1\right)}{(\ell X_2-X_1^2)(\ell W_2-W_1^2)}\right),
	\end{align*}
	where $X_1=\sum_{i=1}^{\ell} x_i$, $X_2=\sum_{i=1}^{\ell} x_i^2$, $W_1=\sum_{i=1}^{\ell} w_i$, $W_2=\sum_{i=1}^{\ell} w_i^2$, $C_1=\sum_{i=1}^a x_i=\sum_{i=1}^a w_i$, $C_2=\sum_{i=1}^a x_i^2=\sum_{i=1}^a w_i^2$ and $C_3=\sum_{i=1}^{\ell}\sum_{j=1,j\neq i}^{\ell}x_ix_j$. The mutual information between $(\hat{\alpha_1},\hat{\beta_1})$ and $(\hat{\alpha_2},\hat{\beta_2})$ is:
	\begin{align*}
		&I((\hat{\alpha_1},\hat{\beta_1});(\hat{\alpha_2},\hat{\beta_2})) \\
		=&\ -\frac{1}{2}\log\left(1-\frac{\left(\ell C_2-2C_1X_1+aX_2\right)\left(\ell C_2-2C_1W_1+aW_2\right)}{(\ell X_2-X_1^2)(\ell W_2-W_1^2)}\right. \\
		&\qquad\left.+\frac{\left((a-1)C_2-C_3\right)\left((a-1)C_2-C_3+\ell(X_2+W_2)-2X_1W_1\right)}{(\ell X_2-X_1^2)(\ell W_2-W_1^2)}\right).
	\end{align*}
\end{proposition}

\begin{proof}
	The expressions for $H(\hat{\alpha_1},\hat{\beta_1})$ and $H(\hat{\alpha_2},\hat{\beta_2})$ follow from Propositions \ref{regression_hypothesis_distribution} and \ref{multivariate_normal_differential_entropy}, and the expression for $I((\hat{\alpha_1},\hat{\beta_1});(\hat{\alpha_2},\hat{\beta_2}))$ is given by:
	$$I((\hat{\alpha_1},\hat{\beta_1});(\hat{\alpha_2},\hat{\beta_2}))=H(\hat{\alpha_1},\hat{\beta_1})+H(\hat{\alpha_2},\hat{\beta_2})-H(\hat{\alpha_1},\hat{\beta_1},\hat{\alpha_2},\hat{\beta_2}).$$
	It remains to derive the expression for $H(\hat{\alpha_1},\hat{\beta_1},\hat{\alpha_2},\hat{\beta_2})$.
	
	First, define the matrices
	$$\mathbf{X}=
	\begin{pmatrix}
		1 & x_1 \\
		1 & x_2 \\
		\vdots & \vdots \\
		1 & x_\ell
	\end{pmatrix},
	\qquad
	\mathbf{W}=
	\begin{pmatrix}
		1 & w_1 \\
		1 & w_2 \\
		\vdots & \vdots \\
		1 & w_\ell
	\end{pmatrix}.
	$$
	Then
	\begin{align*}
		\boldsymbol{\hat{\theta}}:=
		\begin{pmatrix}
			\hat{\alpha_1} \\
			\hat{\beta_1} \\
			\hat{\alpha_2} \\
			\hat{\beta_2}
		\end{pmatrix}
		&=
		\begin{pmatrix}
			\alpha \\
			\beta \\
			\alpha \\
			\beta
		\end{pmatrix}
		+
		\begin{pmatrix}
			(\mathbf{X}^T\mathbf{X})^{-1}\mathbf{X}^T\mathbf{U} \\
			(\mathbf{W}^T\mathbf{W})^{-1}\mathbf{W}^T\mathbf{V}
		\end{pmatrix} \\
		&=\begin{pmatrix}
			\alpha \\
			\beta \\
			\alpha \\
			\beta
		\end{pmatrix}
		+
		\begin{pmatrix}
			(\mathbf{X}^T\mathbf{X})^{-1}\mathbf{X}^T & 0 \\
			0 & (\mathbf{W}^T\mathbf{W})^{-1}\mathbf{W}^T
		\end{pmatrix}
		\begin{pmatrix}
			\mathbf{U} \\
			\mathbf{V}
		\end{pmatrix},
	\end{align*}
	where $\mathbf{U}$, $\mathbf{V}\sim \mathcal{N}(0,\sigma^2 \mathbf{I}_\ell)$. Hence, it follows that:
	$$\var(\boldsymbol{\hat{\theta}})=
	\begin{pmatrix}
		(\mathbf{X}^T\mathbf{X})^{-1}\mathbf{X}^T & 0 \\
		0 & (\mathbf{W}^T\mathbf{W})^{-1}\mathbf{W}^T
	\end{pmatrix}
	\var{
		\begin{pmatrix}
			\mathbf{U} \\
			\mathbf{V}
		\end{pmatrix}
	}
	\begin{pmatrix}
		\mathbf{X}(\mathbf{X}^T\mathbf{X})^{-1} & 0 \\
		0 & \mathbf{W}(\mathbf{W}^T\mathbf{W})^{-1}
	\end{pmatrix}
	.$$
	For any matrix $\mathbf{M}=(M_{i,j})$, let $[\mathbf{M}]_a$ denote the matrix with the same dimensions as $\mathbf{M}$, and with entries
	$$([\mathbf{M}]_a)_{i,j}=
	\begin{cases}
		M_{i,j} & \text{if }i,j\leq a, \\
		0 & \text{otherwise}.
	\end{cases}
	$$
	Note that
	$$\var{
		\begin{pmatrix}
			\mathbf{U} \\
			\mathbf{V}
		\end{pmatrix}
	}=
	\begin{pmatrix}
		\sigma^2 \mathbf{I}_\ell & \sigma^2 [\mathbf{I}_\ell]_a \\
		\sigma^2 [\mathbf{I}_\ell]_a & \sigma^2 \mathbf{I}_\ell
	\end{pmatrix},
	$$
	which implies that
	\begin{align*}
		\var(\boldsymbol{\hat{\theta}})=\sigma^2
		\begin{pmatrix}
			(\mathbf{X}^T\mathbf{X})^{-1} & [(\mathbf{X}^T\mathbf{X})^{-1}\mathbf{X}^T]_a(\mathbf{W}(\mathbf{W}^T\mathbf{W})^{-1}) \\
			[(\mathbf{W}^T\mathbf{W})^{-1}\mathbf{W}^T]_a(\mathbf{X}(\mathbf{X}^T\mathbf{X})^{-1}) & (\mathbf{W}^T\mathbf{W})^{-1}
		\end{pmatrix}.
	\end{align*}
	Hence,
	$$\det(\var(\boldsymbol{\hat{\theta}}))=\sigma^8\det(\mathbf{A}-\mathbf{B}\mathbf{D}^{-1}\mathbf{C})\det(\mathbf{D}),$$
	where
	\begin{align*}
		\mathbf{A}&=(\mathbf{X}^T\mathbf{X})^{-1}=
		\begin{pmatrix}
			\frac{X_2}{\ell X_2-X_1^2} & -\frac{X_1}{\ell X_2-X_1^2} \\
			-\frac{X_1}{\ell X_2-X_1^2} & \frac{\ell}{\ell X_2-X_1^2}
		\end{pmatrix}, \\
		\mathbf{B}&=[(\mathbf{X}^T\mathbf{X})^{-1}\mathbf{X}^T]_a(\mathbf{W}(\mathbf{W}^T\mathbf{W})^{-1}) \\
		&=
		\begin{pmatrix}
			\frac{\sum_{i=1}^{a} (X_2-x_iX_1)(W_2-w_iW_1)}{(\ell X_2-X_1^2)(\ell W_2-W_1^2)} & \frac{\sum_{i=1}^{a} (\ell w_i-W_1)(X_2-x_iX_1)}{(\ell X_2-X_1^2)(\ell W_2-W_1^2)} \\
			\frac{\sum_{i=1}^{a} (\ell x_i-X_1)(W_2-w_iW_1)}{(\ell X_2-X_1^2)(\ell W_2-W_1^2)} & \frac{\sum_{i=1}^{a} (\ell x_i-X_1)(\ell w_i-W_1)}{(\ell X_2-X_1^2)(\ell W_2-W_1^2)}
		\end{pmatrix},\\
		\mathbf{C}&=[(\mathbf{W}^T\mathbf{W})^{-1}\mathbf{W}^T]_a(\mathbf{X}(\mathbf{X}^T\mathbf{X})^{-1}) \\
		&=
		\begin{pmatrix}
			\frac{\sum_{i=1}^{a} (X_2-x_iX_1)(W_2-w_iW_1)}{(\ell X_2-X_1^2)(\ell W_2-W_1^2)} & \frac{\sum_{i=1}^{a} (\ell x_i-X_1)(W_2-w_iW_1)}{(\ell X_2-X_1^2)(\ell W_2-W_1^2)} \\
			\frac{\sum_{i=1}^{a} (\ell w_i-W_1)(X_2-x_iX_1)}{(\ell X_2-X_1^2)(\ell W_2-W_1^2)} & \frac{\sum_{i=1}^{a} (\ell x_i-X_1)(\ell w_i-W_1)}{(\ell X_2-X_1^2)(\ell W_2-W_1^2)}
		\end{pmatrix}, \\
		\mathbf{D}&=(\mathbf{W}^T\mathbf{W})^{-1}=
		\begin{pmatrix}
			\ell & W_1 \\
			W_1 & W_2
		\end{pmatrix}^{-1}.
	\end{align*}
	After a lengthy computation, we obtain the following expression for $\det(\var(\boldsymbol{\hat{\theta}}))$:
	\begin{align*}
		&\frac{\sigma^8}{(\ell X_2-X_1^2)(\ell W_2-W_1^2)}\left(1-\frac{\left(\ell C_2-2C_1X_1+aX_2\right)\left(\ell C_2-2C_1W_1+aW_2\right)}{(\ell X_2-X_1^2)(\ell W_2-W_1^2)}\right. \\
		&\qquad\left.+\frac{\left((a-1)C_2-C_3\right)\left((a-1)C_2-C_3+\ell(X_2+W_2)-2X_1W_1\right)}{(\ell X_2-X_1^2)(\ell W_2-W_1^2)}\right).
	\end{align*}
	The expression for $H(\hat{\alpha_1},\hat{\beta_1},\hat{\alpha_2},\hat{\beta_2})$ follows by applying Proposition \ref{multivariate_normal_differential_entropy}. \qed
\end{proof}

\section{Learning with Linear Regression (LWLR)}\label{LWLRsec}
In this section, we define LWLR and reduce its hardness to LWE. As mentioned in Note \ref{note1}, our RGPC protocol allows freedom in tweaking the target rounded Gaussian distribution $\mathrm{\Psi}(0, \hat{\sigma}^2)$ by simply selecting the desired standard deviation $\sigma$. Therefore, when referring to the desired rounded Gaussian distribution for the hardness proofs, we use $\mathrm{\Psi}(0, \hat{\sigma}^2)$ (or $\mathrm{\Psi}(0, \hat{\sigma}^2) \bmod m$, i.e., $\mathrm{\Psi}_m(0, \hat{\sigma}^2))$ without divulging into the specific value of $\sigma$.

Let $\mathcal{P} = \{P_i\}_{i=1}^n$ be a set of $n$ parties.

\begin{definition}
	\emph{For modulus $m$ and a uniformly sampled $\textbf{a} \leftarrow \mathbb{Z}^w_m$, the learning with linear regression (LWLR) distribution LWLR${}_{\textbf{s},m,w}$ over $\mathbb{Z}_m^w \times \mathbb{Z}_m$ is defined as: $(\textbf{a}, x + e_{x}),$ where $x = \langle \textbf{a}, \textbf{s} \rangle$ and $e_{x} \in \mathrm{\Psi}(0, \hat{\sigma}^2)$ is a rounded Gaussian error generated by the RGPC protocol, on input $x$.}
\end{definition}

\begin{theorem}\label{LWLRthm}
	For modulus $m$, security parameter $\L$, $\ell = g(\L)$ samples (where $g$ is a superpolynomial function), PPT adversary $\mathcal{A} \notin \mathcal{P}$, some distribution over secret $\textbf{s} \in \mathbb{Z}_m^w$, and a deterministic mapping $\mathfrak{M}_h: \mathbb{Z} \to \mathrm{\Psi}(0, \hat{\sigma}^2)$ generated by the RGPC protocol, where $\mathrm{\Psi}(0, \hat{\sigma}^2)$ is the target rounded Gaussian distribution, solving decision-LWLR${}_{\textbf{s},m,w}$ is at least as hard as solving the decision-LWE${}_{\textbf{s},m,w}$ problem for the same distribution over $\textbf{s}$.
\end{theorem}
\begin{proof}
	Recall from \Cref{Lemma1} that, since $\ell$ is superpolynomial, the errors belong to the desired rounded Gaussian distribution $\mathrm{\Psi}(0, \hat{\sigma}^2)$. As given, for a fixed secret $\textbf{s} \in \mathbb{Z}^w_m$, a decision-LWLR${}_{\textbf{s},m,w}$ instance is defined as $(\textbf{a}, x + e_x)$ for $\textbf{a} \xleftarrow{\; \$ \;} \mathbb{Z}^w_m$ and $x = \langle \textbf{a},\textbf{s} \rangle$. Recall that a decision-LWE${}_{\textbf{s},m,w}$ instance is defined as $(\textbf{a}, \langle \textbf{a}, \textbf{s} \rangle + e)$ for $\textbf{a} \leftarrow \mathbb{Z}^w_m$ and $e \leftarrow \chi$ for a rounded (or discrete) Gaussian distribution $\chi$. We know from \Cref{Obs1} that $\mathfrak{M}_h$ is a deterministic mapping from random inputs $x \leftarrow \mathbb{Z}$ to errors $e_x \in \mathrm{\Psi}(0, \hat{\sigma}^2)$. We define the following two games: 
	\begin{itemize}
		\item $\mathfrak{G}_1$: in this game, we begin by fixing a secret $\textbf{s}$. Each query from the attacker is answered with an LWLR${}_{\textbf{s},m,w}$ instance as: $(\textbf{a}, x + e_x)$ for a unique $\textbf{a} \xleftarrow{\; \$ \;} \mathbb{Z}^w_m$, and $x = \langle \textbf{a},\textbf{s} \rangle$. The error $e_x \in \mathrm{\Psi}(0, \hat{\sigma}^2)$ is generated as: $e_x = \mathfrak{M}_h(x)$.
		\item $\mathfrak{G}_2$: in this game, we begin by fixing a secret $\textbf{s}$. Each query from the attacker is answered with an LWE${}_{\textbf{s},m,w}$ instance as: $(\textbf{a}, \langle \textbf{a}, \textbf{s} \rangle + e)$ for $\textbf{a} \xleftarrow{\; \$ \;} \mathbb{Z}^w_m$ and $e \leftarrow \mathrm{\Psi}(0, \tilde{\sigma}^2)$, where $\mathrm{\Psi}(0, \tilde{\sigma}^2)$ denotes a rounded Gaussian distribution that is suitable for sampling LWE errors.   
	\end{itemize}
	Let the adversary $\mathcal{A}$ be able to distinguish LWLR${}_{\textbf{s},m,w}$ from LWE${}_{\textbf{s},m,w}$ with some non-negligible advantage, i.e., $Adv_{\mathcal{A}}(\mathfrak{G}_1, \mathfrak{G}_2) \geq \varphi(w)$ for a non-negligible function $\varphi$. Hence, it follows that $Adv_{\mathcal{A}}(\mathrm{\Psi}(0, \tilde{\sigma}^2), \mathrm{\Psi}(0, \hat{\sigma}^2)) \geq \varphi(w)$. However, we have already established in \Cref{Obs1} that $\mathfrak{M}_h$ is random to $\mathcal{A} \notin \mathcal{P}$. Furthermore, we know that $\hat{\sigma}$ can be brought arbitrarily close to $\tilde{\sigma}$ (see \Cref{note1}). Therefore, for appropriate Gaussian parameters, it holds that $Adv_{\mathcal{A}}(\mathrm{\Psi}(0, \tilde{\sigma}^2), \mathrm{\Psi}(0, \hat{\sigma}^2)) \leq \eta(w)$ for a negligible function $\eta$, which directly leads to $Adv_{\mathcal{A}}(\mathfrak{G}_1, \mathfrak{G}_2) \leq \eta(w)$. Hence, for any distribution over a secret $\textbf{s} \in \mathbb{Z}_m^w$, solving decision-LWLR${}_{\textbf{s},m,w}$ is at least as hard as solving the decision-LWE${}_{\textbf{s},m,w}$ problem for the same distribution over $\textbf{s}$. \qed	  	 
\end{proof}

\section{Star-specific Key-homomorphic PRF\lowercase{s}}\label{Sec7}
In this section, we use LWLR to construct the first star-specific key-homomorphic (SSKH) PRF family. We adapt the key-homomorphic PRF construction from~\cite{Ban[14]} by replacing the deterministic errors generated from the rounding function in LWR with the errors produced via the deterministic mapping, generated by our RGPC protocol. 

\subsection{Background}
For the sake of completeness, we begin by recalling the key-homomorphic PRF construction from~\cite{Ban[14]}. Let $T$ be a full binary tree with at least one node, i.e., every non-leaf node in $T$ has two children. Let $T.r$ and $T.l$ denote its right and left subtree, respectively, and $\lfloor \cdot \rceil_p$ denote the rounding function from LWR (see \Cref{LWR} for an introduction to LWR). 

Let $q \geq 2$, $d = \lceil \log q \rceil$, and $x[i]$ denote the $i^{th}$ bit of a bit-string $x$. Define a gadget vector as:
\[\textbf{g} = (1,2,4,\dots, 2^{d-1}) \in \mathbb{Z}^d_q.\]
Define a decomposition function $\textbf{g}^{-1}: \mathbb{Z}_q \rightarrow \{0,1\}^d$ such that $\textbf{g}^{-1}(a)$ is a ``short'' vector and $ \forall a \in \mathbb{Z}_q$, it holds that: $\langle \textbf{g}, \textbf{g}^{-1}(a) \rangle = a$, where $\langle \cdot \rangle$ denotes the inner product. Function $\textbf{g}^{-1}$ is defined as:
\begin{center}
	$\textbf{g}^{-1}(a) = (x[0], x[1], \dots, x[d-1]) \in \{0,1\}^d,$
\end{center}
where $a = \sum\limits^{d-1}_{i=0} x[i] \cdot 2^i$ is the binary representation of $a$. The gadget vector is used to define the gadget matrix $\textbf{G}$ as:
\[\textbf{G} = \textbf{I}_w \otimes \textbf{g} = \text{diag}(\textbf{g}, \dots, \textbf{g}) \in \mathbb{Z}^{w \times wd}_q,\]  
where $\textbf{I}_w$ is the $w \times w$ identity matrix and $\otimes$ denotes the Kronecker product~\cite{Kath[04]}. The binary decomposition function, $\textbf{g}^{-1}$, is applied entry-wise to vectors and matrices over $\mathbb{Z}_q$. Thus, $\textbf{g}^{-1}$ can be extended to get another deterministic decomposition function $$\textbf{G}^{-1}: \mathbb{Z}^{w \times u}_q \rightarrow \{0,1\}^{wd \times u}$$ such that $\textbf{G} \cdot \textbf{G}^{-1}(\textbf{A}) = \textbf{A}$.

Given uniformly sampled matrices, $\textbf{A}_0, \textbf{A}_1 \in \mathbb{Z}^{w \times wd}_q$, define function $\textbf{A}_T(x): \{0,1\}^{|T|} \rightarrow \mathbb{Z}^{w \times wd}_q$ as:
\begin{align*}
\textbf{A}_T(x) &= 
\begin{cases}
\textbf{A}_x \qquad \qquad \qquad \qquad \qquad \qquad \; \text{if } |T| = 1, \\
\textbf{A}_{T.l}(x_l) \cdot \textbf{G}^{-1}(\textbf{A}_{T.r}(x_r)) \qquad \; \; \text{otherwise},
\end{cases}
\numberthis \label{AlignEq}
\end{align*}
where $|T|$ denotes the total number of leaves in $T$ and $x \in \{0,1\}^{|T|}$ such that $x = x_l || x_r$ for $x_l \in \{0,1\}^{|T.l|}$ and $x_r \in \{0,1\}^{|T.r|}$. The key-homomorphic PRF family is defined as:
\[\mathcal{F}_{\textbf{A}_0, \textbf{A}_1, T, p} = 
	\left\lbrace F_\textbf{s}: \{0,1\}^{|T|} \rightarrow \mathbb{Z}^{wd}_p \right\rbrace,\]
where $p \leq q$ is the modulus. A member of the function family $\mathcal{F}$ is indexed by the seed $\textbf{s} \in \mathbb{Z}^w_q$ as: \[F_{\textbf{s}}(x) = \lfloor \textbf{s} \cdot \textbf{A}_{T}(x) \rceil_p.\]

\subsection{Our Construction}
We are now ready to present the construction for the first SSKH PRF family. 

\subsubsection{Settings.}
Let $\mathcal{P} = \{P_i\}_{i=1}^n$ be a set of $n$ honest parties, that are arranged as the vertices of an interconnection graph $G = (V, E)$, which is comprised of $S_k$ stars $\partial_1, \ldots, \partial_\rho$, i.e., each subgraph $\partial_i$ is a star with $k$ leaves. As mentioned in Section~\ref{broad}, we assume that each party in $\partial_i$ is connected to $\partial_i$'s central hub $C_i$ via two channels: one Gaussian channel with the desired parameters and another error corrected channel. Each party in $\mathcal{P}$ receives parameters $\textbf{A}_0, \textbf{A}_1$, i.e., all parties are provisioned with identical parameters. Hence, physical layer communications and measurements are the exclusive source of variety and secrecy in this protocol. Since we are dealing with vectors, the data points for linear regression analysis, i.e., the messages exchanged among the parties in the stars, are of the form $\{(\mathbf{x}_i, \mathbf{y}_i)\}_{i=1}^\ell$, where $\textbf{x}_i, \textbf{y}_i \in \mathbb{Z}^{wd}.$ Consequently, the target rounded Gaussian distribution becomes $\mathrm{\Psi}^{wd}(\textbf{0}, \hat{\sigma}^2)$. Let the parties in each star exchange messages in accordance to the RGPC protocol such that messages from different central hubs $C_i, C_j~(\forall i,j \in [\rho];\ i \neq j)$ are distinguishable to the parties belonging to multiple stars.

\subsubsection{Construction.}
Without loss of generality, consider a star $\partial_i \subseteq V(G)$. Each party in $\partial_i$ follows the RGPC protocol to generate its linear regression hypothesis $h^{(\partial_i)}$. Parties in star $\partial_i$ construct a $\partial_i$-specific key-homomorphic PRF family, whose member $F_{\textbf{s}_{i}}^{(\partial_i)}(x)$, indexed by the key/seed $\textbf{s}_{i} \in \mathbb{Z}^w_m$, is defined as:
\begin{equation}
	F_{\textbf{s}_{i}}^{(\partial_i)}(x) = \textbf{s}_{i} \cdot \textbf{A}_{T}(x) + \textbf{e}^{(\partial_i)}_{\textbf{b}} \bmod m, \label{eq1}
\end{equation}
where $\textbf{A}_T(x)$ is as defined by Equation \eqref{AlignEq}, $\textbf{b} = \textbf{s}_{i} \cdot \textbf{A}_{T}(x)$, and $\textbf{e}^{(\partial_i)}_{\textbf{b}} = \mathfrak{M}_{h^{(\partial_i)}}(\textbf{b})$ denotes a rounded Gaussian error computed by the deterministic mapping $\mathfrak{M}_{h^{(\partial_i)}}$ on input $\textbf{b}$. Recall that $\mathfrak{M}_{h^{(\partial_i)}}$ is generated by our RGPC protocol from hypothesis $h^{(\partial_i)}$. The star-specific secret $\textbf{s}_{i}$ can be generated by using a reconfigurable antenna (RA) \cite{MohAli[21],inbook} at the central hub, $C_i$, and thereafter reconfiguring it to randomize the error-corrected channel between itself and the parties in $\partial_i$. Specifically, $\textbf{s}_{i}$ can be generated via the following procedure:
\begin{enumerate}
	\item After performing the RGPC protocol, each party $P_j \in \partial_i$ sends a random $r_j \in [\ell]$ to $C_i$ via the error corrected channel. $C_i$ broadcasts $r_j$ to all parties in $\partial_i$ and randomizes all error-corrected channels by reconfiguring its RA. If two parties' $r_j$ values arrive simultaneously, then $C_i$ randomly discards one of them and notifies the corresponding party to resend another random value. This ensures that the channels are re-randomized after receiving each $r_j$ value. By the end of this cycle, each party receives $k$ random values $\{r_j\}_{j=1}^k$. Let $\wp_i$ denote the set of all $r_j$ values received by the parties in $\partial_i$.
	\item Each party in $\partial_i$ computes $\bigoplus\limits_{r_j \in \wp_i} r_j = s \bmod m$. 
	\item This procedure is repeated to extract the required number of bits to generate the vector $\textbf{s}_{i}$. 
\end{enumerate} 
Since $C_i$ randomizes all its channels by simply reconfiguring its RA, no active or passive adversary can compromise all $r_j \in \wp_i$ values \cite{Aono[05],MehWall[14],YanPan[21],MohAli[21],inbook,Alan[20],PanGer[19],MPDaly[12]}. In honest settings, secrecy of the star-specific secret $\textbf{s}_{i}$, generated by the aforementioned procedure, follows directly from the following three facts:

\begin{enumerate}[label = {(\roman*)}]
	\item All parties are honest.
	\item All data points $\{x_i\}_{i=1}^\ell$ are randomly sampled integers, i.e., $x_i \xleftarrow{\$} \mathbb{Z}$.
	\item The coefficients of $f(x)$, and hence $f(x)$ itself, are random.
\end{enumerate}
In the following section, we examine the settings with active/passive and internal/external adversaries. Note that the protocol does not require the parties to share their identities. Hence, the above protocol is trivially anonymous over anonymous channels (see \cite{GeoCla[08],MattYen[09]} for surveys on anonymous communications). Since anonymity has multiple applications in cryptographic protocols \cite{Stinson[87],Phillips[92],Blundo[96],Kishi[02],Deng[07],SehrawatVipin[21],Sehrawat[17],AmosMatt[07],Daza[07],Gehr[97],Anat[15],Mida[03],OstroKush[06],DijiHua[07]}, it is a useful feature of our construction.
 
\subsection{Maximum number of SSKH PRFs and Defenses Against Various Attacks}\label{Max}
In this section, we employ our results from \Cref{Extremal} to derive the maximum number of SSKH PRFs that can be constructed by a set of $n$ parties. Recall that we use the terms star and star graph interchangeably. We know that in order to construct a SSKH PRF family, the parties are arranged in an interconnection graph $G$ wherein the --- possibly overlapping --- subsets of $\mathcal{P}$ form different star graphs, $\partial_1, \ldots, \partial_\rho$, within $G$. We assume that for all $i \in [\rho]$, it holds that: $|\partial_i| = k$. Recall from \Cref{Extremal} that we derived various bounds on the size of the following set families $\mathcal{H}$ defined over a set of $n$ elements:
\begin{enumerate}
	\item $\mathcal{H}$ is an at most $t$-intersecting $k$-uniform family,
	\item $\mathcal{H}$ is a maximally cover-free at most $t$-intersecting $k$-uniform family. 
\end{enumerate} 
We set $n$ to be the number of vertices in $G$. Hence, $k$ represents the size of each star with $t$ being equal to (or greater than) $\max\limits_{i\neq j}(|\partial_i \cap \partial_j|)$. 

In our SSKH PRF construction, no member of a star $\partial$ has any secrets that are hidden from the other members of $\partial$. Also, irrespective of their memberships, all parties are provisioned with identical set of initial parameters. The secret keys and regression models are generated via physical layer communications and collaboration. Due to these facts, the parties in our SSKH PRF construction must be either honest or semi-honest but non-colluding. We consider these factors while computing the maximum number of SSKH PRFs that can be constructed securely against various types of adversaries. For a star $\partial$, let $\mathcal{O}_{\partial}$ denote an oracle for the SSKH PRF $F^{(\partial)}_{\textbf{s}}$, i.e., on receiving input $x$, $\mathcal{O}_{\partial}$ outputs $F^{(\partial)}_{\textbf{s}}(x)$. Given oracle access to $\mathcal{O}_{\partial}$, it must holds that for a PPT adversary $\mathcal{A}$ who is allowed $\poly(\L)$ queries to $\mathcal{O}_{\partial}$, the SSKH PRF $F^{(\partial)}_{\textbf{s}}$ remains indistinguishable from a uniformly random function $U$ --- defined over the same domain and range as $F^{(\partial)}_{\textbf{s}}$. 

Let $E_{i}$ denote the set of Gaussian and error-corrected channels that are represented by the edges in star $\partial_i$.

\subsubsection{External Adversary with Oracle Access.}
In this case, the adversary can only query the oracle for the SSKH PRF, and hence the secrecy follows directly from the hardness of LWLR. Therefore, at most $t$-intersecting $k$-uniform families are sufficient for this case, i.e., we do not need the underlying set family $\mathcal{H}$ to be maximally cover-free. Moreover, $t = k-1$ suffices for this case because maximum overlap between different stars can be tolerated. Hence, it follows from \Cref{asymptotic_bound} (in \Cref{Extremal}) that the maximum number of SSKH PRFs that can be constructed is:
$$\zeta\sim\frac{n^{k}}{k!}.$$

\subsubsection{Eavesdropping Adversary with Oracle Access.}
Wyner’s wiretap model \cite{WireWy[75]} models an eavesdropper observing a degraded version of the information exchanged on the main channel through a wiretap channel. Let $\mathcal{A}$ be an eavesdropping adversary with wiretap channels, observing a subset $E'$ of Gaussian and/or error-corrected channels between parties and central hubs. Furthermore, assume that $\mathcal{A}$ has oracle access to the traget star's SSKH PRF. Without loss of generality, let $\partial_i$ be the target star. Also, note that there can be more than one target stars. Let us analyze the security with respect to this adversary.
\begin{enumerate}
	\item Secrecy of $\textbf{s}_i$: After each party $P_z \in \partial_i$ contributes to the generation of $\textbf{s}_i$ by sending a random value $r_z$ to $C_i$, which then broadcasts $r_z$ to all parties in $\partial_i$, $C_i$ randomizes all error-corrected channels by reconfiguring its RA. This means that $\mathcal{A}$ cannot compromise all $r_z$ values. Hence, it follows that no information about $\textbf{s}_i$ is leaked to $\mathcal{A}$. Furthermore, it follows from the wiretap channel model that $\mathcal{A}$ cannot get precise value of any $r_z$.
	\item Messages exchanged via the channels in $E'$: leakage of enough messages exchanged within star $\partial_i$ would allow $\mathcal{A}$ to closely approximate the deterministic mapping $\mathfrak{M}_{h^{(\partial_i)}}$. Note that for this attack to work, $\mathcal{A}$ must successfully eavesdrop on enough channels in $\partial_i$. However, since $\mathcal{A}$ is an outsider with only access to wiretap channels, it follows from \Cref{mainProp}, and known information-theoretic results about physical layer communications \cite{YonWu[18],ShiJia[21]} and (Gaussian) wiretap eavesdropping \cite{WireWy[75],YanHell[78],MegJaya[22],Wire2[84],NafYen[18],NafYen[19],GoldCuff[16],KongKadd[21],JaoRod[06]} that $\mathcal{A}$ cannot gain any non-negligible information on $\mathfrak{M}_{h^{(\partial_i)}}$.
\end{enumerate}

Hence, an eavesdropping adversary with wiretap channels and oracle access has no non-negligible advantage over an external adversary with oracle access. Therefore, it follows that the maximum number of SSKH PRFs that can be constructed in this case is:
$$\zeta\sim\frac{n^{k}}{k!}.$$

\subsubsection{Non-colluding Semi-honest Parties.}
Let $\mathcal{P}_{\partial_i} \subseteq \mathcal{P}$ denote the set of parties that form the star $\partial_i$ --- in addition to the central hub $C_i$. Suppose that some or all parties in $\mathcal{P}$ are semi-honest, i.e., they follow the protocol correctly but try to gain/infer more information than what is allowed by the protocol. Further suppose that the parties do not collude with each other. However, this assumption is redundant if all malicious parties belong to same set of stars because broadcast in our protocol provides malicious parties with all the information exchanged in the stars they belong to. The restriction of non-collusion is required when at least one pair of malicious parties are members of different sets of stars. In such settings, the only way any party $P_j \notin \partial_i$ can gain additional information about the SSKH PRF $F^{(\partial_i)}_{\textbf{s}_{i}}$ is to (mis)use its membership of other stars. For instance, if $P_j \in \mathcal{P}_{\partial_d}, \mathcal{P}_{\partial_j}, \mathcal{P}_{\partial_o}$ and $\mathcal{P}_{\partial_i} \subset \mathcal{P}_{\partial_o} \cup \mathcal{P}_{\partial_j} \cup \mathcal{P}_{\partial_d}$, then because the parties send identical messages to all central hubs they are connected to, it follows that $H(F^{(\partial_i)}_{\textbf{s}_{i}} | P_j) = 0$. This follows trivially because $P_j$ can compute $\textbf{s}_i$. Having maximally cover-free families eliminates this vulnerability against non-colluding semi-honest parties. This holds because with members of a maximally cover-free family denoting unique stars in $G$, the following can never hold true for any $\mathcal{P}_{\partial_i}$:
\[\mathcal{P}_{\partial_i} \subseteq \bigcup_{j \in \varrho} \mathcal{P}_{\partial_j},\]
where $\varrho \subseteq [\rho] - \{i\}$. Hence, it follows from \Cref{asymptotic_bound_for_maximally_cover_free} that the maximum number of SSKH PRFs that can be constructed with non-colluding semi-honest parties is at least $Cn$ for some positive real number $C < 1$.

Thus, in order to construct SSKH PRFs that are secure against all types of adversaries and threat models discussed in this section, the underlying family of sets must be maximally cover-free, at most $(k-1)$-intersecting and $k$-uniform.

\subsubsection{Man-in-the-Middle.}

Physical-layer-based key generation schemes exploit the channel reciprocity for secret key extraction, which can achieve information-theoretic secrecy against eavesdroppers. However, these schemes have been shown to be vulnerable against man-in-the-middle (MITM) attacks. During a typical MITM attack, the adversary creates separate connection(s) with the communicating node(s) and relays altered transmission packets to them. Eberz et al. \cite{EbeMatt[12]} demonstrated a practical MITM attack against RSS-based key generation protocols \cite{JanaSri[09],SuWade[08]}, wherein the MITM adversary $\mathcal{A}$ exploits the same channel characteristics as the target/communicating parties $P_1, P_2$. To summarize, in the attack from \cite{EbeMatt[12]}, $\mathcal{A}$ injects packets that cause a similar channel measurement at both $P_1$ and $P_2$. This attack enables $\mathcal{A}$ to recover up to 47\% of the secret bits generated by $P_1$ and $P_2$. 

To defend against such attacks, we can apply techniques that allow us to detect an MITM attack over physical layer \cite{LeAle[16]}, and if one is detected, the antenna of $\partial_i$'s central hub, $C_i$, can be reconfigured to randomize all channels in $\partial_i$ \cite{YanPan[21]}. This only requires a reconfigurable antenna (RA) at each central hub. An RA can swiftly reconfigure its radiation pattern, polarization, and frequency by rearranging its antenna currents \cite{MohAli[21],inbook}. It has been shown that due to multipath resulting from having an RA, even small variation by the RA can create large variations in the channel, effectively creating fast varying channels with a random distribution \cite{JunPhD[21]}. One way an RA may randomize the channels is by randomly selecting antenna configurations in the transmitting array at the symbol rate, leading to a random phase and amplitude multiplied to the transmitted symbol. The resulting randomness is compensated by appropriate element weights so that the intended receiver does not experience any random variations. In this manner, an RA can be used to re-randomize the channel and hence break the temporal correlation of the channels between $\mathcal{A}$ and the attacked parties, while preserving the reciprocity of the other channels. 

Therefore, even if an adversary $\mathcal{A}$ is able to perform successful injection in communication round {\ss}, its channels with the attacked parties would be randomly modified (by the RA) when it attempts injections in round {\ss}+1. On the other hand, the channels between the parties in star $\partial_i$ and the central hub $C_i$ remain reciprocal, i.e., they can still make correct/identical measurements. Hence, by reconfiguring $C_i$'s RA, we can prevent further injections from $\mathcal{A}$ without affecting the legitimate parties' ability to make correct channel measurements. Further details on this defense technique are beyond the scope of this paper. For detailed introduction to the topic and its applications in different settings, we refer the interested reader to \cite{Aono[05],MehWall[14],YanPan[21],MohAli[21],inbook,Alan[20],PanGer[19],MPDaly[12]}. In this manner, channel state randomization can be used to effectively reduce an MITM attack to the less harmful jamming attack \cite{MitChor[21]}. See \cite{HossHua[21]} for a thorough introduction to jamming and anti-jamming techniques.   

\subsection{Runtime and Key Size}
We know that the complexity of a single evaluation of the key-homomorphic PRF from~\cite{Ban[14]} is $\mathrm{\Theta}(|T| w^\omega \log^2 m)$ ring operations in $\mathbb{Z}_m$, where $\omega \in [2, 2.37286]$ is the exponent of matrix multiplication~\cite{Josh[21],CohBla[23]}. Using the fast integer multiplication algorithm from \cite{HarveyHoeven[21]}, this gives a time complexity of $\mathrm{\Theta}(|T| w^\omega m\log^3 m)$. The time taken by the setup of our SSKH PRF construction is equal to the time required by our RGPC algorithm to find the optimal hypothesis, which we know from Section~\ref{Time} to be $\mathrm{\Theta}(\ell)$ additions and multiplications. If $B$ is an upper bound on $x_i$ and $y_i$, then the time complexity is $O\left(\ell B\log B\right)$. Once the optimal hypothesis is known, it takes $\mathrm{\Theta}(wm\log^2 m)$ time to generate a deterministic LWLR error for a single input. Hence, after the initial setup, the time complexity of a single function evaluation of our SSKH PRF remains $\mathrm{\Theta}(|T| w^\omega m\log^3 m)$. 

Similarly, the key size for our SSKH PRF family is the same as that of the key-homomorphic PRF family from~\cite{Ban[14]}. Specifically, for security parameter $\L$ and $2^{\L}$ security against the well known lattice reduction algorithms \cite{Ajtai[01],Fincke[85],Gama[06],Gama[13],Gama[10],LLL[82],Ngu[10],Vid[08],DanPan[10],DaniPan[10],Poha[81],Schnorr[87],Schnorr[94],Schnorr[95],Nguyen[09],EamFer[21],RicPei[11],MarCar[15],AvrAda[03],GuaJoh[21],SatoMasa[21],TamaStep[20],AleQi[21],EriWag[22]}, the key size for our SSKH PRF family is $\L$.

\subsection{Correctness and Security}
Recall that LWR employs rounding to hide all but some of the most significant bits of $\lfloor \textbf{s} \cdot \textbf{A} \rceil_p$; therefore, the rounded-off bits become the deterministic error. On the other hand, our solution, i.e., LWLR, uses linear regression hypothesis to generate the desired rounded Gaussian errors, which are derived from the (independent) errors occurring in the physical layer communications over Gaussian channel(s). For the sake of simplicity, the proofs assume honest parties in the absence of any adversary. For other supported cases, it is easy to adapt the statements of the results according to the bounds/conditions established in \Cref{Max}. 

Recall that the RGPC protocol ensures that all parties in a star $\partial$ receive an identical dataset $\mathcal{D}$, and therefore arrive at the same linear regression hypothesis $h^{(\partial)}$ and errors $\textbf{e}^{(\partial)}_{\textbf{b}}$.

\begin{theorem}\label{thm1}
	The function family defined by \Cref{eq1} is a star-specific key-homomorphic PRF under the decision-LWE assumption.
\end{theorem}
\begin{proof}
	We know from \Cref{LWLRthm} that for $\textbf{s}_i \xleftarrow{\$} \mathbb{Z}_m^w$ and a superpolynomial number of samples $\ell$, the LWLR instances generated in \Cref{eq1} are as hard as LWE --- to solve for $\textbf{s}_i$ (and $\textbf{e}_\textbf{b}^{(\partial_i)})$. The randomness of the function family follows directly from the randomness of $\textbf{s}_i$. The deterministic behavior follows from the above observation and the fact that $\textbf{A}_T(x)$ is a deterministic function. Hence, the family of functions defined in \Cref{eq1} is a PRF family. 

	Define
		$$G_{\mathbf{s}}^{(\partial_i)}(x) = \textbf{s} \cdot \mathbf{A}_{T}(x) + \lfloor\boldsymbol{\varepsilon}^{(\partial_i)}_{\textbf{b}}\rceil \bmod m,$$
		where $\boldsymbol{\varepsilon}^{(\partial_i)}_{\textbf{b}}$ is the (raw) Gaussian error corresponding to $\mathbf{b}$ for star $\partial_i$; define $G_{\mathbf{s}}^{(\partial_j)}$ similarly. Since the errors $\boldsymbol{\varepsilon}^{(\partial_i)}_{\textbf{b}}$ and $\boldsymbol{\varepsilon}^{(\partial_j)}_{\textbf{b}}$ are independent Gaussian random variables, each with variance $\sigma^2$, it holds that:
	\begin{align*}
		\Pr[G_{\mathbf{s}}^{(\partial_i)}(x)=G_{\mathbf{s}}^{(\partial_j)}(x)]&=\Pr[\lfloor\boldsymbol{\varepsilon}^{(\partial_i)}_{\textbf{b}}\rceil=\lfloor\boldsymbol{\varepsilon}^{(\partial_j)}_{\textbf{b}}\rceil] \\
		&\leq \Pr[||\boldsymbol{\varepsilon}^{(\partial_i)}_{\textbf{b}}-\boldsymbol{\varepsilon}^{(\partial_j)}_{\textbf{b}}||_\infty<1] \\
		&= \Pr[|Z|<(\sqrt{2}\sigma)^{-1}]^w
	\end{align*}
	where $Z$ is a standard Gaussian random variable.
	
	Furthermore, since the number of samples is superpolynomial in the security parameter $\L$, by drowning/smudging, the statistical distance between $\boldsymbol{e}^{(\partial_i)}_{\textbf{b}}$ and $\boldsymbol{\varepsilon}^{(\partial_i)}_{\textbf{b}}$ is negligible (similarly for $\boldsymbol{\varepsilon}^{(\partial_j)}_{\textbf{b}}$ and $\boldsymbol{e}^{(\partial_j)}_{\textbf{b}}$). Hence,
	\begin{align*}
		\Pr[F_{\mathbf{s}}^{(\partial_i)}(x)=F_{\mathbf{s}}^{(\partial_j)}(x)]&=\Pr[G_{\mathbf{s}}^{(\partial_i)}(x)=G_{\mathbf{s}}^{(\partial_j)}(x)]+\eta(\L) \\
&\leq\Pr[|Z|<(\sqrt{2}\sigma)^{-1}]^w+\eta(\L),
	\end{align*}
	where $\eta(\L)$ is a negligible function in $\L$. By choosing $\delta=\Pr[|Z|<(\sqrt{2}\sigma)^{-1}]$, this function family satisfies \Cref{MainDef}(i).
	
	Finally, by Chebyshev's inequality and the union bound, for any $\tau>0$,
	\[F_{\textbf{s}_1}^{(\partial)}(x) + F_{\textbf{s}_2}^{(\partial)}(x)  = F_{\textbf{s}_1 + \textbf{s}_2}^{(\partial)}(x) + \textbf{e}' \bmod m,\] 
	where each entry of $\textbf{e}'$ lies in $[-3\tau\hat{\sigma}, 3\tau\hat{\sigma}]$ with probability at least $1-3/\tau^2$. For example, choosing $\tau=\sqrt{300}$ gives us the bound that the absolute value of each entry is bounded by $\sqrt{2700}\hat{\sigma}$ with probability at least $0.99$.

	Therefore, the function family defined by \Cref{eq1} is a SSKH PRF family --- as defined by \Cref{MainDef} --- under the decision-LWE assumption. \qed
\end{proof}

\section{Conclusion and Future Work}\label{Sec8}
In this paper, we introduced a novel derandomized variant of the celebrated learning with errors (LWE) problem, called learning with linear regression (LWLR), which derandomizes LWE via deterministic, yet sufficiently independent, errors that are generated by using special linear regression models whose training data consists of physical layer communications over Gaussian channels. Prior to our work, learning with rounding and its variant nearby learning with lattice rounding were the only known derandomized variant of the LWE problem; both of which relied on rounding. LWLR relies on the naturally occurring errors in physical layer communications to derandomize LWE while maintaining its hardness --- for specific parameters.  

We also introduced star-specific key-homomorphic (SSKH) pseudorandom functions (PRFs), which are directly defined by the physical layer communications among the respective sets of parties that construct them. We used LWLR to construct the first SSKH PRF family. In order to quantify the maximum number of SSKH PRFs that can be constructed by sets of overlapping parties, we derived: 
\begin{itemize}
	\item a formula to compute the mutual information between linear regression models that are generated from overlapping training datasets,
	\item bounds on the size of at most $t$-intersecting $k$-uniform families of sets. We also gave an explicit construction to build such set systems,
	\item bounds on the size of maximally cover-free at most $t$-intersecting $k$-uniform families of sets.  
\end{itemize}
Using these results, we established the maximum number of SSKH PRFs that can be constructed by a given set of parties in the presence of active/passive and internal/external adversaries. 

It would be interesting to investigate whether our RGPC protocol can be adapted to sample deterministic yet sufficiently independent Bernoulli noise. Realizing this could allow secure derandomization of learning parity with noise (LPN), which can be viewed as a restricted case of LWE, that asks to recover a secret vector given a system of noisy linear equations over $\mathbb{F}_2$, where the noise follows a Bernoulli distribution \cite{Kry[12]}. LPN-based cryptosystems tend to be simple and efficient in terms of code-size, and space and computational overheads. This makes them a prime candidate for realizing post-quantum security in resource constrained environments, such as RFID tags and IoT devices. A derandomized variant of LPN could make it more versatile and expand its application space.

\bibliographystyle{plainurl}
\bibliography{ref}
\addcontentsline{toc}{section}{Bibliography}

\end{document}